\documentclass[reprint,amsfonts, amssymb, amsmath,  showkeys, nofootinbib,pra, superscriptaddress, twocolumn,longbibliography]{revtex4-2}
\usepackage{float}
\makeatletter
\let\newfloat\newfloat@ltx
\makeatother
\usepackage[english]{babel}
\usepackage[utf8]{inputenc}
\usepackage{graphics}
\usepackage{selinput}
\usepackage{enumerate}

\usepackage{braket}
\usepackage{amsthm}
\usepackage{mathtools}
\usepackage{physics}
\usepackage{graphicx}
\usepackage[left=16mm,right=16mm,top=35mm,columnsep=15pt]{geometry} 
\usepackage{adjustbox}
\usepackage{placeins}
\usepackage[T1]{fontenc}
\usepackage{lipsum}
\usepackage{csquotes}
\usepackage{mathbbol}  
\usepackage[linesnumbered,ruled,vlined]{algorithm2e}
\SetKwInput{kwInit}{Init}
\usepackage{xcolor} 
\usepackage{tabularray}
\definecolor{WildSand}{rgb}{0.968,0.968,0.968}     
\definecolor{NewYorkPink}{rgb}{0.918,0.604,0.682}
\definecolor{Feijoa}{rgb}{0.459, 0.765, 0.541}          
\definecolor{Glacier}{rgb}{0.682,0.843,0.918}

\def\LC{\mathcal{L}}

\def\a{\alpha}

\usepackage[makeroom]{cancel}
\usepackage[colorlinks=true,citecolor=blue,linkcolor=magenta]{hyperref}

\usepackage{tikz}
\tikzset{every picture/.style=remember picture}

\newcommand{\poly}{\operatorname{poly}}

\newcommand{\Ebb}{\mathbb{E}}

\newcommand{\DC}{\mathcal{D}}

\newcommand{\OC}{\mathcal{O}}
\newcommand{\PC}{\mathcal{P}}

\newcommand{\VC}{\mathcal{V}}

\newcommand{\Var}{{\rm Var}}
\newcommand{\Cov}{{\rm Cov}}

\renewcommand{\geq}{\geqslant}
\renewcommand{\leq}{\leqslant}

\DeclareMathOperator*{\argmax}{arg\,max}
\DeclareMathOperator*{\argmin}{arg\,min}
\renewcommand{\vec}[1]{\boldsymbol{#1}}

\newcommand{\bs}{\textsf{BS}}

\newcommand{\al}{\alpha }

\renewcommand{\th}{\theta }

\newcommand{\sinc}{\mathrm{sinc} }
\newcommand{\av}{\vec{a} }
\newcommand{\bv}{\vec{b} }
\newcommand{\zv}{\vec{z} }
\newcommand{\xv}{\vec{x} }
\newcommand{\yv}{\vec{y} }
\newcommand{\thv}{\vec{\theta}}

\def\be{\begin{equation}}
\def\ee{\end{equation}}
\def\bs{\begin{split}}
\def\e{\end{split}}
\def\ba{\begin{eqnarray}}
\def\bea{\begin{eqnarray}}

\def\tea{\end{eqnarray}}
\def\ea{\end{eqnarray}}
\def\eea{\end{eqnarray}}

\def\a{\alpha}

\def\a{\alpha}

\def\a{\alpha}

\newtheorem{theorem}{Theorem}
\newtheorem{theoremappendix}{Theorem}
\newtheorem{lemma}{Lemma}
\newtheorem{corollary}{Corollary}

\newtheorem{proposition}{Proposition}
\newtheorem*{proposition*}{Proposition}

\newtheorem{assumption}{Assumption}

\usepackage{amssymb}
\usepackage{dsfont}

\usepackage[normalem]{ulem}
\usepackage{hyperref}

\usepackage{ulem}

\usepackage{minitoc}
\usepackage{tocloft}

\makeatletter
\renewcommand{\part}[1]{
  \refstepcounter{part}
  \addcontentsline{toc}{part}{#1}
  \mtcaddpart[part]{#1}
}
\makeatother

\makeatletter
\renewcommand\onecolumngrid{
\do@columngrid{one}{\@ne}
\def\set@footnotewidth{\onecolumngrid}
\def\footnoterule{\kern-6pt\hrule width 1.5in\kern6pt}
}
\renewcommand\twocolumngrid{
        \def\footnoterule{
        \dimen@\skip\footins\divide\dimen@\thr@@
        \kern-\dimen@\hrule width.5in\kern\dimen@}
        \do@columngrid{mlt}{\tw@}
}
\makeatother

\newcommand{\nparams}{m}

\begin{document}
\doparttoc 
\faketableofcontents 
\part{}

\title{IQP Born Machines under Data-dependent and Agnostic Initialization Strategies}

\author{Sacha Lerch}
\affiliation{Institute of Physics, Ecole Polytechnique F\'{e}d\'{e}rale de Lausanne (EPFL), CH-1015 Lausanne, Switzerland}
\affiliation{Centre for Quantum Science and Engineering, \'{E}cole Polytechnique F\'{e}d\'{e}rale de Lausanne (EPFL), Lausanne, Switzerland}

\author{Joseph Bowles}
\affiliation{Xanadu, Toronto, ON, M5G 2C8, Canada}

\author{Ricard Puig}
\affiliation{Institute of Physics, Ecole Polytechnique F\'{e}d\'{e}rale de Lausanne (EPFL), CH-1015 Lausanne, Switzerland}
\affiliation{Centre for Quantum Science and Engineering, \'{E}cole Polytechnique F\'{e}d\'{e}rale de Lausanne (EPFL), Lausanne, Switzerland}

\author{Erik Armengol}
\affiliation{ICFO-Institut de Ciències Fotòniques, The Barcelona Institute of Science and Technology, 08860 Castelldefels (Barcelona), Spain}
\affiliation{Eurecat, Centre Tecnològic de Catalunya, Multimedia Technologies, 08005 Barcelona, Spain}

\author{Zo\"{e} Holmes}
\affiliation{Institute of Physics, Ecole Polytechnique F\'{e}d\'{e}rale de Lausanne (EPFL), CH-1015 Lausanne, Switzerland}
\affiliation{Centre for Quantum Science and Engineering, \'{E}cole Polytechnique F\'{e}d\'{e}rale de Lausanne (EPFL), Lausanne, Switzerland}

\author{Supanut Thanasilp}
\affiliation{Chula Intelligent and Complex Systems, Department of Physics, Faculty of Science, Chulalongkorn University, Bangkok, Thailand, 10330}

\date{\today}

\begin{abstract}
Quantum circuit Born machines based on instantaneous quantum polynomial-time (IQP) circuits are natural candidates for quantum generative modeling, both because of their probabilistic structure and because IQP sampling is provably classically hard in certain regimes. Recent proposals focus on training IQP-QCBMs using Maximum Mean Discrepancy (MMD) losses built from low-body Pauli-$Z$ correlators, but the effect of initialization on the resulting optimization landscape remains poorly understood.
In this work, we address this by first proving that the MMD loss landscape suffers from barren plateaus for random full-angle-range initializations of IQP circuits. We then establish lower bounds on the loss variance for identity and an unbiased data-agnostic initialization. We then additionally consider a data-dependent initialization that is better aligned with the target distribution and, under suitable assumptions, yields provable gradients and generally converges quicker to a good minimum (as indicated by our training of circuits with 150 qubits on genomic data). Finally, as a by-product, the developed variance lower bound framework is applicable to a general class of non-linear losses, offering a broader toolset for analyzing warm-starts in quantum machine learning.
\end{abstract}

\maketitle

\section{Introduction}

Quantum generative models are viewed as a promising route to quantum advantage because they can offer intrinsically efficient sampling and can potentially efficiently represent complex distributions~\cite{sweke2021quantum, gao2022enhancing,huang2025generative,hibat2024framework,coyle2020born}. A standard framework is the Quantum Circuit Born machine (QCBM)~\cite{coyle2020born,benedetti2019generative, liu2018differentiable, rudolph2023trainability}, in which a parametrized quantum state is measured in the computational basis to produce samples from a model distribution. More recently, attention has shifted toward “train-classically-deploy-quantumly” generative models, where training is performed entirely classically and then the resulting model is executed on quantum hardware to generate samples~\cite{cerezo2023does, bako2025fermionic,huang2025generative, recio2025train,herrero2025born, kasture2023protocols, kurkin2025universality}.

However, a key obstacle for quantum generative models, as for variational quantum algorithms more broadly, is trainability~\cite{larocca2024review, anschuetz2022quantum, abbas2023quantum,rudolph2023trainability,shen2026characterizing,mcclean2018barren,herbst2025limits}. It is now well established that variational quantum models can exhibit exponential concentration phenomena in which optimization signals decay exponentially with system size under random initialization~\cite{mcclean2018barren, larocca2024review}. In such cases the loss landscape can be visualized as a so-called `barren plateau' where exponential precision is typically required to successfully train.
In quantum generative modeling, “explicit” losses such as the KL divergence are known to be especially susceptible to exponential concentration~\cite{rudolph2023trainability, liu2018differentiable}. By contrast, the Maximum Mean Discrepancy (MMD) loss, when used with an appropriate bandwidth, can be expressed in terms of low-body Pauli-$Z$ correlators and therefore provably avoids exponential concentration for shallow hardware-efficient ans\"atze~\cite{rudolph2023trainability}. Thus, much of the effort on training QCBMs has shifted toward MMD-based losses~\cite{recio2025train, herrero2025born, kasture2023protocols, kurkin2025universality, zou2025conquer,majumder2024variational,bako2025fermionic,herbst2025limits}.

Instantaneous Quantum Polynomial-time (IQP) based QCBMs have recently attracted attention as natural candidates for quantum generative learning because they offer a convenient combination of features. Firstly, the key quantities needed for MMD-based training (notably low-body Pauli-$Z$ correlators) can be evaluated efficiently on a classical computer~\cite{nest2009simulating,recio2025train}. In parallel, the corresponding sampling task is nevertheless believed to be classically hard under standard complexity-theoretic assumptions~\cite{shepherd2009temporally,bremner2011classical,bremner2016average, hangleiter2023computational}. This makes IQP-QCBMs a convenient testbed for probing trainability and optimization behavior with classical resources today, while still offering the potential for a quantum advantage in the longer term once hardware improves~\cite{recio2025train, huang2025generative,herrero2025born, kasture2023protocols, kurkin2025universality}.

Motivated by this paradigm, recent empirical studies suggest that IQP-QCBMs exponentially concentrate under full-range uniform initialization, yet can be trained at surprisingly large scales when equipped with data-dependent initialization strategies~\cite{recio2025train}. At present, however, these numerical observations lack a clear theoretical explanation.
Here we address this gap by analyzing the loss-landscape variance for three families of initialization: (i) uniformly random initializations over the full angle range, (ii) small-angle initializations in data-agnostic regions, and (iii) data-dependent small-angle initialization strategies. In particular, we rigorously show the emergence of barren plateaus over the whole loss landscape with all-to-all IQP topology, excluding the practicality of the full angle range initialization. For small-angle range initialization strategies, we explicitly connect the loss variance around a small patch to the curvature of the loss landscape at the initialization center. Consequently, we show that agnostic initialization strategies and data-dependent strategy under reasonable physical assumptions of a target distribution results in patches with substantial gradients, exhibiting initial trainability guarantees, which we confirm for small-scale circuits via numerical simulations in the PennyLane software platform  \cite{bergholm2018pennylane}. More broadly, these variance lower bounds apply to a general class of non-linear losses, providing a versatile toolset for analyzing warm-starts across quantum machine learning.

Going beyond the variance analysis, we provide empirical evidence with large scale $150$ qubit simulation (on genomic data via the IQPopt package \cite{recio2025iqpopt, ErikRecio2025}) that data-dependent initialization strategies consistently converge faster than an agnostic unbiased strategy where the initial guess corresponds to the uniform distribution over all possible bit strings, and have better performance over the (popular) identity initialization strategy. This analysis clarifies when and why data-dependent warm starts can outperform data-agnostic ones. A summary of our results is provided in Fig.~\ref{fig:main_fig} and  Section~\ref{sec:overview}.

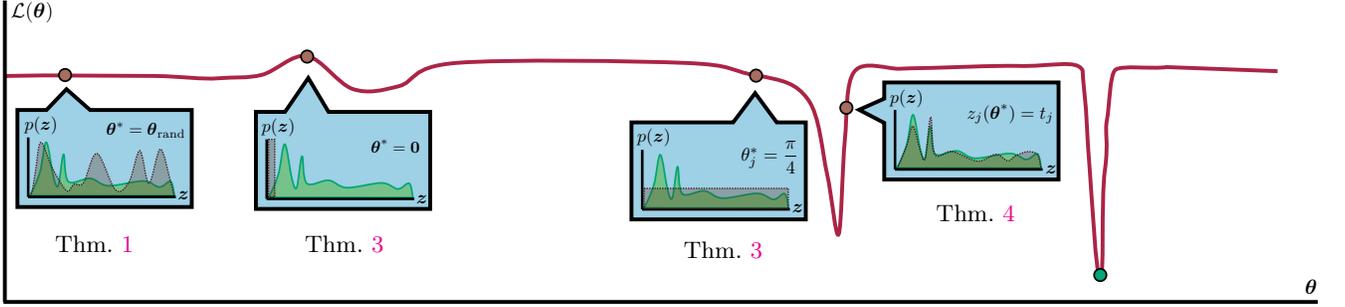
\begin{figure*}[ht]
    \centering
    \begin{tikzpicture}
        \node[anchor=south west, inner sep=0] (image) at (0,0) {
           \includegraphics[width=0.95\linewidth]{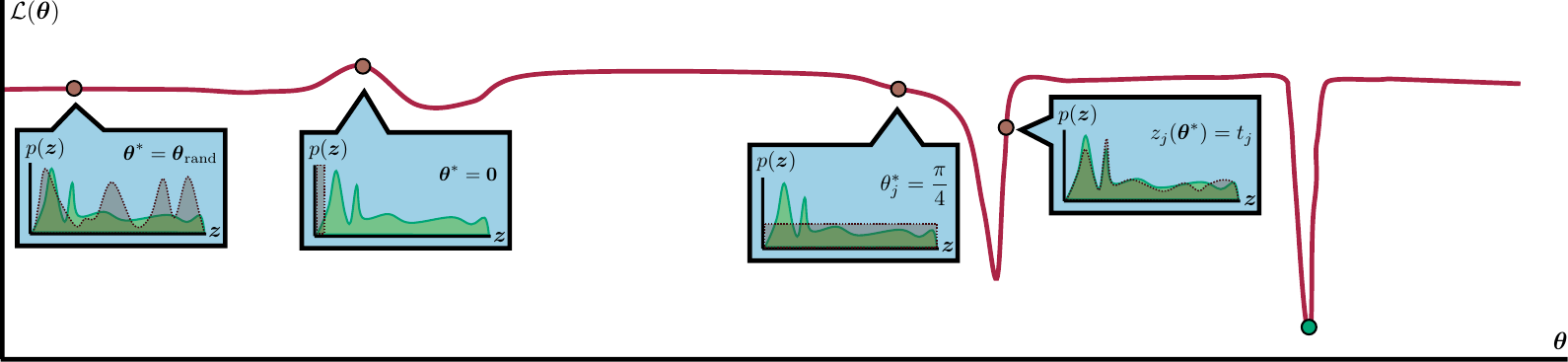}
        };

        \begin{scope}[x={(image.south east)}, y={(image.north west)}]
            \node at (0.07, 0.2) {Thm.~\ref{thm:mmd-loss-concentration-random-alltoall}};
            \node at (0.26, 0.2) {Thm.~\ref{thm:variance-guarantee-agnostic-initialisation-main}};
            \node at (0.548, 0.18) {Thm.~\ref{thm:variance-guarantee-agnostic-initialisation-main}};
            \node at (0.74, 0.3) {Thm.~\ref{thm:data-dependent-local-trainability-main}};
        \end{scope}
    \end{tikzpicture}
    \caption{\textbf{Illustration of our main results.} Schematic MMD loss landscape $\mathcal{L}(\boldsymbol{\theta})$ for an IQP generative model. The insets show the corresponding computational-basis distributions at representative initialization points: the model distribution $p_{\thv}(\zv)$ is shown in gray, and the target distribution $p_{\rm data}(\zv)$ in green. Identity initialization ($\thv^*=\vec{0}$) yields a highly biased distribution, $p_{\thv^*}(\vec{0})=1$, which corresponds to a local maximum of the MMD. An unbiased data-agnostic initialization ($\theta_j^{*}=\pi/4$) yields a uniform distribution over bit strings and usually corresponds to a region with gradient. A data-dependent initialization yields an initial distribution that is better aligned with the target. The green point indicates schematically the minimum corresponding to $p_{\rm data}(\zv)$. }
    \label{fig:main_fig}
\end{figure*}

\section{Framework}

\subsection{Quantum generative modeling}

Generative modeling aims to train a model using samples drawn from an unknown target distribution, such that the model distribution closely approximates the target and enables direct sampling from the learned distribution~\cite{vaswani2017attention,ho2020denoising,kingma2013auto,goodfellow2020generative}. 
More precisely, consider the target distribution $\PC_{\rm data} = \{ p_{\mathrm{data}}(\zv) \}_{\zv}$ over bit strings $\zv\in\{0,1\}^n$ where $p_{\rm data}(\zv)$ is a target probability associated with the bit string $\zv$, and similarly a model distribution $\PC_{\thv} = \{p_{\thv}(\zv) \}_{\zv}$ parametrized by a set of trainable parameters $\thv$. The goal is to identify optimal parameters $\thv^*$ such that $\PC_{\thv^*} \approx  \PC_{\rm data}$ thereby enabling the generation of new samples according to $\zv_{\rm new} \sim \PC_{\thv^*}$.

A quantum state acts as a natural sampler: performing a measurement in the computational basis yields a bit string drawn from the probability distribution defined by the Born rule. Among the various quantum generative models~\cite{zoufal2019quantum,dallaire2018quantum,Amin2018Quantum, coopmans2023sample,zoufal2021variational,demidik2025expressive,chang2024latent,barthe2025parameterized,romero2021variational, martinez2025quantum,wu2025multidimensional,majumder2024variational}, a Quantum Circuit Born Machine (QCBM) leverages this property, functioning as a generative model by parameterizing the outcome distribution through a quantum state~\cite{coyle2020born,benedetti2019generative, liu2018differentiable, rudolph2023trainability}. Specifically, we consider a parameterized state of the form
\begin{align}
    |\psi(\thv)\rangle = U(\thv) |\psi_0\rangle \;,
\end{align}
where $U(\thv)$ is a parametrized quantum circuit and  $|\psi_0\rangle$ is some initial state.
The resulting model distribution, obtained through measurement in the computational basis, is defined as
\begin{equation}\label{eq:def-model-dist}
    p_{\thv}(\zv)=|\langle \zv | \psi(\thv)\rangle|^2 \;\;.
\end{equation}

Fundamental properties of the model distribution, such as its expressivity and scalability, depend significantly on the architectural design of $U(\boldsymbol{\theta})$. In this work, we focus on parameterized Instantaneous Quantum Polynomial-time Quantum Circuits Born Machines (IQP-QCBMs)~\cite{recio2025train, huang2025generative,herrero2025born, kasture2023protocols, kurkin2025universality}. This architecture is motivated by proposals for demonstrating a quantum advantage in sampling tasks, where certain instances of IQP circuits produce distributions that are computationally intractable to sample from using classical resources~\cite{shepherd2009temporally,bremner2011classical,bremner2016average, hangleiter2023computational}.

An IQP-QCBM is constructed by applying an initial layer of Hadamard gates to the all-zero state $|0\rangle^{\otimes n}$, followed by a sequence of diagonal one- and two-body Pauli-$Z$ rotations, and concluding with a final layer of Hadamard gates. In the $X$-basis, this circuit architecture is equivalent to the unitary operator
\begin{align}\label{eq:IQP-circuit}
    U(\thv) = \exp\left(i\sum_{j=1}^n \theta_j X_j + i\sum_{(j,k)\in E}\theta_{jk}X_jX_k\right) \;\;,
\end{align}
where $E \subseteq \{(j,k) : j < k\}$ denotes the set of edges defining the interaction graph (i.e., the qubit topology). Notably, since all generators in the exponent commute, the unitary $U(\thv)$ is invariant under the ordering of the individual gate operations. Here we focus primarily on the case of all-to-all two-qubit connectivity.

\subsection{Maximum Mean Discrepancy loss}

Training the quantum generative model is performed by minimizing a loss function $\LC(\thv)$ that quantifies the distance between the model and target distributions to find
\begin{align}
    \thv^* = \argmin_{\thv} \LC (\thv) \;\;.
\end{align}
Since the probability distributions $p_{\mathrm{data}}(\mathbf{z})$ and $p_{\boldsymbol{\theta}}(\mathbf{z})$ can be represented as vectors in a high-dimensional space, various metrics can be employed to quantify their discrepancy. An \textit{explicit} loss function directly compares individual probabilities across the distributions' support. Specifically, for each bit string $\zv$, the loss evaluates the proximity of $p_{\boldsymbol{\theta}}(\zv)$ to $p_{\mathrm{data}}(\zv)$. Many standard divergence measures can be fundamentally characterized as \textit{explicit} losses~\cite{rudolph2023trainability}, including Kullback-Leibler Divergence~\cite{kullback1951on}, Total Variation Distance, and Jensen-Shannon Divergence~\cite{lin1991divergence}.

In contrast, the Maximum Mean Discrepancy (MMD) loss~\cite{gretton2012kernel}, is an \textit{implicit} loss. Namely, rather than requiring explicit access to the probability distributions, it implicitly compares them by matching the \textit{correlators} of the two distributions~\cite{rudolph2023trainability,recio2025train}. If all correlators (or moments) of the distributions align perfectly, the two distributions are identical. More precisely, consider a subset of qubits $A \subseteq [n]$ with $[n] = \{1,2,...,n\}$; a model correlator associated with this subset can be expressed as the expectation value of the multi-qubit Pauli-Z operator 
\begin{align}\label{eq:def-correlator}
    z_A(\thv) = \langle \psi(\thv) | Z_A|\psi(\thv)\rangle \;,
\end{align}
where the Pauli-Z operator on the subset is defined as
\begin{equation}
   Z_A := \bigotimes_{j \in A} Z_j \;\;.
\end{equation}
Similarly, the target correlators associated with the distribution $p_{\mathrm{data}}(\zv)$ are defined as
\begin{align}\label{eq:target-correlators-definition}
    t_A&:=\Ebb_{\yv\sim p_{\mathrm{data}}}[\bra{\yv}Z_A\ket{\yv}] \\
    &=\Ebb_{\yv\sim p_{\mathrm{data}}}[(-1)^{\sum_{j\in A}y_j}] \;\;,
\end{align}
The MMD loss is then expressed as a weighted sum of the squared differences between the model and target correlators
\begin{equation}\label{eq:MMD-loss}
    \mathcal{L}(\thv)=\sum_{A\subseteq[n]} w_A \big(z_A(\thv)-t_A\big)^2 \;\;,
\end{equation}
where, for the Gaussian kernel of bandwidth $\sigma$, the weights are given by $w_A=(1-p_\sigma)^{n-|A|}p_\sigma^{|A|}$ with $p_\sigma=\frac{1-e^{-1/(2\sigma^2)}}{2}$. Note that this loss is the squared value of the MMD distance \cite{gretton2012kernel}.

Notably, the bandwidth parameter $\sigma$ dictates the `bodyness' of the MMD loss \cite{rudolph2023trainability}. Specifically, when $\sigma \in \Theta(\sqrt{n})$, the weights $w_A$ primarily emphasize low-body correlators (i.e., small $|A|$). Conversely, for $\sigma \in \mathcal{O}(1)$, the loss assigns significant weight to higher-order correlators. This effectively shifts the model's sensitivity, prioritizing the matching of many-body correlations over simpler, low-order ones. Therefore, we will refer to low-body MMD loss as those in which the bandwidth $\sigma \in \Theta(\sqrt{n})$.

\subsection{Exponential concentration}
\label{sec:global-vs-local-trainability}

Variational quantum algorithms exhibit a range of scalability challenges that also apply to generative models. Training can be hindered by limited circuit expressivity, local minima, or an unfavorable optimization landscape. Among these issues, exponential concentration of the loss landscape—equivalently, the barren plateau phenomenon~\cite{mcclean2018barren,larocca2024review,arrasmith2021equivalence}—is the most thoroughly studied.

A landscape is said to exponentially concentrate if the variance of the loss over parameters vanishes exponentially with the system size \(n\), i.e.,
\begin{align}
\Var_{\thv}[\LC(\thv)] \in \OC(\exp(-n))\,.
\end{align}
Combined with Chebyshev’s inequality,
\begin{align}
{\rm Pr}\!\left(|\LC(\thv) - \Ebb_{\thv\sim\DC}[\LC(\thv)]|\geq \epsilon\right)\leq\frac{{\rm Var}_{\thv\sim\DC}[\LC(\thv)]}{\epsilon^2}\, , 
\end{align}
this implies that the probability a random point in the landscape deviates appreciably from its mean vanishes exponentially in \(n\). Consequently, detecting such deviations at typical parameter values requires exponentially fine statistical precision, and hence an exponential number of measurement shots~\cite{arrasmith2020effect, aghaei2026pitfalls}.

Exponential concentration is now understood to arise from several mechanisms, including highly expressive circuits~\cite{mcclean2018barren, holmes2021connecting, fontana2023theadjoint, ragone2023unified, srimahajariyapong2025connecting}, properties of the target problem~\cite{holmes2020barren}, and other effects such as entanglement structure~\cite{patti2020entanglement,marrero2020entanglement}, loss-function locality~\cite{cerezo2020cost,letcher2023tight}, and noise~\cite{wang2020noise,crognaletti2024estimates}. More broadly, these phenomena can be seen as a manifestation of the curse of dimensionality~\cite{cerezo2023does}. 

While barren plateaus were initially characterized within the context of variational quantum algorithms, they were more recently extended to quantum machine learning problems involving training data and non-linear losses~\cite{xiong2025role, thanaslip2021subtleties, thanasilp2022exponential, kubler2021inductive, suzuki2023effect, xiong2023fundamental, shaydulin2021importance,leone2022practical}, including quantum generative modeling~\cite{rudolph2023trainability, shen2026characterizing, tangpanitanon2020expressibility,hirviniemi2026preventing, chang2024latent}. This transition highlights the shared fundamental challenge of exponential concentration across diverse parametrized quantum architectures and frameworks.

\subsection{Initialization strategies}\label{sec:initialization_strategies}
The presence of a barren plateau landscape for a uniform initialization across the full parameter space does not necessarily imply that a quantum model is untrainable. Even within a predominantly flat landscape, there may exist small regions with substantial gradients that contain favorable local minima. This observation has motivated a range of analytical and numerical studies exploring alternative initialization strategies to bypass barren plateaus~\cite{mhiri2025unifying,puig2024variational,puig2026warm,dborin2022matrix, goh2023lie,sauvage2021flip, gibbs2024exploiting, liu2023mitigating, rudolph2022synergy,grimsley2022adapt, mele2022avoiding,grant2019initialization,volkoff2021large,zhang2022escaping,wang2023trainability,park2023hamiltonian,park2024hardware,chang2024latent,shi2024avoiding}. 
One well-known problem-agnostic approach is the (near) identity initialization, in which the circuit parameters are chosen within a small neighborhood around the identity of the ansatz~\cite{grant2019initialization,volkoff2021large,zhang2022escaping,wang2023trainability,park2023hamiltonian,park2024hardware,chang2024latent,shi2024avoiding}. Although this represents one of the earliest and most extensively studied techniques for combating barren plateaus, it does not exploit problem-specific structure and may not always successfully avoid gradient suppression~\cite{mhiri2025unifying}.

Notably, for losses that are expressed as a single expectation value of an observable Ref.~\cite{mhiri2025unifying} unifies previous results on small-angle initialization and extends them to a broader class of circuit architectures. The key insight is the established link between local curvature (i.e., the Hessian) and the fluctuations of the loss function. Specifically, if the curvature at the center of the patch is at least polynomially large in the system size, the variance of the loss within a small neighborhood of that point is likewise at least polynomially large. Consequently, initializing parameters within this neighborhood ensures that the gradients do not vanish exponentially. In this crude form, the argument may appear somewhat circular and borderline trivial but we argue that this intuition implicitly underlies all known provable guarantees for small-angle initializations~\cite{grant2019initialization,volkoff2021large,zhang2022escaping,wang2023trainability,park2023hamiltonian,park2024hardware,chang2024latent,shi2024avoiding}.
In generative modeling, where loss functions are typically non-linear functions of expectation values, the relationship between alternative initialization and barren plateaus however remains less explored. 

Importantly, beyond the barren plateau phenomenon, initialization plays a crucial role in determining overall optimization performance. This has been empirically observed across various domains, including quantum chemistry~\cite{zou2025generative,ravi2022cafqa, mitarai2022quadratic,wilson2019optimizing,grimsley2022adapt,rudolph2022synergy}, quantum optimization~\cite{zhou2020quantum,egger2021warm,verdon2019learning,tate2021classically,wurtz2021fixed,zhou2020quantum, akshay2021parameter} as well as machine learning~\cite{recio2025train,niu2023warm,rudolph2022synergy}. 
Specifically for generative modeling, Ref.~\cite{recio2025train} employs a data-dependent technique to classically train IQP-QCBMs, achieving remarkable success on diverse target distributions for systems as large as $1,000$ qubits. Given that IQP-QCBMs exhibit barren plateaus under full-angle random initialization, as corroborated by Ref.~\cite{shen2026characterizing} and the present work, these results provide strong empirical evidence that data-dependent initialization can substantially mitigate barren plateaus. Nevertheless, a rigorous analytical understanding of this data-dependent initialization strategy remains an open problem.

In this work, we investigate alternative initialization strategies in the context of generative models, focusing primarily on IQP circuits and the MMD loss. We consider a parameter patch centered around the point $\thv^*$ with a half-width $r$, defined as the hypercube  
\begin{align}
    \VC(\thv^*,r) = \{\thv \;|\;  \th_i \in  [\th_i^* -r , \th_i^*+r]\,\;;\;\forall \,i\}\,\;\;,
\end{align}
with a uniform distribution over this region given by
\begin{align}\label{eq:param-drawn}
    \DC(\thv^*,r) = {\rm Unif}[ \VC(\thv^*,r)] \;\;.
\end{align}
We then examine the following initialization strategies, starting with the data-agnostic approaches.
\begin{description}
    \item[Full-angle random initialization] In this case $\thv^*=\vec{0}$ and $r=\pi/2$.
    This corresponds to the random initialization over the whole loss landscape. 
    \item[Identity initialization] In this case $\thv^*=\vec{0}$ and $r\ll 1$.
    This defines a restricted patch around the identity and serves as a canonical small-angle baseline for comparison with a data-dependent method.
    \item[Unbiased initialization] In this case $\thv^*$ is such that single qubit gates parameters are set to $\theta_j^*=\pi/4$ and two qubits gates to $\theta_{ij}^*=0$, and we consider $r \ll 1$. 
    At this reference point, the model distribution in Eq.~\eqref{eq:def-model-dist} becomes a uniform distribution over all possible bit strings, hence representing the maximally unbiased initial guess.
\end{description}
The final initialization strategy we consider uses prior knowledge of the data structure.
\begin{description}
    \item [Data-dependent initialization] In this case $\thv^*$ depends on the training data (or empirical target statistics) through marginal matching.
    We will focus on a simplified version of the initialization used in Ref.~\cite{recio2025train}.
\end{description}

In this paper we analytically investigate whether these initialization strategies allow the model to operate in parameter regimes where the loss function has substantial enough variation to allow for training at initialization. To quantify this, we adopt the following technical criterion: a loss $\LC(\thv)$ is considered \emph{initially trainable} with respect to a given initialization strategy if 
\begin{equation}\label{eq:initially-trainable}
    \Var_{\thv}[\LC(\thv)]\in\Omega\!\left(\frac{1}{\poly(n)}\right) \;\;.
\end{equation}
In practical terms, this ensures that fluctuations in the loss $\LC(\thv)$ can be resolved with a polynomially scaling number of samples.

\section{Results}
\label{sec:results}

\subsection{Overview and practical guidelines}\label{sec:overview}
We investigate how different initialization strategies influence variance scaling and training for the MMD loss with IQP circuit architectures. In this section, we provide an executive summary and guidelines for practitioners. Our analytical results are summarized in Fig.~\ref{fig:main_fig} which provides a conceptual illustration of the optimization landscapes under these distinct initialization regimes. Our key findings are informally summarized as follows: \\

\noindent \textbf{1.~Full-angle random initializations exponentially concentrate.} 
    For IQP circuits with all-to-all connectivity and full-angle random initialization, we prove that the MMD loss exponentially concentrates (see Theorem~\ref{thm:mmd-loss-concentration-random-alltoall}). However, for less connected models, e.g., those with constant degree connectivity, exponential concentration can be avoided.

\medskip

\noindent  \textbf{2.~Small-angle initializations can avoid exponential concentration at any point with sufficient curvature.} Despite the overall barren plateau landscape, we show in  Sec.~\ref{sec:restricted-patch-can-avoid-BP} that the correlators and the MMD loss exhibit at least polynomially large variance within any patch (with width $r \in \OC(1/\poly(n))$) provided that the curvature (i.e., the Hessian) at the center of the patch is non-vanishing ($\in  \OC(1/\poly(n))$). This is shown in Theorem~\ref{thm:MMD-variance-patch-lower-bound-curvature-informal}. 
Practically, this suggests that even in non-linear losses one can diagnose the susceptibility of an initialization point to concentration by evaluating its local curvature at that point. As detailed in Sec.~\ref{sec:restricted-patch-can-avoid-BP}, we show that this curvature decomposes into two distinct components: a \emph{mismatch-driven} term and a \emph{model-sensitivity} term.

\medskip

\noindent  \textbf{3.~Data-agnostic initialization strategies.} In Sec.~\ref{sec:data-agnostic-initialization-main}, we show that initializing within a small patch around the identity ($\thv^*=\vec{0}$) or the unbiased point ($\theta_j^* = \frac{\pi}{4}$) yields sufficient curvature to ensure polynomially large variances, thereby evading barren plateaus as the system scales (see Theorem~\ref{thm:variance-guarantee-agnostic-initialisation-main}).
However, their practical utility differs significantly. While identity initialization is a popular approach, the resulting model distribution at this point is a delta distribution concentrated on the all-zero bit string. Such a highly biased starting point is often detrimental to generative modeling (as shown in Fig.~\ref{fig:performance_curves}). In contrast, the unbiased initialization produces a uniform distribution over all bit strings, representing perhaps the most natural prior for a data-agnostic strategy.

\medskip

\noindent \textbf{4.~Data-dependent initialization:}  
For many real world data sets (see~Assumption~\ref{ass:weakly-correlated-target} and Assumption~\ref{ass:no-scaling-with-n}), we prove that the data-dependent marginal-matching initialization guarantees inverse-polynomial local curvature (we will also refer to this as \textit{polynomially large} curvature), thereby ensuring non-vanishing patch variance for the MMD loss. 
Empirical evidence further demonstrates that this strategy yields superior convergence speed and model fidelity compared to data-agnostic alternatives. 
This highlights a fundamental shift in perspective: while both agnostic and data-dependent strategies can evade barren plateaus by providing non-vanishing gradients, the latter ensures the optimization begins in a significantly more favorable region. By aligning the initial model distribution with the underlying correlations of the target, this strategy generally initializes the parameters within a ``high-quality'' patch of the landscape, thereby significantly reducing the effective distance to a global minimum (or, at least, a higher-performing local minimum).    

\subsection{Barren plateaus under the full-angle random initialization}
\label{sec:global-concentration}

We start by investigating the variance scaling of the MMD loss with for all-to-all IQP architectures with full-range random initializations. That is, we here study the case where the parameters are sampled independently according to $\DC(\vec{0},\pi/2)$. We find that the loss exponentially concentrates in this regime and therefore the model is untrainable when initialized randomly over the whole parameter space.

Intuitively, this concentration is expected at the level of individual correlators. With all-to-all connectivity, the effective Pauli light-cone of any $Z_A$ operator, when back-propagated through the circuit, becomes global i.e., encompasses all $n$ qubits. This setup therefore triggers something akin to so called `globality-induced' concentration of the correlators~\cite{cerezo2020cost,letcher2023tight}.

In Appendix~\ref{app:full-angle-correlator-variance}, we provide a general derivation of the correlator variance that explicitly reveals its dependence on the underlying circuit topology. This result shows that the variance of any $z_A(\thv)$ is fundamentally governed by the reach of its effective Pauli light-cone within the architecture. By specializing this bound to the all-to-all case, we find that the variance of every non-trivial correlator vanishes exponentially with the system size $n$.

While no individual component of the MMD loss provides a detectable signal, and therefore it is entirely intuitive that it will not be able to train the full loss, formally establishing that the full MMD loss exponentially concentrates is a little fiddly. One must further ensure that the cross-correlations arising from the quadratic expansion of the variance in Eq.~\eqref{eq:MMD-loss} do not constructively combine to produce a non-vanishing signal. By bounding these interference terms, we demonstrate that no such signal recovery occurs, leading to the following theorem:

\begin{theorem}[Exponential concentration of the MMD loss]
\label{thm:mmd-loss-concentration-random-alltoall} 
Let the parameters of the IQP circuit with all-to-all connectivity in Eq.~\eqref{eq:IQP-circuit} be independently drawn from ${\rm Unif}([-\pi/2,\pi/2])$. Then the variance of the MMD loss over the parameter space exponentially concentrates with
\begin{equation}
    \Var_{\thv}[\mathcal{L}(\thv)]\in\OC(2^{-n}) \;\;.
\end{equation}
\end{theorem}
The proof is given in Appendix~\ref{app:exponential-concentration-MMD}. 
Notably, this concentration bound holds regardless of the specific structure or correlations present in the target distribution.
Theorem~\ref{thm:mmd-loss-concentration-random-alltoall} therefore shows that full-angle random initialization in all-to-all IQP circuits leads to ineffective training at scale.

In Appendix~\ref{app:mmd-variance-lower-bound-full-angle-random}, we further show that IQP circuits with low connectivity do not concentrate with a low-body MMD ($\sigma\in\Theta(\sqrt{n})$). In particular, as long as there is at least one target low-body correlator (with $|A|\in\OC(1)$) such that $t_A\in\Omega(1/{\poly}(n))$, we show that the MMD loss function with an IQP circuit in which each qubit is connected to $K$ other qubits, with $K\in\order{\log(\poly(n))}$, will have a non-exponentially vanishing variance $\Var_{\thv}[\mathcal{L}(\thv)]\in\Omega(1/\poly(n))$. This is proven in Proposition~\ref{thm:mmd-trainability-random-kregular}, in Appendix~\ref{app:mmd-variance-lower-bound-full-angle-random}.  

\subsection{Evading barren plateaus with small-angle initializations}
\label{sec:restricted-patch-can-avoid-BP}

Having established that IQP-QCBMs exhibit exponential concentration over the full loss landscape, we now turn to the analysis of local subregions. In particular, concentration over a large domain does not preclude the existence of smaller regions with high curvature and only inverse-polynomial loss variance, in which gradient signals remain accessible. \\

\paragraph*{Variance of an arbitrary non-linear loss in patch around curvature.}
We begin by considering a more general setting, comprising of an \textit{arbitrary} non-linear loss $\mathcal{L}(\thv)$ and an arbitrary parametrized quantum circuit. Here, we prove a variance lower bound showing that if a point exhibits substantial curvature and a loss satisfies certain conditions regarding high-order derivatives, then a local neighborhood around that point must exhibit substantial loss variance. 

Technically, this requires extending the proof techniques developed for observable expectations in Ref.~\cite{mhiri2025unifying} to an arbitrary non-linear functions. 
The required conditions on the loss primarily ensure that higher-order derivatives remain bounded within the region of interest (see. Eq.~\eqref{eq:assumption-even-deriv-bound-thm} and Eq.~\eqref{eq:assumption-fourth-deriv-bound-thm}).
As a primary application, we focus on the IQP circuit architecture and the MMD loss, which involves quadratic combinations of model correlators together with target-dependent offsets.  

\begin{theorem}[Lower bound guarantee of an arbitrary non-linear loss with sufficient curvature, informal]
\label{thm:MMD-variance-patch-lower-bound-curvature-informal}
Consider an arbitrary non-linear loss $\LC(\thv)$ that satisfies some bounded derivative conditions and any quantum process parametrized with the set of parameters $\thv$.  Particularly, these include the MMD loss in Eq.~\eqref{eq:MMD-loss} and the IQP circuit in Eq.~\eqref{eq:IQP-circuit}. Further suppose that the parameters are distributed over a hypercube patch
\begin{equation}
    \thv \sim \thv^*+\mathrm{Unif}([-r,r]^m) \;\;.
\end{equation}
If there exists a parameter index $\alpha$ such that
\begin{equation}
\left|
\left.
\left(\frac{\partial^2 \mathcal{L}(\thv)}{\partial\theta_\alpha^2}
\right)\right|_{\thv=\thv^*}
\right|
\in \Omega\!\left(\frac{1}{\poly(n)}\right),
\end{equation}
then there exists a patch with half-width $r\in \mathcal{O}(1/\poly(n))$ such that
\begin{equation}
\Var_{\thv}[\mathcal{L}(\thv)]\in \Omega\!\left(\frac{1}{\poly(n)}\right).
\end{equation}
\end{theorem}
A formal version of Theorem~\ref{thm:MMD-variance-patch-lower-bound-curvature-informal} is provided in in Appendix~\ref{app:patch-variance-curvature} together with its proof.

We emphasize that this analytical framework is highly general and not restricted to the MMD loss. By bridging the gap between local curvature and variance lower bounds, it offers a broad toolset for analyzing warm-start strategies across diverse quantum machine learning paradigms. However, a degree of caution is necessary when applying these results beyond the MMD loss. While a non-vanishing loss variance is a strong indicator that it is possible to avoid the problems associated with exponential concentration, it is not a stand-alone guarantee of success. As highlighted in Ref.~\cite{aghaei2026pitfalls}, one must pay attention to the specific quantities in the loss that are estimated with a quantum computer. If the underlying components of the loss function themselves concentrate, the impact of exponential concentration will still be detrimental to training, even if some post processing step superficially leads to a non-exponentially vanishing loss variance~\footnote{For a toy example to illustrate this, one could turn a loss from one that exponentially concentrates to one that does not by taking the exponential of it. However, this does not avoid the problems of exponential concentration because it also amplifies the shot noise exponentially.}.  \\

\paragraph*{Curvature Mechanisms at Initialization.}

While the previous subsection establishes that local curvature guarantees patch variance, it does not prescribe how to locate such regions of high curvature. To bridge this gap, we now investigate the mechanisms that generate curvature. Concretely, we derive a formal curvature decomposition for generative modelling with the MMD loss. Note that this curvature decomposition is not specific to IQP-QCBMs and applicable to other circuit architectures.
We will later apply this decomposition to data-agnostic and data-dependent initialization strategies in Section~\ref{sec:data-agnostic-initialization-main} and Section~\ref{sec:data-dependent-guarantee} respectively. 

For a parameter $\theta_\alpha$ associated with a single-qubit generator, the MMD curvature (i.e., the second-order derivative of the MMD loss) can be expressed as:
\footnotesize
\begin{align}\label{eq:MMD-double-derivative-single-parameter-main}
    \frac{\partial^2 \mathcal{L}(\thv)}{\partial\theta_{\alpha}^2}
    =\sum_{\substack{A\subseteq [n]\\ \alpha \in A}}\!2w_A\!
    \left[
    4 \underbrace{z_A(\thv)(t_A-z_A(\thv))}_{\rm Data\,Mismatch}
    +\!\underbrace{\big(g_A^{(\alpha)}(\thv)\big)^2}_{\rm Model\, Sensitivity}
    \right] ,
\end{align}
\normalsize
where $g_A^{(\alpha)}(\thv)=\partial_{\theta_\alpha} z_A(\thv)$ represents the gradient of the correlator and we recall that $t_A$ represents the target correlators as defined in Eq.~\eqref{eq:target-correlators-definition}. Thus we see that the local curvature of the MMD loss arises from two fundamental mechanisms:
\begin{enumerate}
    \item \emph{Data-mismatch contribution:} $z_A(\thv)(t_A-z_A(\thv))$. 
    This term scales with the discrepancy between model and target correlators. It is large when the model and target distributions are significantly misaligned but share common feature support.
    \item \emph{Model-sensitivity contribution:} $\big(g_A^{(\alpha)}(\thv^*)\big)^2$. This term is intrinsic to the model's architecture, reflecting how sensitively the model's output changes with respect to its parameters independently of the specific target data.
\end{enumerate}

\subsection{Data-agnostic initialization strategies}
\label{sec:data-agnostic-initialization-main}
We now examine standard data-agnostic initialization strategies, which serve as examples of where each curvature mechanism (data-mismatch or model-sensitivity) dominates. \\

\paragraph*{Identity initialization (Mismatch-driven).}
We first consider the identity initialization strategy, $\thv = \vec{0}$, as described in Section~\ref{sec:initialization_strategies}. This corresponds to initializing the loss function in a maximum with respect to $\thv$. This follows trivially from the fact that the maximum value of the correlators $z_A$ is 1 and  $z_A(\vec{0}) = \langle\vec{0}|Z_A|\vec{0}\rangle = 1$. Then, because this is a maximum, we can trivially conclude that the corresponding derivatives are zero, $g_A^{(\alpha)}(\vec{0}) = \partial_{\theta_\alpha} z_A(\thv) = 0$). Hence the model sensitivity term vanishes and the curvature is entirely dominated by the data mismatch in Eq.~\eqref{eq:MMD-double-derivative-single-parameter-main}, and we have
\begin{equation}
    \frac{\partial^2 \mathcal{L}(\vec{0})}{\partial\theta_{\alpha}^2}
    = -\sum_{\substack{A\subseteq[n]\\ \alpha\in A}} 8w_A\,(1-t_A).
\end{equation}
Because $ 0\leq t_A\leq 1$, we find that the second derivative is always negative. Therefore, we can conclude that the full loss $\LC(\vec{0})$ is also a maximum. 

In this manner, we observe that the identity initialization does indeed have substantial gradients and that this comes from a substantial target--model mismatch across many low-weight features. Intuitively, always guessing the all-zero bit string is rarely a "good" guess for complex data - but this means there is substantial local curvature which, via Theorem~\ref{thm:MMD-variance-patch-lower-bound-curvature-informal}, translates into a non-vanishing patch variance (see Theorem~\ref{thm:variance-guarantee-agnostic-initialisation-main} below). In this sense, identity initialization can be used to avoid barren plateaus at initialization, but does not necessarily provide a high-quality starting point for the generative task in terms of training performance. \\

\paragraph*{Unbiased initialization (Sensitivity-driven).}
We next consider an unbiased initialization where all single-qubit angles are set to $\pi/4$ and all two-qubit angles to $0$. At this point, the model distribution is uniform over all computational-basis states, resulting in:
\begin{equation}
    z_A(\thv^*)=0 \qquad\text{for all non-trivial }A.
\end{equation}
Consequently, the mismatch-driven contribution in Eq.~\eqref{eq:MMD-double-derivative-single-parameter-main} vanishes. The only surviving terms are those arising from model sensitivity. For a single-qubit parameter $\theta_\alpha$, all the model sensitivities vanish except for the one corresponding\footnote{Since two-qubit parameters are set to $0$, we have $z_A(\thv^*)=\prod_{j\in A}{\rm cos}(2\theta_j^*)$ and $g_{\{\al\}}^{(\al)}(\thv^*)=-2{\rm sin}(2\theta_{\al}^*)=-2$.} to $[g^{(\alpha)}_{\{\alpha\}}(\thv^*)]^2 = 4$, thus we have
\begin{equation}
    \left|
    \partial_{\theta_\alpha}^2 \mathcal{L}(\thv^*)
    \right|
    =8w_{\{\alpha\}}.
\end{equation}
Thus this provides a target-agnostic example where a non-vanishing optimization signal arises solely from the architectural geometry. It demonstrates that sensitivity-driven curvature is not exclusive to target-aware strategies; even when all non-trivial correlators vanish and the mismatch signal is suppressed, the local landscape remains curvature-rich due to the sensitivity of the IQP circuit's first layer. \\

To formally conclude this subsection, the following theorem summarizes how both the identity and unbiased initialization strategies successfully provide a non-vanishing signal within a local patch.
\begin{theorem}[Variance guarantee for agnostic initialization strategies, informal]
\label{thm:variance-guarantee-agnostic-initialisation-main}
Consider the low-body MMD, $\sigma\in\Theta(\sqrt{n})$, loss and the following initialization scenarios, either
\begin{enumerate}
    \item The initialization center is at identity $\thv^*=\vec{0}$, and further assume that there exists at least one single bit marginal that is system size independent, or:
    \item The unbiased initialization center such that single qubit gates parameters are set to $\pi/4$ and two-qubit gates parameters are set to $0$.
\end{enumerate}     
Then, there exists a patch size
\begin{equation}
r\in\mathcal{O}\!\left(\frac{1}{\poly(n)}\right)
\end{equation}
such that
\begin{equation}
\Var_{\thv}[\mathcal{L}(\thv)]
\in\Omega\!\left(\frac{1}{\poly(n)}\right)
\end{equation}
for $\thv\sim\mathrm{Unif}([-r,r]^m)$.
\end{theorem}
A formal version of this Theorem and the accompanying proofs is given in Appendix~\ref{app:data-agnostic-initialization-curvature}. We further note that the assumption regarding the target distribution for the identity initialization is primarily a technical requirement for the proof. The only family of distributions that fails to satisfy this condition consists of highly peaked distributions centered around the all-zero bit string which are not sufficiently complex as to be interesting anyway.

\subsection{Data-dependent initialization strategy: variance guarantee and superior empirical performance}
\label{sec:data-dependent-guarantee}

To further refine the starting point, we consider a data-dependent initialization strategy where the model's low-weight correlators are explicitly matched to those of the target distribution. Specifically, we initialize the single-qubit parameters $\theta_\alpha$ such that the model's single-body correlators $z_{\{\alpha\}}(\thv^*)$ equal the target correlators $t_{\{\alpha\}}$. All two-qubit parameters are set to zero. \\

\paragraph*{Variance guarantee.}

For this initialization, first-order mismatch terms in the curvature vanish by construction. However, information about the target distribution is implicitly encoded in the model-sensitivity terms.
More explicitly, and as shown in Appendix~\ref{app:data-dependent-initialization-curvature}, for a single-qubit parameter $\theta_\alpha$, the curvature in Eq.~\eqref{eq:MMD-double-derivative-single-parameter-main} becomes
\begin{equation}\label{eq:curvature-data-dependent-initialization-centre}
\frac{\partial^2 \mathcal{L}(\thv^*)}{\partial\theta_{\alpha}^2}
= \sum_{\substack{A\subseteq[n]\\ \alpha\in A}} 8w_A
\Big( M_A^{(\alpha)} + S_A^{(\alpha)} \Big),
\end{equation}
where
\begin{align}
M_A^{(\alpha)} &:= \left(t_A-\prod_{j\in A}t_j\right)\prod_{j\in A}t_j\;\;,\\
S_A^{(\alpha)} &:= (1-t_\alpha^2)\prod_{j\in A\setminus\{\alpha\}} t_j^2 \;\;.
\end{align}
Here, $M_A^{(\alpha)}$ represents a residual mismatch contribution arising from higher-order correlations that single-qubit matching cannot resolve, while $S_A^{(\alpha)}$ represents the model-sensitivity contribution, now modulated by the target's own marginals.

To proceed further we observe that many natural data distributions are dominated by low-body correlations with higher-order correlators typically very small. A canonical physical realization of this structure is found in the Curie-Weiss model within high-temperature mean-field regimes, where higher-order moments are functionally determined by lower-order correlators~\cite{kirsch2019curie,arous1999increasing,vsamaj1988improved}. Beyond physical systems, this trend is observed in the empirical image and genetic datasets of Ref.~\cite{recio2025train}, and more broadly this observation could be seen as underpinning the success of modern generative modelling in high-dimensional spaces. 
We will thus assume the target distribution to have only low-body correlations in the following sense:\begin{assumption}[Approximately factorizable target correlators]
\label{ass:weakly-correlated-target}
Consider the target Pauli-$Z$ correlators $t_A=\Ebb_{\xv\sim p_{\mathrm{data}}}[(-1)^{\sum_{j\in A}x_j}]$.
Assume that for all subsets $A\subseteq[n]$ with $|A|\ge 2$,
\begin{equation}\label{eq:weak-correlation-assumption}
\Big|t_A-\prod_{j\in A}t_j\Big|
\le \left(\frac{C}{n}\right)^{|A|/2}\,,
\end{equation}
where $C$ is a constant and we consider $n$ large enough such that $C \leq n$.
\end{assumption}
Physically, this implies that the target distribution behaves as an approximate product state, wherein the influence of higher-order correlators decays rapidly as the subsystem size increases.

In addition, we exclude from our analysis the trivial case of target distributions that are highly concentrated around a single bit string. Such highly peaked distributions lack the complexity required for meaningful generative tasks. Formally, we require that the target distribution maintains non-vanishing fluctuations in at least one marginal, leading to the following assumption
\begin{assumption}[No peak target around one bit string]\label{ass:no-scaling-with-n}
The target distribution has at least one constant single bit marginal:
\begin{equation}\label{eq:assumption-nonsaturated-marginals-main}
\exists\al \in[n] \;: \quad1-t_{\al}^2\in\Theta(1).
\end{equation}
\end{assumption}
Crucially, we note that while these assumptions are essential for the formal rigor of our proof (especially for Assumption~\ref{ass:no-scaling-with-n}), they should be viewed as technical sufficient conditions rather than strict requirements for the method's practicality.

We can now return to the analysis of Eq.~\eqref{eq:curvature-data-dependent-initialization-centre} in light of the above discussion. Under Assumptions~\ref{ass:weakly-correlated-target} and~\ref{ass:no-scaling-with-n}, the residual mismatch terms $M_A^{(\alpha)}$ for $|A| \ge 2$ are exponentially suppressed. However, the non-saturation condition, Eq.~\eqref{eq:assumption-nonsaturated-marginals-main}, implies that $|t_\alpha| < 1$ which prevents the sensitivity-driven terms $S_A^{(\alpha)}$ from vanishing. In particular, the first-order sensitivity $S_{\{\alpha\}}^{(\alpha)} = 1 - t_\alpha^2$ remains bounded away from zero. This leads to the following variance guarantee for data-dependent initialization (with the proof in Appendix~\ref{app:target-analysis}).
As detailed further in the appendices, we suspect our assumptions on the target distribution are only necessary for our proof strategy and the initialization strategy is likely to be effective for a broader regime of target distributions. Indeed, this hunch is supported by our numerical implementations below.

\begin{theorem}[Variance guarantee for data-dependent initialization strategy, informal]
\label{thm:data-dependent-local-trainability-main}
Consider the low-body MMD, $\sigma\in\Theta(\sqrt{n})$, regime and a marginal-matching initialization center $\thv^*$ with all single-qubit parameters initialized to match one-body marginals and all two-qubit parameters initialized to $0$.
Under Assumptions~\ref{ass:weakly-correlated-target} and~\ref{ass:no-scaling-with-n}, there exists at least one single-qubit parameter direction $\alpha$ such that
\begin{equation}
\left|
\left.
\left(\frac{\partial^2 \mathcal{L}(\thv)}{\partial\theta_\alpha^2}
\right)\right|_{\thv=\thv^*}
\right|
\in \Omega\!\left(\frac{1}{\poly(n)}\right).
\end{equation}
Consequently, by Theorem~\ref{thm:MMD-variance-patch-lower-bound-curvature-informal}, there exists a patch with half-width
\begin{equation}
r\in\mathcal{O}\!\left(\frac{1}{\poly(n)}\right)
\end{equation}
such that
\begin{equation}
\Var_{\thv}[\mathcal{L}(\thv)]
\in\Omega\!\left(\frac{1}{\poly(n)}\right)
\end{equation}
for $\thv\sim\thv^*+\mathrm{Unif}([-r,r]^m)$.
\end{theorem}

Theorem~\ref{thm:MMD-variance-patch-lower-bound-curvature-informal} showcases that for a low-body MMD, data-dependent strategies guarantee non-exponentially vanishing loss functions. Because of their nature, one could expect data-dependent initialization strategies to work better than the data-agnostic strategies. However, as this is hard to establish via variance guarantees alone, we proceed to analyze this numerically.

\medskip

\subsection{Numerical studies}
\label{sec:numerical-studies}
\begin{figure*}
    \centering
    \includegraphics[width=1\linewidth]{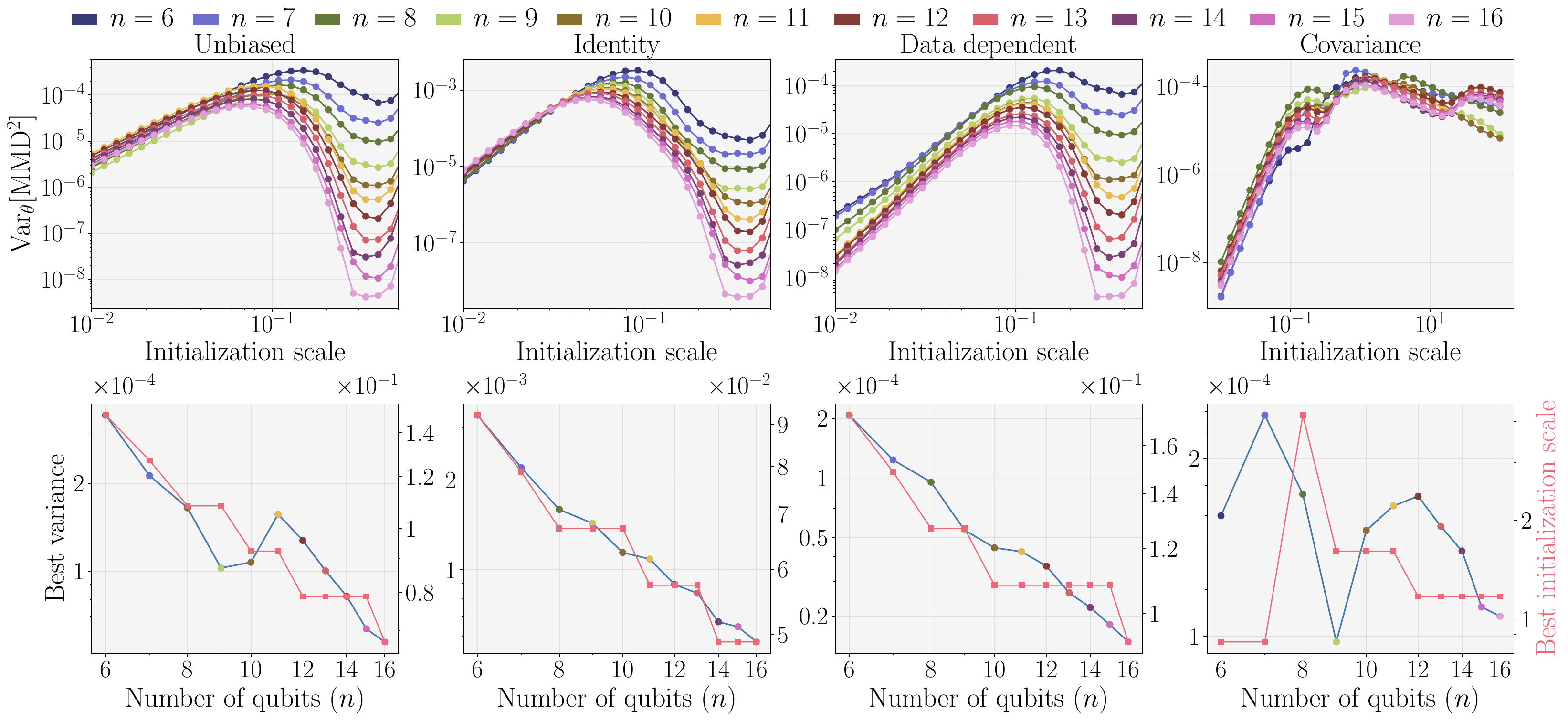}
    \caption{\textbf{Variance of the low-body MMD estimator versus initialization scale.}
Top row: Variance of the MMD estimator as a function of the initialization scale for four initialization schemes (identity, unbiased, data-dependent, and covariance) using a kernel with bandwidth that scales as $\sigma\in\Theta(\sqrt{n})$. Curves correspond to different numbers of qubits $n$ ranging from $n=6$ to $n=16$. For the first three schemes, parameters are initialized as
$\theta_j\sim\theta_j^*+{\rm Unif}[-\frac{\pi}{2}s,\frac{\pi}{2}s]$
where $s$ is the initialization scale. The covariance initialization (right column) follows Eq.~\eqref{eq:new_data_dependent_init_covariances}, where two-qubit gate parameters are correlated and perturbations are rescaled according to the target distribution covariances. The target distribution is given by a genomic dataset.
Bottom row: Maximum variance over initialization scales (blue) and the initialization scale achieving this maximum (red dashed) as functions of the number of qubits $n$.}
    \label{fig:mmd_var_vs_patch-size}
\end{figure*}

We devote this section to analyzing the performance of data-dependent initializations and comparing them to data-independent ones. Concretely, we compare two data-agnostic initializations ($\thv^* = \vec{0}$ and $\th_j = \pi/4$) with two data-dependent strategies. These include (i) the approach we theoretically analyzed (which matches one-body marginals and sets the two-qubit gate parameters to zero) and (ii) the initialization method from Ref.~\cite{recio2025train}. In the latter, the one-qubit gates are fixed to match the one-body marginals, as in (i), and two-qubit gate angles are initialized based on the covariances of the target distribution, $C_{ij}=t_{ij}-t_it_j$:
\begin{align}\label{eq:new_data_dependent_init_covariances}
    \theta_{jk}\sim\frac{C_{ij}}{C_{\rm max}}{\rm Unif}\left(\left[-\frac{\pi}{2} s,\frac{\pi}{2} s\right]\right),
\end{align}
where $C_{\max}=\max_{ij}|C_{ij}|$ and $s\in[0,1]$ denotes the initialization scale. For the three cases studied analytically we have $r=\frac{\pi}{2}s$ so the initialization scale can be interpreted as the patch half-width $r$. The target distributions are given by empirical distributions of the first $n$ dimensions of the genomic training dataset described in \cite{recio2025train}. We investigate the scaling of the MMD loss variance as well as loss curves resulting from training IQP-QCBMs with 150 qubits. This is done via state-vector simulation via the JAX interface of PennyLane \cite{bergholm2018pennylane} and the IQPopt package \cite{recio2025iqpopt, ErikRecio2025}, respectively.

\begin{figure*}
    \centering
    \includegraphics[width=1\linewidth]{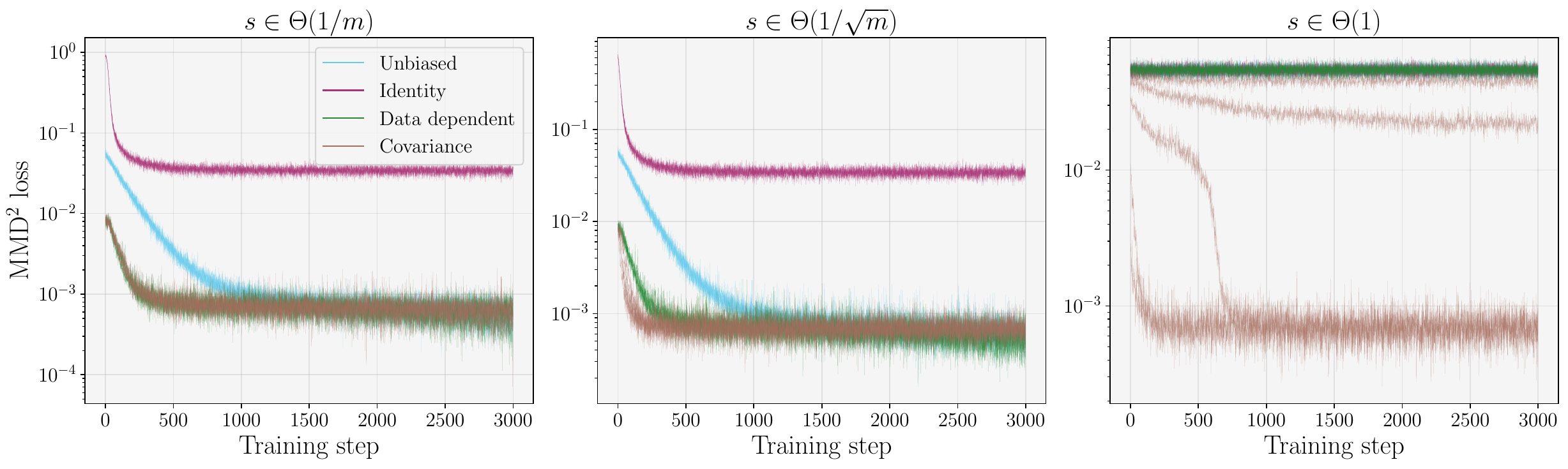}
    \caption{\textbf{Training of the low-body MMD estimator for different initializations. }Training loss $\mathcal{L}(\thv)$ as a function of gradient-descent iterations for a $n=150$ qubit model trained on the genomic dataset. Each curve corresponds to one of the four initialization strategies: identity (red), unbiased (blue), marginal matching (green), and the covariance-based data-dependent initialization (light green). The three panels correspond to different initialization scales $s$ (as defined in Fig.~\ref{fig:mmd_var_vs_patch-size}): 
 $s=1/m$, $s=1/\sqrt{m}$, and $s=1$, where $m$ denotes the total number of parameters. Each model is trained five times, with most initializations leading to similar overlapping loss curves except for the covariance-based initialization for $s\in\Theta(1)$.
For the linear and sqrt scalings, identity initialization shows an initial rapid decrease of the loss but quickly reaches a plateau at substantially larger loss values, leading to poor convergence compared to the other strategies. The unbiased initialization remains trainable for these scalings but converges more slowly than the data-dependent approaches. In contrast,  under the unit scaling  (corresponding to a full angle random initialization) the marginal-matching and unbiased initializations fail to train, whereas the covariance-based initialization can still achieve successful training in some trials, highlighting the benefit of incorporating parameter correlations in the initialization.}
    \label{fig:performance_curves}
\end{figure*}

In the upper panels of Fig.~\ref{fig:mmd_var_vs_patch-size}, we plot the variance of the low-body MMD loss  as function of the initialization scale $s$, initialized according to the identity, the unbiased, and the two data-dependent strategies. This figure illustrates that while data-independent strategies provide a region with non-zero variance, the width of this region is noticeably smaller. This is further detailed in the lower panels, where we plot the value of $s\sim r$ that maximizes the variance as a function of the number of qubits. On the right axis, we show that although $\argmax_s \Var[\LC(\thv)]$ roughly scales polynomially with the number of qubits (as the $\log{\rm-}\log$ plots look linear) for all methods, the data-dependent values are strictly larger.

This advantage is then clearly reflected in the convergence behaviour of the data-dependent initialization strategies. Indeed, Fig.~\ref{fig:performance_curves} presents the convergence of a low-bodied $\sigma\in\Theta{(\sqrt{n})}$ MMD loss function that targets two-body correlations on average, demonstrating that data-dependent initializations converge significantly better and faster. Specifically, this panel shows the MMD loss value for $150$ qubits over the number of training steps for each initialization strategy for $s= 1/m$, $s= 1/\sqrt{m}$ and $s=1$. The two curves corresponding to the data-dependent strategies both exhibit a rapid decay for linear and sqrt (first and second panel). Conversely, the zero-initialization strategy  (identity) quickly stalls and fails to converge.

Finally, while the uniform initialization strategy appears to converge, it requires a significantly larger number of steps, suggesting potential scalability issues for different datasets or higher qubit counts. As expected, when randomly initializing over the full landscape, most strategies fail. This is because, for all initializations except the covariance one, the initialization is equivalent to that of Theorem~\ref{thm:mmd-loss-concentration-random-alltoall}. As described in Eq.~\eqref{eq:new_data_dependent_init_covariances}, the covariance approach has a multiplicative prefactor in $s$ that can make the range smaller for some parameters and such that all $2$-qubit gate parameters are correlated. These correlations explain why the covariance based approach can converge even for the case of `full angle range' initialization.

\section{Discussion}
\label{sec:discussion}

Our results provide an analysis of the MMD loss landscape for IQP-QCBMs that separates (i) concentration over the full parameter domain, (ii) curvature inside restricted data agnostic initialization patches, and (iii) curvature inside data-dependent initialization patches. These distinctions are useful for understanding warm starts both in the ``variational quantum algorithm'' and in the ``train on classical, deploy on quantum'' paradigms. In particular, we observe that data-dependent initializations are advantageous because they combine sufficient local loss variances with a better target alignment at initialization.

While data-dependent initialization provides a significant competitive edge over agnostic strategies, it does not provide an absolute guarantee of optimization success. The efficacy of the strategy is fundamentally capped by the information content of the available samples; specifically, if the provided training data do not contain sufficient statistical information to resolve the target distribution's structure, the initialization remains uninformed. Put more crudely, matching the $k_{\rm th}$ order marginals cannot \textit{strictly guarantee} a signal for the $(k+1)_{\rm th}$ order marginals. 

A pedagogical example of this limitation occurs when the target distribution is a computational-basis measurement of a Haar-random state, $|\psi_{\rm Haar}\rangle$. In this regime, one can show that with a realistic access to polynomial number of samples drawn from $p_{\rm data}(\zv) = |\langle \zv| \psi_{\rm Haar}\rangle|^2$, estimated correlators carry no instance-specific information with a probability exponentially close to $1$~\cite{aghaei2026pitfalls}. Consequently, the MMD loss remains statistically insensitive to the particular target instance.
Notably, this result applies to any QCBM architecture, as the bottleneck lies within the target itself rather than the model's expressive capacity. More generally, there exists a large family of target distributions that cannot be efficiently learned. For instance, Ref.~\cite{nietner2025average} utilizes the statistical query model to demonstrate the inefficiency of learning distributions arising from quantum random circuits, while Ref.~\cite{herbst2025limits} shows that target distributions exhibiting anti-concentration inevitably lead to loss concentration.

Recent work has highlighted connections between the provable absence of barren plateaus and the classical simulability (or surrogatability) of the hybrid optimization loop in variational quantum algorithms~\cite{cerezo2023does, bermejo2024quantum, angrisani2024classically,lerch2024efficient}. However, these results concern the simulability of expectation values rather than samples. IQP circuits sit in this gap: relevant expectation values can be computed efficiently classically, while sampling from typical IQP circuits is believed to be classically hard. This motivates the possibility of training an IQP model with an MMD objective (classically or quantumly) while targeting a setting in which sampling is performed on quantum hardware.
A key subtlety is that standard IQP hardness results are formulated in terms of small total-variation (i.e., $\ell_1$) distance, whereas our training guarantees are for the MMD loss. It therefore remains unclear whether a classical simulator could approximate samples from an IQP distribution to within the MMD accuracy achieved by training.

More fundamentally, an important next step is to study trainability beyond the first few training steps since optimization may quickly leave regions where initialization-based guarantees are informative. Analytically understanding gradient-variance degradation during optimization, or indeed to dive beyond our average case analyses and understand training gradient flows, remains an important but hard open problem. However, we hope the mechanism-based viewpoint developed here could provide intuitions to help with future developments. 

\medskip

\paragraph*{\textbf{Added note:}}
While preparing this manuscript, the independent work \emph{Characterizing Trainability of Instantaneous Quantum Polynomial Circuit Born Machines} \cite{shen2026characterizing} appeared. Despite considerable overlap at the early part regarding the exponential concentration under the full angle random initialization, the two works are complementary in that Ref.~\cite{shen2026characterizing} studies gradient variance for gaussian near-identity initializations, including average-case guarantees over unstructured target ensembles, whereas here we study the loss variance and derive curvature-based guarantees around multiple initialization centers including data-dependent guarantees under certain structural assumptions on the target.

\section{Acknowledgments}
R.P. acknowledges the support of the SNF Quantum Flagship Replacement Scheme (grant No. 215933). EA acknowledges the support of the Government of Spain (Severo Ochoa CEX2019-000910-S, FUNQIP and European Union NextGenerationEU PRTR-C17.I1), Fundació Cellex, Fundació Mir-Puig, Generalitat de Catalunya (CERCA program) and European Union (PASQuanS2.1, 101113690). EA is a fellow of Eurecat's \emph{Vicente López} PhD grant program. Z.H. acknowledges support from the Sandoz Family Foundation-Monique de Meuron program for Academic Promotion. ST acknowledges Exchange Faculty Travel Grant: TG168033 from Chulalongkorn University, as well as funding from National Research Council of Thailand
(NRCT) [grant number N42A680126]. ST further acknowledges Thailand Science research and Innovation Fund Chulalongkorn University (IND\_FF\_69\_258\_2300\_062).

\bibliography{quantum.bib,IQP-trainability.bib}

\begin{thebibliography}{108}%
\makeatletter
\providecommand \@ifxundefined [1]{%
 \@ifx{#1\undefined}
}%
\providecommand \@ifnum [1]{%
 \ifnum #1\expandafter \@firstoftwo
 \else \expandafter \@secondoftwo
 \fi
}%
\providecommand \@ifx [1]{%
 \ifx #1\expandafter \@firstoftwo
 \else \expandafter \@secondoftwo
 \fi
}%
\providecommand \natexlab [1]{#1}%
\providecommand \enquote  [1]{``#1''}%
\providecommand \bibnamefont  [1]{#1}%
\providecommand \bibfnamefont [1]{#1}%
\providecommand \citenamefont [1]{#1}%
\providecommand \href@noop [0]{\@secondoftwo}%
\providecommand \href [0]{\begingroup \@sanitize@url \@href}%
\providecommand \@href[1]{\@@startlink{#1}\@@href}%
\providecommand \@@href[1]{\endgroup#1\@@endlink}%
\providecommand \@sanitize@url [0]{\catcode `\\12\catcode `\$12\catcode `\&12\catcode `\#12\catcode `\^12\catcode `\_12\catcode `\%12\relax}%
\providecommand \@@startlink[1]{}%
\providecommand \@@endlink[0]{}%
\providecommand \url  [0]{\begingroup\@sanitize@url \@url }%
\providecommand \@url [1]{\endgroup\@href {#1}{\urlprefix }}%
\providecommand \urlprefix  [0]{URL }%
\providecommand \Eprint [0]{\href }%
\providecommand \doibase [0]{https://doi.org/}%
\providecommand \selectlanguage [0]{\@gobble}%
\providecommand \bibinfo  [0]{\@secondoftwo}%
\providecommand \bibfield  [0]{\@secondoftwo}%
\providecommand \translation [1]{[#1]}%
\providecommand \BibitemOpen [0]{}%
\providecommand \bibitemStop [0]{}%
\providecommand \bibitemNoStop [0]{.\EOS\space}%
\providecommand \EOS [0]{\spacefactor3000\relax}%
\providecommand \BibitemShut  [1]{\csname bibitem#1\endcsname}%
\let\auto@bib@innerbib\@empty
\bibitem [{\citenamefont {Sweke}\ \emph {et~al.}(2021)\citenamefont {Sweke}, \citenamefont {Seifert}, \citenamefont {Hangleiter},\ and\ \citenamefont {Eisert}}]{sweke2021quantum}%
  \BibitemOpen
  \bibfield  {author} {\bibinfo {author} {\bibfnamefont {R.}~\bibnamefont {Sweke}}, \bibinfo {author} {\bibfnamefont {J.-P.}\ \bibnamefont {Seifert}}, \bibinfo {author} {\bibfnamefont {D.}~\bibnamefont {Hangleiter}},\ and\ \bibinfo {author} {\bibfnamefont {J.}~\bibnamefont {Eisert}},\ }\bibfield  {title} {\bibinfo {title} {On the quantum versus classical learnability of discrete distributions},\ }\href {https://doi.org/10.22331/q-2021-03-23-417} {\bibfield  {journal} {\bibinfo  {journal} {Quantum}\ }\textbf {\bibinfo {volume} {5}},\ \bibinfo {pages} {417} (\bibinfo {year} {2021})}\BibitemShut {NoStop}%
\bibitem [{\citenamefont {Gao}\ \emph {et~al.}(2022)\citenamefont {Gao}, \citenamefont {Anschuetz}, \citenamefont {Wang}, \citenamefont {Cirac},\ and\ \citenamefont {Lukin}}]{gao2022enhancing}%
  \BibitemOpen
  \bibfield  {author} {\bibinfo {author} {\bibfnamefont {X.}~\bibnamefont {Gao}}, \bibinfo {author} {\bibfnamefont {E.~R.}\ \bibnamefont {Anschuetz}}, \bibinfo {author} {\bibfnamefont {S.-T.}\ \bibnamefont {Wang}}, \bibinfo {author} {\bibfnamefont {J.~I.}\ \bibnamefont {Cirac}},\ and\ \bibinfo {author} {\bibfnamefont {M.~D.}\ \bibnamefont {Lukin}},\ }\bibfield  {title} {\bibinfo {title} {Enhancing generative models via quantum correlations},\ }\href {https://doi.org/10.1103/PhysRevX.12.021037} {\bibfield  {journal} {\bibinfo  {journal} {Phys. Rev. X}\ }\textbf {\bibinfo {volume} {12}},\ \bibinfo {pages} {021037} (\bibinfo {year} {2022})}\BibitemShut {NoStop}%
\bibitem [{\citenamefont {Huang}\ \emph {et~al.}(2025)\citenamefont {Huang}, \citenamefont {Broughton}, \citenamefont {Eassa}, \citenamefont {Neven}, \citenamefont {Babbush},\ and\ \citenamefont {McClean}}]{huang2025generative}%
  \BibitemOpen
  \bibfield  {author} {\bibinfo {author} {\bibfnamefont {H.-Y.}\ \bibnamefont {Huang}}, \bibinfo {author} {\bibfnamefont {M.}~\bibnamefont {Broughton}}, \bibinfo {author} {\bibfnamefont {N.}~\bibnamefont {Eassa}}, \bibinfo {author} {\bibfnamefont {H.}~\bibnamefont {Neven}}, \bibinfo {author} {\bibfnamefont {R.}~\bibnamefont {Babbush}},\ and\ \bibinfo {author} {\bibfnamefont {J.~R.}\ \bibnamefont {McClean}},\ }\bibfield  {title} {\bibinfo {title} {Generative quantum advantage for classical and quantum problems},\ }\bibfield  {journal} {\bibinfo  {journal} {arXiv preprint arXiv:2509.09033}\ }\href {https://doi.org/10.48550/arXiv.2509.09033} {10.48550/arXiv.2509.09033} (\bibinfo {year} {2025})\BibitemShut {NoStop}%
\bibitem [{\citenamefont {Hibat-Allah}\ \emph {et~al.}(2024)\citenamefont {Hibat-Allah}, \citenamefont {Mauri}, \citenamefont {Carrasquilla},\ and\ \citenamefont {Perdomo-Ortiz}}]{hibat2024framework}%
  \BibitemOpen
  \bibfield  {author} {\bibinfo {author} {\bibfnamefont {M.}~\bibnamefont {Hibat-Allah}}, \bibinfo {author} {\bibfnamefont {M.}~\bibnamefont {Mauri}}, \bibinfo {author} {\bibfnamefont {J.}~\bibnamefont {Carrasquilla}},\ and\ \bibinfo {author} {\bibfnamefont {A.}~\bibnamefont {Perdomo-Ortiz}},\ }\bibfield  {title} {\bibinfo {title} {A framework for demonstrating practical quantum advantage: comparing quantum against classical generative models},\ }\href {https://www.nature.com/articles/s42005-024-01552-6} {\bibfield  {journal} {\bibinfo  {journal} {Communications Physics}\ }\textbf {\bibinfo {volume} {7}},\ \bibinfo {pages} {68} (\bibinfo {year} {2024})}\BibitemShut {NoStop}%
\bibitem [{\citenamefont {Coyle}\ \emph {et~al.}(2020)\citenamefont {Coyle}, \citenamefont {Mills}, \citenamefont {Danos},\ and\ \citenamefont {Kashefi}}]{coyle2020born}%
  \BibitemOpen
  \bibfield  {author} {\bibinfo {author} {\bibfnamefont {B.}~\bibnamefont {Coyle}}, \bibinfo {author} {\bibfnamefont {D.}~\bibnamefont {Mills}}, \bibinfo {author} {\bibfnamefont {V.}~\bibnamefont {Danos}},\ and\ \bibinfo {author} {\bibfnamefont {E.}~\bibnamefont {Kashefi}},\ }\bibfield  {title} {\bibinfo {title} {The born supremacy: quantum advantage and training of an ising born machine},\ }\href {https://doi.org/10.1038/s41534-020-00288-9} {\bibfield  {journal} {\bibinfo  {journal} {npj Quantum Information}\ }\textbf {\bibinfo {volume} {6}},\ \bibinfo {pages} {60} (\bibinfo {year} {2020})}\BibitemShut {NoStop}%
\bibitem [{\citenamefont {Benedetti}\ \emph {et~al.}(2019)\citenamefont {Benedetti}, \citenamefont {Garcia-Pintos}, \citenamefont {Perdomo}, \citenamefont {Leyton-Ortega}, \citenamefont {Nam},\ and\ \citenamefont {Perdomo-Ortiz}}]{benedetti2019generative}%
  \BibitemOpen
  \bibfield  {author} {\bibinfo {author} {\bibfnamefont {M.}~\bibnamefont {Benedetti}}, \bibinfo {author} {\bibfnamefont {D.}~\bibnamefont {Garcia-Pintos}}, \bibinfo {author} {\bibfnamefont {O.}~\bibnamefont {Perdomo}}, \bibinfo {author} {\bibfnamefont {V.}~\bibnamefont {Leyton-Ortega}}, \bibinfo {author} {\bibfnamefont {Y.}~\bibnamefont {Nam}},\ and\ \bibinfo {author} {\bibfnamefont {A.}~\bibnamefont {Perdomo-Ortiz}},\ }\bibfield  {title} {\bibinfo {title} {A generative modeling approach for benchmarking and training shallow quantum circuits},\ }\href {https://doi.org/10.1038/s41534-019-0157-8} {\bibfield  {journal} {\bibinfo  {journal} {npj Quantum Information}\ }\textbf {\bibinfo {volume} {5}},\ \bibinfo {pages} {45} (\bibinfo {year} {2019})}\BibitemShut {NoStop}%
\bibitem [{\citenamefont {Liu}\ and\ \citenamefont {Wang}(2018)}]{liu2018differentiable}%
  \BibitemOpen
  \bibfield  {author} {\bibinfo {author} {\bibfnamefont {J.-G.}\ \bibnamefont {Liu}}\ and\ \bibinfo {author} {\bibfnamefont {L.}~\bibnamefont {Wang}},\ }\bibfield  {title} {\bibinfo {title} {Differentiable learning of quantum circuit born machines},\ }\href {https://doi.org/10.1103/PhysRevA.98.062324} {\bibfield  {journal} {\bibinfo  {journal} {Phys. Rev. A}\ }\textbf {\bibinfo {volume} {98}},\ \bibinfo {pages} {062324} (\bibinfo {year} {2018})}\BibitemShut {NoStop}%
\bibitem [{\citenamefont {Rudolph}\ \emph {et~al.}(2024)\citenamefont {Rudolph}, \citenamefont {Lerch}, \citenamefont {Thanasilp}, \citenamefont {Kiss}, \citenamefont {Shaya}, \citenamefont {Vallecorsa}, \citenamefont {Grossi},\ and\ \citenamefont {Holmes}}]{rudolph2023trainability}%
  \BibitemOpen
  \bibfield  {author} {\bibinfo {author} {\bibfnamefont {M.~S.}\ \bibnamefont {Rudolph}}, \bibinfo {author} {\bibfnamefont {S.}~\bibnamefont {Lerch}}, \bibinfo {author} {\bibfnamefont {S.}~\bibnamefont {Thanasilp}}, \bibinfo {author} {\bibfnamefont {O.}~\bibnamefont {Kiss}}, \bibinfo {author} {\bibfnamefont {O.}~\bibnamefont {Shaya}}, \bibinfo {author} {\bibfnamefont {S.}~\bibnamefont {Vallecorsa}}, \bibinfo {author} {\bibfnamefont {M.}~\bibnamefont {Grossi}},\ and\ \bibinfo {author} {\bibfnamefont {Z.}~\bibnamefont {Holmes}},\ }\bibfield  {title} {\bibinfo {title} {Trainability barriers and opportunities in quantum generative modeling},\ }\href {https://doi.org/https://doi.org/10.1038/s41534-024-00902-0} {\bibfield  {journal} {\bibinfo  {journal} {npj Quantum Information}\ }\textbf {\bibinfo {volume} {10}},\ \bibinfo {pages} {116} (\bibinfo {year} {2024})}\BibitemShut {NoStop}%
\bibitem [{\citenamefont {Cerezo}\ \emph {et~al.}(2025)\citenamefont {Cerezo}, \citenamefont {Larocca}, \citenamefont {Garc{\'\i}a-Mart{\'\i}n}, \citenamefont {Diaz}, \citenamefont {Braccia}, \citenamefont {Fontana}, \citenamefont {Rudolph}, \citenamefont {Bermejo}, \citenamefont {Ijaz}, \citenamefont {Thanasilp} \emph {et~al.}}]{cerezo2023does}%
  \BibitemOpen
  \bibfield  {author} {\bibinfo {author} {\bibfnamefont {M.}~\bibnamefont {Cerezo}}, \bibinfo {author} {\bibfnamefont {M.}~\bibnamefont {Larocca}}, \bibinfo {author} {\bibfnamefont {D.}~\bibnamefont {Garc{\'\i}a-Mart{\'\i}n}}, \bibinfo {author} {\bibfnamefont {N.~L.}\ \bibnamefont {Diaz}}, \bibinfo {author} {\bibfnamefont {P.}~\bibnamefont {Braccia}}, \bibinfo {author} {\bibfnamefont {E.}~\bibnamefont {Fontana}}, \bibinfo {author} {\bibfnamefont {M.~S.}\ \bibnamefont {Rudolph}}, \bibinfo {author} {\bibfnamefont {P.}~\bibnamefont {Bermejo}}, \bibinfo {author} {\bibfnamefont {A.}~\bibnamefont {Ijaz}}, \bibinfo {author} {\bibfnamefont {S.}~\bibnamefont {Thanasilp}}, \emph {et~al.},\ }\bibfield  {title} {\bibinfo {title} {Does provable absence of barren plateaus imply classical simulability?},\ }\href {https://doi.org/10.1038/s41467-025-63099-6} {\bibfield  {journal} {\bibinfo  {journal} {Nature Communications}\ }\textbf {\bibinfo {volume} {16}},\ \bibinfo {pages} {7907} (\bibinfo {year} {2025})}\BibitemShut
  {NoStop}%
\bibitem [{\citenamefont {Bak{\'o}}\ \emph {et~al.}(2025)\citenamefont {Bak{\'o}}, \citenamefont {Kolarovszki},\ and\ \citenamefont {Zimboras}}]{bako2025fermionic}%
  \BibitemOpen
  \bibfield  {author} {\bibinfo {author} {\bibfnamefont {B.}~\bibnamefont {Bak{\'o}}}, \bibinfo {author} {\bibfnamefont {Z.}~\bibnamefont {Kolarovszki}},\ and\ \bibinfo {author} {\bibfnamefont {Z.}~\bibnamefont {Zimboras}},\ }\bibfield  {title} {\bibinfo {title} {Fermionic born machines: Classical training of quantum generative models based on fermion sampling},\ }\bibfield  {journal} {\bibinfo  {journal} {arXiv preprint arXiv:2511.13844}\ }\href {https://doi.org/https://doi.org/10.48550/arXiv.2511.13844} {https://doi.org/10.48550/arXiv.2511.13844} (\bibinfo {year} {2025})\BibitemShut {NoStop}%
\bibitem [{\citenamefont {Recio-Armengol}\ \emph {et~al.}(2025)\citenamefont {Recio-Armengol}, \citenamefont {Ahmed},\ and\ \citenamefont {Bowles}}]{recio2025train}%
  \BibitemOpen
  \bibfield  {author} {\bibinfo {author} {\bibfnamefont {E.}~\bibnamefont {Recio-Armengol}}, \bibinfo {author} {\bibfnamefont {S.}~\bibnamefont {Ahmed}},\ and\ \bibinfo {author} {\bibfnamefont {J.}~\bibnamefont {Bowles}},\ }\bibfield  {title} {\bibinfo {title} {Train on classical, deploy on quantum: scaling generative quantum machine learning to a thousand qubits},\ }\bibfield  {journal} {\bibinfo  {journal} {arXiv preprint arXiv:2503.02934}\ }\href {https://doi.org/10.48550/arXiv.2503.02934} {10.48550/arXiv.2503.02934} (\bibinfo {year} {2025})\BibitemShut {NoStop}%
\bibitem [{\citenamefont {Herrero-Gonzalez}\ \emph {et~al.}(2025)\citenamefont {Herrero-Gonzalez}, \citenamefont {Coyle}, \citenamefont {McDowall}, \citenamefont {Grassie}, \citenamefont {Beentjes}, \citenamefont {Khamseh},\ and\ \citenamefont {Kashefi}}]{herrero2025born}%
  \BibitemOpen
  \bibfield  {author} {\bibinfo {author} {\bibfnamefont {M.}~\bibnamefont {Herrero-Gonzalez}}, \bibinfo {author} {\bibfnamefont {B.}~\bibnamefont {Coyle}}, \bibinfo {author} {\bibfnamefont {K.}~\bibnamefont {McDowall}}, \bibinfo {author} {\bibfnamefont {R.}~\bibnamefont {Grassie}}, \bibinfo {author} {\bibfnamefont {S.}~\bibnamefont {Beentjes}}, \bibinfo {author} {\bibfnamefont {A.}~\bibnamefont {Khamseh}},\ and\ \bibinfo {author} {\bibfnamefont {E.}~\bibnamefont {Kashefi}},\ }\bibfield  {title} {\bibinfo {title} {The born ultimatum: Conditions for classical surrogation of quantum generative models with correlators},\ }\bibfield  {journal} {\bibinfo  {journal} {arXiv preprint arXiv:2511.01845}\ }\href {https://doi.org/10.48550/arXiv.2511.01845} {10.48550/arXiv.2511.01845} (\bibinfo {year} {2025})\BibitemShut {NoStop}%
\bibitem [{\citenamefont {Kasture}\ \emph {et~al.}(2023)\citenamefont {Kasture}, \citenamefont {Kyriienko},\ and\ \citenamefont {Elfving}}]{kasture2023protocols}%
  \BibitemOpen
  \bibfield  {author} {\bibinfo {author} {\bibfnamefont {S.}~\bibnamefont {Kasture}}, \bibinfo {author} {\bibfnamefont {O.}~\bibnamefont {Kyriienko}},\ and\ \bibinfo {author} {\bibfnamefont {V.~E.}\ \bibnamefont {Elfving}},\ }\bibfield  {title} {\bibinfo {title} {Protocols for classically training quantum generative models on probability distributions},\ }\href {https://journals.aps.org/pra/abstract/10.1103/PhysRevA.108.042406} {\bibfield  {journal} {\bibinfo  {journal} {Physical Review A}\ }\textbf {\bibinfo {volume} {108}},\ \bibinfo {pages} {042406} (\bibinfo {year} {2023})}\BibitemShut {NoStop}%
\bibitem [{\citenamefont {Kurkin}\ \emph {et~al.}(2025)\citenamefont {Kurkin}, \citenamefont {Shen}, \citenamefont {Pielawa}, \citenamefont {Wang},\ and\ \citenamefont {Dunjko}}]{kurkin2025universality}%
  \BibitemOpen
  \bibfield  {author} {\bibinfo {author} {\bibfnamefont {A.}~\bibnamefont {Kurkin}}, \bibinfo {author} {\bibfnamefont {K.}~\bibnamefont {Shen}}, \bibinfo {author} {\bibfnamefont {S.}~\bibnamefont {Pielawa}}, \bibinfo {author} {\bibfnamefont {H.}~\bibnamefont {Wang}},\ and\ \bibinfo {author} {\bibfnamefont {V.}~\bibnamefont {Dunjko}},\ }\bibfield  {title} {\bibinfo {title} {Universality and kernel-adaptive training for classically trained, quantum-deployed generative models},\ }\href {https://arxiv.org/abs/2510.08476} {\bibfield  {journal} {\bibinfo  {journal} {arXiv preprint arXiv:2510.08476}\ } (\bibinfo {year} {2025})}\BibitemShut {NoStop}%
\bibitem [{\citenamefont {Larocca}\ \emph {et~al.}(2025)\citenamefont {Larocca}, \citenamefont {Thanasilp}, \citenamefont {Wang}, \citenamefont {Sharma}, \citenamefont {Biamonte}, \citenamefont {Coles}, \citenamefont {Cincio}, \citenamefont {McClean}, \citenamefont {Holmes},\ and\ \citenamefont {Cerezo}}]{larocca2024review}%
  \BibitemOpen
  \bibfield  {author} {\bibinfo {author} {\bibfnamefont {M.}~\bibnamefont {Larocca}}, \bibinfo {author} {\bibfnamefont {S.}~\bibnamefont {Thanasilp}}, \bibinfo {author} {\bibfnamefont {S.}~\bibnamefont {Wang}}, \bibinfo {author} {\bibfnamefont {K.}~\bibnamefont {Sharma}}, \bibinfo {author} {\bibfnamefont {J.}~\bibnamefont {Biamonte}}, \bibinfo {author} {\bibfnamefont {P.~J.}\ \bibnamefont {Coles}}, \bibinfo {author} {\bibfnamefont {L.}~\bibnamefont {Cincio}}, \bibinfo {author} {\bibfnamefont {J.~R.}\ \bibnamefont {McClean}}, \bibinfo {author} {\bibfnamefont {Z.}~\bibnamefont {Holmes}},\ and\ \bibinfo {author} {\bibfnamefont {M.}~\bibnamefont {Cerezo}},\ }\bibfield  {title} {\bibinfo {title} {Barren plateaus in variational quantum computing},\ }\href {https://doi.org/10.1038/s42254-025-00813-9} {\bibfield  {journal} {\bibinfo  {journal} {Nature Reviews Physics}\ }\textbf {\bibinfo {volume} {7}},\ \bibinfo {pages} {174–189} (\bibinfo {year} {2025})}\BibitemShut {NoStop}%
\bibitem [{\citenamefont {Anschuetz}\ and\ \citenamefont {Kiani}(2022)}]{anschuetz2022quantum}%
  \BibitemOpen
  \bibfield  {author} {\bibinfo {author} {\bibfnamefont {E.~R.}\ \bibnamefont {Anschuetz}}\ and\ \bibinfo {author} {\bibfnamefont {B.~T.}\ \bibnamefont {Kiani}},\ }\bibfield  {title} {\bibinfo {title} {Quantum variational algorithms are swamped with traps},\ }\href {https://doi.org/10.1038/s41467-022-35364-5} {\bibfield  {journal} {\bibinfo  {journal} {Nature Communications}\ }\textbf {\bibinfo {volume} {13}},\ \bibinfo {pages} {7760} (\bibinfo {year} {2022})}\BibitemShut {NoStop}%
\bibitem [{\citenamefont {Abbas}\ \emph {et~al.}(2023)\citenamefont {Abbas}, \citenamefont {King}, \citenamefont {Huang}, \citenamefont {Huggins}, \citenamefont {Movassagh}, \citenamefont {Gilboa},\ and\ \citenamefont {McClean}}]{abbas2023quantum}%
  \BibitemOpen
  \bibfield  {author} {\bibinfo {author} {\bibfnamefont {A.}~\bibnamefont {Abbas}}, \bibinfo {author} {\bibfnamefont {R.}~\bibnamefont {King}}, \bibinfo {author} {\bibfnamefont {H.-Y.}\ \bibnamefont {Huang}}, \bibinfo {author} {\bibfnamefont {W.~J.}\ \bibnamefont {Huggins}}, \bibinfo {author} {\bibfnamefont {R.}~\bibnamefont {Movassagh}}, \bibinfo {author} {\bibfnamefont {D.}~\bibnamefont {Gilboa}},\ and\ \bibinfo {author} {\bibfnamefont {J.~R.}\ \bibnamefont {McClean}},\ }\bibfield  {title} {\bibinfo {title} {On quantum backpropagation, information reuse, and cheating measurement collapse},\ }\href {https://papers.nips.cc/paper_files/paper/2023/hash/8c3caae2f725c8e2a55ecd600563d172-Abstract-Conference.html} {\bibfield  {journal} {\bibinfo  {journal} {Advances in Neural Information Processing Systems}\ }\textbf {\bibinfo {volume} {36}},\ \bibinfo {pages} {44792} (\bibinfo {year} {2023})}\BibitemShut {NoStop}%
\bibitem [{\citenamefont {Shen}\ \emph {et~al.}(2026)\citenamefont {Shen}, \citenamefont {Pielawa}, \citenamefont {Dunjko},\ and\ \citenamefont {Wang}}]{shen2026characterizing}%
  \BibitemOpen
  \bibfield  {author} {\bibinfo {author} {\bibfnamefont {K.}~\bibnamefont {Shen}}, \bibinfo {author} {\bibfnamefont {S.}~\bibnamefont {Pielawa}}, \bibinfo {author} {\bibfnamefont {V.}~\bibnamefont {Dunjko}},\ and\ \bibinfo {author} {\bibfnamefont {H.}~\bibnamefont {Wang}},\ }\bibfield  {title} {\bibinfo {title} {Characterizing trainability of instantaneous quantum polynomial circuit born machines},\ }\bibfield  {journal} {\bibinfo  {journal} {arXiv preprint arXiv:2602.11042}\ }\href {https://doi.org/https://doi.org/10.48550/arXiv.2602.11042} {https://doi.org/10.48550/arXiv.2602.11042} (\bibinfo {year} {2026})\BibitemShut {NoStop}%
\bibitem [{\citenamefont {McClean}\ \emph {et~al.}(2018)\citenamefont {McClean}, \citenamefont {Boixo}, \citenamefont {Smelyanskiy}, \citenamefont {Babbush},\ and\ \citenamefont {Neven}}]{mcclean2018barren}%
  \BibitemOpen
  \bibfield  {author} {\bibinfo {author} {\bibfnamefont {J.~R.}\ \bibnamefont {McClean}}, \bibinfo {author} {\bibfnamefont {S.}~\bibnamefont {Boixo}}, \bibinfo {author} {\bibfnamefont {V.~N.}\ \bibnamefont {Smelyanskiy}}, \bibinfo {author} {\bibfnamefont {R.}~\bibnamefont {Babbush}},\ and\ \bibinfo {author} {\bibfnamefont {H.}~\bibnamefont {Neven}},\ }\bibfield  {title} {\bibinfo {title} {Barren plateaus in quantum neural network training landscapes},\ }\href {https://doi.org/10.1038/s41467-018-07090-4} {\bibfield  {journal} {\bibinfo  {journal} {Nature {C}ommunications}\ }\textbf {\bibinfo {volume} {9}},\ \bibinfo {pages} {1} (\bibinfo {year} {2018})}\BibitemShut {NoStop}%
\bibitem [{\citenamefont {Herbst}\ \emph {et~al.}(2025)\citenamefont {Herbst}, \citenamefont {Brandi{\'c}},\ and\ \citenamefont {P{\'e}rez-Salinas}}]{herbst2025limits}%
  \BibitemOpen
  \bibfield  {author} {\bibinfo {author} {\bibfnamefont {S.}~\bibnamefont {Herbst}}, \bibinfo {author} {\bibfnamefont {I.}~\bibnamefont {Brandi{\'c}}},\ and\ \bibinfo {author} {\bibfnamefont {A.}~\bibnamefont {P{\'e}rez-Salinas}},\ }\bibfield  {title} {\bibinfo {title} {Limits of quantum generative models with classical sampling hardness},\ }\href {https://arxiv.org/abs/2512.24801} {\bibfield  {journal} {\bibinfo  {journal} {arXiv preprint arXiv:2512.24801}\ } (\bibinfo {year} {2025})}\BibitemShut {NoStop}%
\bibitem [{\citenamefont {Zou}\ \emph {et~al.}(2025{\natexlab{a}})\citenamefont {Zou}, \citenamefont {Duan}, \citenamefont {Fleming}, \citenamefont {Liu}, \citenamefont {Kompella}, \citenamefont {Ren},\ and\ \citenamefont {Xu}}]{zou2025conquer}%
  \BibitemOpen
  \bibfield  {author} {\bibinfo {author} {\bibfnamefont {X.}~\bibnamefont {Zou}}, \bibinfo {author} {\bibfnamefont {S.}~\bibnamefont {Duan}}, \bibinfo {author} {\bibfnamefont {C.}~\bibnamefont {Fleming}}, \bibinfo {author} {\bibfnamefont {G.}~\bibnamefont {Liu}}, \bibinfo {author} {\bibfnamefont {R.~R.}\ \bibnamefont {Kompella}}, \bibinfo {author} {\bibfnamefont {S.}~\bibnamefont {Ren}},\ and\ \bibinfo {author} {\bibfnamefont {X.}~\bibnamefont {Xu}},\ }\bibfield  {title} {\bibinfo {title} {Conquer: Modular architectures for control and bias mitigation in iqp quantum generative models},\ }\href {https://arxiv.org/abs/2509.22551} {\bibfield  {journal} {\bibinfo  {journal} {arXiv preprint arXiv:2509.22551}\ } (\bibinfo {year} {2025}{\natexlab{a}})}\BibitemShut {NoStop}%
\bibitem [{\citenamefont {Majumder}\ \emph {et~al.}(2024)\citenamefont {Majumder}, \citenamefont {Krumm}, \citenamefont {Radkohl}, \citenamefont {Fiderer}, \citenamefont {Nautrup}, \citenamefont {Jerbi},\ and\ \citenamefont {Briegel}}]{majumder2024variational}%
  \BibitemOpen
  \bibfield  {author} {\bibinfo {author} {\bibfnamefont {A.}~\bibnamefont {Majumder}}, \bibinfo {author} {\bibfnamefont {M.}~\bibnamefont {Krumm}}, \bibinfo {author} {\bibfnamefont {T.}~\bibnamefont {Radkohl}}, \bibinfo {author} {\bibfnamefont {L.~J.}\ \bibnamefont {Fiderer}}, \bibinfo {author} {\bibfnamefont {H.~P.}\ \bibnamefont {Nautrup}}, \bibinfo {author} {\bibfnamefont {S.}~\bibnamefont {Jerbi}},\ and\ \bibinfo {author} {\bibfnamefont {H.~J.}\ \bibnamefont {Briegel}},\ }\bibfield  {title} {\bibinfo {title} {Variational measurement-based quantum computation for generative modeling},\ }\href {https://journals.aps.org/pra/abstract/10.1103/PhysRevA.110.062616} {\bibfield  {journal} {\bibinfo  {journal} {Physical Review A}\ }\textbf {\bibinfo {volume} {110}},\ \bibinfo {pages} {062616} (\bibinfo {year} {2024})}\BibitemShut {NoStop}%
\bibitem [{\citenamefont {Van~den Nest}(2009)}]{nest2009simulating}%
  \BibitemOpen
  \bibfield  {author} {\bibinfo {author} {\bibfnamefont {M.}~\bibnamefont {Van~den Nest}},\ }\bibfield  {title} {\bibinfo {title} {Simulating quantum computers with probabilistic methods},\ }\bibfield  {journal} {\bibinfo  {journal} {arXiv preprint arXiv:0911.1624}\ }\href {https://doi.org/https://doi.org/10.48550/arXiv.0911.1624} {https://doi.org/10.48550/arXiv.0911.1624} (\bibinfo {year} {2009})\BibitemShut {NoStop}%
\bibitem [{\citenamefont {Shepherd}\ and\ \citenamefont {Bremner}(2009)}]{shepherd2009temporally}%
  \BibitemOpen
  \bibfield  {author} {\bibinfo {author} {\bibfnamefont {D.}~\bibnamefont {Shepherd}}\ and\ \bibinfo {author} {\bibfnamefont {M.~J.}\ \bibnamefont {Bremner}},\ }\bibfield  {title} {\bibinfo {title} {Temporally unstructured quantum computation},\ }\href {https://doi.org/https://doi.org/10.1098/rspa.2008.0443} {\bibfield  {journal} {\bibinfo  {journal} {Proceedings of the Royal Society A: Mathematical, Physical and Engineering Sciences}\ }\textbf {\bibinfo {volume} {465}},\ \bibinfo {pages} {1413} (\bibinfo {year} {2009})}\BibitemShut {NoStop}%
\bibitem [{\citenamefont {Bremner}\ \emph {et~al.}(2011)\citenamefont {Bremner}, \citenamefont {Jozsa},\ and\ \citenamefont {Shepherd}}]{bremner2011classical}%
  \BibitemOpen
  \bibfield  {author} {\bibinfo {author} {\bibfnamefont {M.~J.}\ \bibnamefont {Bremner}}, \bibinfo {author} {\bibfnamefont {R.}~\bibnamefont {Jozsa}},\ and\ \bibinfo {author} {\bibfnamefont {D.~J.}\ \bibnamefont {Shepherd}},\ }\bibfield  {title} {\bibinfo {title} {Classical simulation of commuting quantum computations implies collapse of the polynomial hierarchy},\ }\href {https://royalsocietypublishing.org/doi/abs/10.1098/rspa.2010.0301} {\bibfield  {journal} {\bibinfo  {journal} {Proceedings of the Royal Society A: Mathematical, Physical and Engineering Sciences}\ }\textbf {\bibinfo {volume} {467}},\ \bibinfo {pages} {459} (\bibinfo {year} {2011})}\BibitemShut {NoStop}%
\bibitem [{\citenamefont {Bremner}\ \emph {et~al.}(2016)\citenamefont {Bremner}, \citenamefont {Montanaro},\ and\ \citenamefont {Shepherd}}]{bremner2016average}%
  \BibitemOpen
  \bibfield  {author} {\bibinfo {author} {\bibfnamefont {M.~J.}\ \bibnamefont {Bremner}}, \bibinfo {author} {\bibfnamefont {A.}~\bibnamefont {Montanaro}},\ and\ \bibinfo {author} {\bibfnamefont {D.~J.}\ \bibnamefont {Shepherd}},\ }\bibfield  {title} {\bibinfo {title} {Average-case complexity versus approximate simulation of commuting quantum computations},\ }\href {https://doi.org/10.1103/PhysRevLett.117.080501} {\bibfield  {journal} {\bibinfo  {journal} {Physical review letters}\ }\textbf {\bibinfo {volume} {117}},\ \bibinfo {pages} {080501} (\bibinfo {year} {2016})}\BibitemShut {NoStop}%
\bibitem [{\citenamefont {Hangleiter}\ and\ \citenamefont {Eisert}(2023)}]{hangleiter2023computational}%
  \BibitemOpen
  \bibfield  {author} {\bibinfo {author} {\bibfnamefont {D.}~\bibnamefont {Hangleiter}}\ and\ \bibinfo {author} {\bibfnamefont {J.}~\bibnamefont {Eisert}},\ }\bibfield  {title} {\bibinfo {title} {Computational advantage of quantum random sampling},\ }\href {https://journals.aps.org/rmp/abstract/10.1103/RevModPhys.95.035001} {\bibfield  {journal} {\bibinfo  {journal} {Reviews of Modern Physics}\ }\textbf {\bibinfo {volume} {95}},\ \bibinfo {pages} {035001} (\bibinfo {year} {2023})}\BibitemShut {NoStop}%
\bibitem [{\citenamefont {Bergholm}\ \emph {et~al.}(2018)\citenamefont {Bergholm}, \citenamefont {Izaac}, \citenamefont {Schuld}, \citenamefont {Gogolin}, \citenamefont {Alam}, \citenamefont {Ahmed}, \citenamefont {Arrazola}, \citenamefont {Blank}, \citenamefont {Delgado}, \citenamefont {Jahangiri} \emph {et~al.}}]{bergholm2018pennylane}%
  \BibitemOpen
  \bibfield  {author} {\bibinfo {author} {\bibfnamefont {V.}~\bibnamefont {Bergholm}}, \bibinfo {author} {\bibfnamefont {J.}~\bibnamefont {Izaac}}, \bibinfo {author} {\bibfnamefont {M.}~\bibnamefont {Schuld}}, \bibinfo {author} {\bibfnamefont {C.}~\bibnamefont {Gogolin}}, \bibinfo {author} {\bibfnamefont {M.~S.}\ \bibnamefont {Alam}}, \bibinfo {author} {\bibfnamefont {S.}~\bibnamefont {Ahmed}}, \bibinfo {author} {\bibfnamefont {J.~M.}\ \bibnamefont {Arrazola}}, \bibinfo {author} {\bibfnamefont {C.}~\bibnamefont {Blank}}, \bibinfo {author} {\bibfnamefont {A.}~\bibnamefont {Delgado}}, \bibinfo {author} {\bibfnamefont {S.}~\bibnamefont {Jahangiri}}, \emph {et~al.},\ }\bibfield  {title} {\bibinfo {title} {Pennylane: Automatic differentiation of hybrid quantum-classical computations},\ }\href {https://arxiv.org/abs/1811.04968} {\bibfield  {journal} {\bibinfo  {journal} {arXiv preprint arXiv:1811.04968}\ } (\bibinfo {year} {2018})}\BibitemShut {NoStop}%
\bibitem [{\citenamefont {Recio-Armengol}\ and\ \citenamefont {Bowles}(2025)}]{recio2025iqpopt}%
  \BibitemOpen
  \bibfield  {author} {\bibinfo {author} {\bibfnamefont {E.}~\bibnamefont {Recio-Armengol}}\ and\ \bibinfo {author} {\bibfnamefont {J.}~\bibnamefont {Bowles}},\ }\bibfield  {title} {\bibinfo {title} {Iqpopt: Fast optimization of instantaneous quantum polynomial circuits in jax},\ }\bibfield  {journal} {\bibinfo  {journal} {arXiv preprint arXiv:2501.04776}\ }\href {https://doi.org/https://doi.org/10.48550/arXiv.2501.04776} {https://doi.org/10.48550/arXiv.2501.04776} (\bibinfo {year} {2025})\BibitemShut {NoStop}%
\bibitem [{\citenamefont {Recio}\ and\ \citenamefont {Bowles}(2025)}]{ErikRecio2025}%
  \BibitemOpen
  \bibfield  {author} {\bibinfo {author} {\bibfnamefont {E.}~\bibnamefont {Recio}}\ and\ \bibinfo {author} {\bibfnamefont {J.}~\bibnamefont {Bowles}},\ }\href@noop {} {\bibinfo {title} {Fast optimization of instantaneous quantum polynomial circuits}},\ \bibinfo {howpublished} {\url{https://pennylane.ai/qml/demos/tutorial_iqp_circuit_optimization_jax}} (\bibinfo {year} {2025}),\ \bibinfo {note} {date Accessed: 2026-03-13}\BibitemShut {NoStop}%
\bibitem [{\citenamefont {Vaswani}\ \emph {et~al.}(2017)\citenamefont {Vaswani}, \citenamefont {Shazeer}, \citenamefont {Parmar}, \citenamefont {Uszkoreit}, \citenamefont {Jones}, \citenamefont {Gomez}, \citenamefont {Kaiser},\ and\ \citenamefont {Polosukhin}}]{vaswani2017attention}%
  \BibitemOpen
  \bibfield  {author} {\bibinfo {author} {\bibfnamefont {A.}~\bibnamefont {Vaswani}}, \bibinfo {author} {\bibfnamefont {N.}~\bibnamefont {Shazeer}}, \bibinfo {author} {\bibfnamefont {N.}~\bibnamefont {Parmar}}, \bibinfo {author} {\bibfnamefont {J.}~\bibnamefont {Uszkoreit}}, \bibinfo {author} {\bibfnamefont {L.}~\bibnamefont {Jones}}, \bibinfo {author} {\bibfnamefont {A.~N.}\ \bibnamefont {Gomez}}, \bibinfo {author} {\bibfnamefont {{\L}.}~\bibnamefont {Kaiser}},\ and\ \bibinfo {author} {\bibfnamefont {I.}~\bibnamefont {Polosukhin}},\ }\bibfield  {title} {\bibinfo {title} {Attention is all you need},\ }\href {https://proceedings.neurips.cc/paper/2017/file/3f5ee243547dee91fbd053c1c4a845aa-Paper.pdf} {\bibfield  {journal} {\bibinfo  {journal} {Advances in neural information processing systems}\ }\textbf {\bibinfo {volume} {30}} (\bibinfo {year} {2017})}\BibitemShut {NoStop}%
\bibitem [{\citenamefont {Ho}\ \emph {et~al.}(2020)\citenamefont {Ho}, \citenamefont {Jain},\ and\ \citenamefont {Abbeel}}]{ho2020denoising}%
  \BibitemOpen
  \bibfield  {author} {\bibinfo {author} {\bibfnamefont {J.}~\bibnamefont {Ho}}, \bibinfo {author} {\bibfnamefont {A.}~\bibnamefont {Jain}},\ and\ \bibinfo {author} {\bibfnamefont {P.}~\bibnamefont {Abbeel}},\ }\bibfield  {title} {\bibinfo {title} {Denoising diffusion probabilistic models},\ }\href {https://proceedings.neurips.cc/paper/2020/hash/4c5bcfec8584af0d967f1ab10179ca4b-Abstract.html} {\bibfield  {journal} {\bibinfo  {journal} {Advances in neural information processing systems}\ }\textbf {\bibinfo {volume} {33}},\ \bibinfo {pages} {6840} (\bibinfo {year} {2020})}\BibitemShut {NoStop}%
\bibitem [{\citenamefont {Kingma}\ and\ \citenamefont {Welling}(2013)}]{kingma2013auto}%
  \BibitemOpen
  \bibfield  {author} {\bibinfo {author} {\bibfnamefont {D.~P.}\ \bibnamefont {Kingma}}\ and\ \bibinfo {author} {\bibfnamefont {M.}~\bibnamefont {Welling}},\ }\bibfield  {title} {\bibinfo {title} {Auto-encoding variational bayes},\ }\href {https://arxiv.org/abs/1312.6114} {\bibfield  {journal} {\bibinfo  {journal} {arXiv preprint arXiv:1312.6114}\ } (\bibinfo {year} {2013})}\BibitemShut {NoStop}%
\bibitem [{\citenamefont {Goodfellow}\ \emph {et~al.}(2020)\citenamefont {Goodfellow}, \citenamefont {Pouget-Abadie}, \citenamefont {Mirza}, \citenamefont {Xu}, \citenamefont {Warde-Farley}, \citenamefont {Ozair}, \citenamefont {Courville},\ and\ \citenamefont {Bengio}}]{goodfellow2020generative}%
  \BibitemOpen
  \bibfield  {author} {\bibinfo {author} {\bibfnamefont {I.}~\bibnamefont {Goodfellow}}, \bibinfo {author} {\bibfnamefont {J.}~\bibnamefont {Pouget-Abadie}}, \bibinfo {author} {\bibfnamefont {M.}~\bibnamefont {Mirza}}, \bibinfo {author} {\bibfnamefont {B.}~\bibnamefont {Xu}}, \bibinfo {author} {\bibfnamefont {D.}~\bibnamefont {Warde-Farley}}, \bibinfo {author} {\bibfnamefont {S.}~\bibnamefont {Ozair}}, \bibinfo {author} {\bibfnamefont {A.}~\bibnamefont {Courville}},\ and\ \bibinfo {author} {\bibfnamefont {Y.}~\bibnamefont {Bengio}},\ }\bibfield  {title} {\bibinfo {title} {Generative adversarial networks},\ }\href {https://dl.acm.org/doi/abs/10.1145/3422622} {\bibfield  {journal} {\bibinfo  {journal} {Communications of the ACM}\ }\textbf {\bibinfo {volume} {63}},\ \bibinfo {pages} {139} (\bibinfo {year} {2020})}\BibitemShut {NoStop}%
\bibitem [{\citenamefont {Zoufal}\ \emph {et~al.}(2019)\citenamefont {Zoufal}, \citenamefont {Lucchi},\ and\ \citenamefont {Woerner}}]{zoufal2019quantum}%
  \BibitemOpen
  \bibfield  {author} {\bibinfo {author} {\bibfnamefont {C.}~\bibnamefont {Zoufal}}, \bibinfo {author} {\bibfnamefont {A.}~\bibnamefont {Lucchi}},\ and\ \bibinfo {author} {\bibfnamefont {S.}~\bibnamefont {Woerner}},\ }\bibfield  {title} {\bibinfo {title} {Quantum generative adversarial networks for learning and loading random distributions},\ }\href {https://doi.org/10.1038/s41534-019-0223-2} {\bibfield  {journal} {\bibinfo  {journal} {npj Quantum Information}\ }\textbf {\bibinfo {volume} {5}},\ \bibinfo {pages} {1} (\bibinfo {year} {2019})}\BibitemShut {NoStop}%
\bibitem [{\citenamefont {Dallaire-Demers}\ and\ \citenamefont {Killoran}(2018)}]{dallaire2018quantum}%
  \BibitemOpen
  \bibfield  {author} {\bibinfo {author} {\bibfnamefont {P.-L.}\ \bibnamefont {Dallaire-Demers}}\ and\ \bibinfo {author} {\bibfnamefont {N.}~\bibnamefont {Killoran}},\ }\bibfield  {title} {\bibinfo {title} {Quantum generative adversarial networks},\ }\href {https://doi.org/10.1103/PhysRevA.98.012324} {\bibfield  {journal} {\bibinfo  {journal} {Physical Review A}\ }\textbf {\bibinfo {volume} {98}},\ \bibinfo {pages} {012324} (\bibinfo {year} {2018})}\BibitemShut {NoStop}%
\bibitem [{\citenamefont {Amin}\ \emph {et~al.}(2018)\citenamefont {Amin}, \citenamefont {Andriyash}, \citenamefont {Rolfe}, \citenamefont {Kulchytskyy},\ and\ \citenamefont {Melko}}]{Amin2018Quantum}%
  \BibitemOpen
  \bibfield  {author} {\bibinfo {author} {\bibfnamefont {M.~H.}\ \bibnamefont {Amin}}, \bibinfo {author} {\bibfnamefont {E.}~\bibnamefont {Andriyash}}, \bibinfo {author} {\bibfnamefont {J.}~\bibnamefont {Rolfe}}, \bibinfo {author} {\bibfnamefont {B.}~\bibnamefont {Kulchytskyy}},\ and\ \bibinfo {author} {\bibfnamefont {R.}~\bibnamefont {Melko}},\ }\bibfield  {title} {\bibinfo {title} {Quantum boltzmann machine},\ }\href {https://doi.org/10.1103/PhysRevX.8.021050} {\bibfield  {journal} {\bibinfo  {journal} {Phys. Rev. X}\ }\textbf {\bibinfo {volume} {8}},\ \bibinfo {pages} {021050} (\bibinfo {year} {2018})}\BibitemShut {NoStop}%
\bibitem [{\citenamefont {Coopmans}\ and\ \citenamefont {Benedetti}(2024)}]{coopmans2023sample}%
  \BibitemOpen
  \bibfield  {author} {\bibinfo {author} {\bibfnamefont {L.}~\bibnamefont {Coopmans}}\ and\ \bibinfo {author} {\bibfnamefont {M.}~\bibnamefont {Benedetti}},\ }\bibfield  {title} {\bibinfo {title} {On the sample complexity of quantum boltzmann machine learning},\ }\href {https://doi.org/10.1038/s42005-024-01763-x} {\bibfield  {journal} {\bibinfo  {journal} {Communications Physics}\ }\textbf {\bibinfo {volume} {7}},\ \bibinfo {pages} {274} (\bibinfo {year} {2024})}\BibitemShut {NoStop}%
\bibitem [{\citenamefont {Zoufal}\ \emph {et~al.}(2021)\citenamefont {Zoufal}, \citenamefont {Lucchi},\ and\ \citenamefont {Woerner}}]{zoufal2021variational}%
  \BibitemOpen
  \bibfield  {author} {\bibinfo {author} {\bibfnamefont {C.}~\bibnamefont {Zoufal}}, \bibinfo {author} {\bibfnamefont {A.}~\bibnamefont {Lucchi}},\ and\ \bibinfo {author} {\bibfnamefont {S.}~\bibnamefont {Woerner}},\ }\bibfield  {title} {\bibinfo {title} {Variational quantum boltzmann machines},\ }\href {https://doi.org/10.1007/s42484-020-00033-7} {\bibfield  {journal} {\bibinfo  {journal} {Quantum Machine Intelligence}\ }\textbf {\bibinfo {volume} {3}},\ \bibinfo {pages} {1} (\bibinfo {year} {2021})}\BibitemShut {NoStop}%
\bibitem [{\citenamefont {Demidik}\ \emph {et~al.}(2025)\citenamefont {Demidik}, \citenamefont {T{\"u}ys{\"u}z}, \citenamefont {Piatkowski}, \citenamefont {Grossi},\ and\ \citenamefont {Jansen}}]{demidik2025expressive}%
  \BibitemOpen
  \bibfield  {author} {\bibinfo {author} {\bibfnamefont {M.}~\bibnamefont {Demidik}}, \bibinfo {author} {\bibfnamefont {C.}~\bibnamefont {T{\"u}ys{\"u}z}}, \bibinfo {author} {\bibfnamefont {N.}~\bibnamefont {Piatkowski}}, \bibinfo {author} {\bibfnamefont {M.}~\bibnamefont {Grossi}},\ and\ \bibinfo {author} {\bibfnamefont {K.}~\bibnamefont {Jansen}},\ }\bibfield  {title} {\bibinfo {title} {Expressive equivalence of classical and quantum restricted boltzmann machines},\ }\href {https://www.nature.com/articles/s42005-025-02353-1} {\bibfield  {journal} {\bibinfo  {journal} {Communications Physics}\ }\textbf {\bibinfo {volume} {8}},\ \bibinfo {pages} {413} (\bibinfo {year} {2025})}\BibitemShut {NoStop}%
\bibitem [{\citenamefont {Chang}\ \emph {et~al.}(2024)\citenamefont {Chang}, \citenamefont {Thanasilp}, \citenamefont {Saux}, \citenamefont {Vallecorsa},\ and\ \citenamefont {Grossi}}]{chang2024latent}%
  \BibitemOpen
  \bibfield  {author} {\bibinfo {author} {\bibfnamefont {S.~Y.}\ \bibnamefont {Chang}}, \bibinfo {author} {\bibfnamefont {S.}~\bibnamefont {Thanasilp}}, \bibinfo {author} {\bibfnamefont {B.~L.}\ \bibnamefont {Saux}}, \bibinfo {author} {\bibfnamefont {S.}~\bibnamefont {Vallecorsa}},\ and\ \bibinfo {author} {\bibfnamefont {M.}~\bibnamefont {Grossi}},\ }\bibfield  {title} {\bibinfo {title} {Latent style-based quantum gan for high-quality image generation},\ }\bibfield  {journal} {\bibinfo  {journal} {arXiv preprint arXiv:2406.02668}\ }\href {https://doi.org/10.48550/arXiv.2406.02668} {10.48550/arXiv.2406.02668} (\bibinfo {year} {2024})\BibitemShut {NoStop}%
\bibitem [{\citenamefont {Barthe}\ \emph {et~al.}(2025)\citenamefont {Barthe}, \citenamefont {Grossi}, \citenamefont {Vallecorsa}, \citenamefont {Tura},\ and\ \citenamefont {Dunjko}}]{barthe2025parameterized}%
  \BibitemOpen
  \bibfield  {author} {\bibinfo {author} {\bibfnamefont {A.}~\bibnamefont {Barthe}}, \bibinfo {author} {\bibfnamefont {M.}~\bibnamefont {Grossi}}, \bibinfo {author} {\bibfnamefont {S.}~\bibnamefont {Vallecorsa}}, \bibinfo {author} {\bibfnamefont {J.}~\bibnamefont {Tura}},\ and\ \bibinfo {author} {\bibfnamefont {V.}~\bibnamefont {Dunjko}},\ }\bibfield  {title} {\bibinfo {title} {Parameterized quantum circuits as universal generative models for continuous multivariate distributions},\ }\href {https://www.nature.com/articles/s41534-025-01064-3} {\bibfield  {journal} {\bibinfo  {journal} {npj Quantum Information}\ }\textbf {\bibinfo {volume} {11}},\ \bibinfo {pages} {121} (\bibinfo {year} {2025})}\BibitemShut {NoStop}%
\bibitem [{\citenamefont {Romero}\ and\ \citenamefont {Aspuru-Guzik}(2021)}]{romero2021variational}%
  \BibitemOpen
  \bibfield  {author} {\bibinfo {author} {\bibfnamefont {J.}~\bibnamefont {Romero}}\ and\ \bibinfo {author} {\bibfnamefont {A.}~\bibnamefont {Aspuru-Guzik}},\ }\bibfield  {title} {\bibinfo {title} {Variational quantum generators: Generative adversarial quantum machine learning for continuous distributions},\ }\href {https://doi.org/10.1002/qute.202000003} {\bibfield  {journal} {\bibinfo  {journal} {Advanced Quantum Technologies}\ }\textbf {\bibinfo {volume} {4}},\ \bibinfo {pages} {2000003} (\bibinfo {year} {2021})}\BibitemShut {NoStop}%
\bibitem [{\citenamefont {Mart{\'\i}nez~de Lejarza}\ \emph {et~al.}(2025)\citenamefont {Mart{\'\i}nez~de Lejarza}, \citenamefont {Wu}, \citenamefont {Kyriienko}, \citenamefont {Rodrigo},\ and\ \citenamefont {Grossi}}]{martinez2025quantum}%
  \BibitemOpen
  \bibfield  {author} {\bibinfo {author} {\bibfnamefont {J.~J.}\ \bibnamefont {Mart{\'\i}nez~de Lejarza}}, \bibinfo {author} {\bibfnamefont {H.-Y.}\ \bibnamefont {Wu}}, \bibinfo {author} {\bibfnamefont {O.}~\bibnamefont {Kyriienko}}, \bibinfo {author} {\bibfnamefont {G.}~\bibnamefont {Rodrigo}},\ and\ \bibinfo {author} {\bibfnamefont {M.}~\bibnamefont {Grossi}},\ }\bibfield  {title} {\bibinfo {title} {Quantum chebyshev probabilistic models for fragmentation functions},\ }\href {https://www.nature.com/articles/s42005-025-02361-1} {\bibfield  {journal} {\bibinfo  {journal} {Communications Physics}\ }\textbf {\bibinfo {volume} {8}},\ \bibinfo {pages} {448} (\bibinfo {year} {2025})}\BibitemShut {NoStop}%
\bibitem [{\citenamefont {Wu}\ \emph {et~al.}(2025)\citenamefont {Wu}, \citenamefont {Elfving},\ and\ \citenamefont {Kyriienko}}]{wu2025multidimensional}%
  \BibitemOpen
  \bibfield  {author} {\bibinfo {author} {\bibfnamefont {H.-Y.}\ \bibnamefont {Wu}}, \bibinfo {author} {\bibfnamefont {V.~E.}\ \bibnamefont {Elfving}},\ and\ \bibinfo {author} {\bibfnamefont {O.}~\bibnamefont {Kyriienko}},\ }\bibfield  {title} {\bibinfo {title} {Multidimensional quantum generative modeling by quantum hartley transform},\ }\href {https://advanced.onlinelibrary.wiley.com/doi/full/10.1002/qute.202400337} {\bibfield  {journal} {\bibinfo  {journal} {Advanced Quantum Technologies}\ }\textbf {\bibinfo {volume} {8}},\ \bibinfo {pages} {2400337} (\bibinfo {year} {2025})}\BibitemShut {NoStop}%
\bibitem [{\citenamefont {Kullback}\ and\ \citenamefont {Leibler}(1951)}]{kullback1951on}%
  \BibitemOpen
  \bibfield  {author} {\bibinfo {author} {\bibfnamefont {S.}~\bibnamefont {Kullback}}\ and\ \bibinfo {author} {\bibfnamefont {R.~A.}\ \bibnamefont {Leibler}},\ }\bibfield  {title} {\bibinfo {title} {On information and sufficiency},\ }\bibfield  {journal} {\bibinfo  {journal} {The annals of mathematical statistics}\ }\textbf {\bibinfo {volume} {22}},\ \href {https://doi.org/10.1214/aoms/1177729694} {10.1214/aoms/1177729694} (\bibinfo {year} {1951})\BibitemShut {NoStop}%
\bibitem [{\citenamefont {Lin}(1991)}]{lin1991divergence}%
  \BibitemOpen
  \bibfield  {author} {\bibinfo {author} {\bibfnamefont {J.}~\bibnamefont {Lin}},\ }\bibfield  {title} {\bibinfo {title} {Divergence measures based on the shannon entropy},\ }\href {https://doi.org/10.1109/18.61115} {\bibfield  {journal} {\bibinfo  {journal} {IEEE Transactions on Information theory}\ }\textbf {\bibinfo {volume} {37}},\ \bibinfo {pages} {145} (\bibinfo {year} {1991})}\BibitemShut {NoStop}%
\bibitem [{\citenamefont {Gretton}\ \emph {et~al.}(2012)\citenamefont {Gretton}, \citenamefont {Borgwardt}, \citenamefont {Rasch}, \citenamefont {Sch{\"o}lkopf},\ and\ \citenamefont {Smola}}]{gretton2012kernel}%
  \BibitemOpen
  \bibfield  {author} {\bibinfo {author} {\bibfnamefont {A.}~\bibnamefont {Gretton}}, \bibinfo {author} {\bibfnamefont {K.~M.}\ \bibnamefont {Borgwardt}}, \bibinfo {author} {\bibfnamefont {M.~J.}\ \bibnamefont {Rasch}}, \bibinfo {author} {\bibfnamefont {B.}~\bibnamefont {Sch{\"o}lkopf}},\ and\ \bibinfo {author} {\bibfnamefont {A.}~\bibnamefont {Smola}},\ }\bibfield  {title} {\bibinfo {title} {A kernel two-sample test},\ }\href {https://doi.org/10.5555/2188385.2188410} {\bibfield  {journal} {\bibinfo  {journal} {The Journal of Machine Learning Research}\ }\textbf {\bibinfo {volume} {13}},\ \bibinfo {pages} {723} (\bibinfo {year} {2012})}\BibitemShut {NoStop}%
\bibitem [{\citenamefont {Arrasmith}\ \emph {et~al.}(2022)\citenamefont {Arrasmith}, \citenamefont {Holmes}, \citenamefont {Cerezo},\ and\ \citenamefont {Coles}}]{arrasmith2021equivalence}%
  \BibitemOpen
  \bibfield  {author} {\bibinfo {author} {\bibfnamefont {A.}~\bibnamefont {Arrasmith}}, \bibinfo {author} {\bibfnamefont {Z.}~\bibnamefont {Holmes}}, \bibinfo {author} {\bibfnamefont {M.}~\bibnamefont {Cerezo}},\ and\ \bibinfo {author} {\bibfnamefont {P.~J.}\ \bibnamefont {Coles}},\ }\bibfield  {title} {\bibinfo {title} {Equivalence of quantum barren plateaus to cost concentration and narrow gorges},\ }\href {https://doi.org/10.1088/2058-9565/ac7d06} {\bibfield  {journal} {\bibinfo  {journal} {Quantum Science and Technology}\ }\textbf {\bibinfo {volume} {7}},\ \bibinfo {pages} {045015} (\bibinfo {year} {2022})}\BibitemShut {NoStop}%
\bibitem [{\citenamefont {Arrasmith}\ \emph {et~al.}(2021)\citenamefont {Arrasmith}, \citenamefont {Cerezo}, \citenamefont {Czarnik}, \citenamefont {Cincio},\ and\ \citenamefont {Coles}}]{arrasmith2020effect}%
  \BibitemOpen
  \bibfield  {author} {\bibinfo {author} {\bibfnamefont {A.}~\bibnamefont {Arrasmith}}, \bibinfo {author} {\bibfnamefont {M.}~\bibnamefont {Cerezo}}, \bibinfo {author} {\bibfnamefont {P.}~\bibnamefont {Czarnik}}, \bibinfo {author} {\bibfnamefont {L.}~\bibnamefont {Cincio}},\ and\ \bibinfo {author} {\bibfnamefont {P.~J.}\ \bibnamefont {Coles}},\ }\bibfield  {title} {\bibinfo {title} {Effect of barren plateaus on gradient-free optimization},\ }\href {https://doi.org/10.22331/q-2021-10-05-558} {\bibfield  {journal} {\bibinfo  {journal} {Quantum}\ }\textbf {\bibinfo {volume} {5}},\ \bibinfo {pages} {558} (\bibinfo {year} {2021})}\BibitemShut {NoStop}%
\bibitem [{\citenamefont {Aghaei~Saem}\ \emph {et~al.}(2026)\citenamefont {Aghaei~Saem}, \citenamefont {Tafreshi}, \citenamefont {Holmes},\ and\ \citenamefont {Thanasilp}}]{aghaei2026pitfalls}%
  \BibitemOpen
  \bibfield  {author} {\bibinfo {author} {\bibfnamefont {R.}~\bibnamefont {Aghaei~Saem}}, \bibinfo {author} {\bibfnamefont {B.}~\bibnamefont {Tafreshi}}, \bibinfo {author} {\bibfnamefont {Z.}~\bibnamefont {Holmes}},\ and\ \bibinfo {author} {\bibfnamefont {S.}~\bibnamefont {Thanasilp}},\ }\bibfield  {title} {\bibinfo {title} {Pitfalls when tackling the exponential concentration of parameterized quantum models},\ }\href {https://doi.org/https://doi.org/10.1088/2058-9565/ae2202} {\bibfield  {journal} {\bibinfo  {journal} {Quantum Science and Technology}\ }\textbf {\bibinfo {volume} {11}},\ \bibinfo {pages} {015049} (\bibinfo {year} {2026})}\BibitemShut {NoStop}%
\bibitem [{\citenamefont {Holmes}\ \emph {et~al.}(2022)\citenamefont {Holmes}, \citenamefont {Sharma}, \citenamefont {Cerezo},\ and\ \citenamefont {Coles}}]{holmes2021connecting}%
  \BibitemOpen
  \bibfield  {author} {\bibinfo {author} {\bibfnamefont {Z.}~\bibnamefont {Holmes}}, \bibinfo {author} {\bibfnamefont {K.}~\bibnamefont {Sharma}}, \bibinfo {author} {\bibfnamefont {M.}~\bibnamefont {Cerezo}},\ and\ \bibinfo {author} {\bibfnamefont {P.~J.}\ \bibnamefont {Coles}},\ }\bibfield  {title} {\bibinfo {title} {Connecting ansatz expressibility to gradient magnitudes and barren plateaus},\ }\href {https://doi.org/10.1103/PRXQuantum.3.010313} {\bibfield  {journal} {\bibinfo  {journal} {PRX Quantum}\ }\textbf {\bibinfo {volume} {3}},\ \bibinfo {pages} {010313} (\bibinfo {year} {2022})}\BibitemShut {NoStop}%
\bibitem [{\citenamefont {Fontana}\ \emph {et~al.}(2024)\citenamefont {Fontana}, \citenamefont {Herman}, \citenamefont {Chakrabarti}, \citenamefont {Kumar}, \citenamefont {Yalovetzky}, \citenamefont {Heredge}, \citenamefont {Sureshbabu},\ and\ \citenamefont {Pistoia}}]{fontana2023theadjoint}%
  \BibitemOpen
  \bibfield  {author} {\bibinfo {author} {\bibfnamefont {E.}~\bibnamefont {Fontana}}, \bibinfo {author} {\bibfnamefont {D.}~\bibnamefont {Herman}}, \bibinfo {author} {\bibfnamefont {S.}~\bibnamefont {Chakrabarti}}, \bibinfo {author} {\bibfnamefont {N.}~\bibnamefont {Kumar}}, \bibinfo {author} {\bibfnamefont {R.}~\bibnamefont {Yalovetzky}}, \bibinfo {author} {\bibfnamefont {J.}~\bibnamefont {Heredge}}, \bibinfo {author} {\bibfnamefont {S.~H.}\ \bibnamefont {Sureshbabu}},\ and\ \bibinfo {author} {\bibfnamefont {M.}~\bibnamefont {Pistoia}},\ }\bibfield  {title} {\bibinfo {title} {Characterizing barren plateaus in quantum ansätze with the adjoint representation},\ }\href {https://doi.org/10.1038/s41467-024-49910-w} {\bibfield  {journal} {\bibinfo  {journal} {Nature Communications}\ }\textbf {\bibinfo {volume} {15}},\ \bibinfo {pages} {7171} (\bibinfo {year} {2024})}\BibitemShut {NoStop}%
\bibitem [{\citenamefont {Ragone}\ \emph {et~al.}(2024)\citenamefont {Ragone}, \citenamefont {Bakalov}, \citenamefont {Sauvage}, \citenamefont {Kemper}, \citenamefont {Ortiz~Marrero}, \citenamefont {Larocca},\ and\ \citenamefont {Cerezo}}]{ragone2023unified}%
  \BibitemOpen
  \bibfield  {author} {\bibinfo {author} {\bibfnamefont {M.}~\bibnamefont {Ragone}}, \bibinfo {author} {\bibfnamefont {B.~N.}\ \bibnamefont {Bakalov}}, \bibinfo {author} {\bibfnamefont {F.}~\bibnamefont {Sauvage}}, \bibinfo {author} {\bibfnamefont {A.~F.}\ \bibnamefont {Kemper}}, \bibinfo {author} {\bibfnamefont {C.}~\bibnamefont {Ortiz~Marrero}}, \bibinfo {author} {\bibfnamefont {M.}~\bibnamefont {Larocca}},\ and\ \bibinfo {author} {\bibfnamefont {M.}~\bibnamefont {Cerezo}},\ }\bibfield  {title} {\bibinfo {title} {A lie algebraic theory of barren plateaus for deep parameterized quantum circuits},\ }\href {https://doi.org/10.1038/s41467-024-49909-3} {\bibfield  {journal} {\bibinfo  {journal} {Nature Communications}\ }\textbf {\bibinfo {volume} {15}},\ \bibinfo {pages} {7172} (\bibinfo {year} {2024})}\BibitemShut {NoStop}%
\bibitem [{\citenamefont {Srimahajariyapong}\ \emph {et~al.}(2025)\citenamefont {Srimahajariyapong}, \citenamefont {Thanasilp},\ and\ \citenamefont {Chotibut}}]{srimahajariyapong2025connecting}%
  \BibitemOpen
  \bibfield  {author} {\bibinfo {author} {\bibfnamefont {K.}~\bibnamefont {Srimahajariyapong}}, \bibinfo {author} {\bibfnamefont {S.}~\bibnamefont {Thanasilp}},\ and\ \bibinfo {author} {\bibfnamefont {T.}~\bibnamefont {Chotibut}},\ }\bibfield  {title} {\bibinfo {title} {Connecting phases of matter to the flatness of the loss landscape in analog variational quantum algorithms},\ }\href {https://arxiv.org/abs/2506.13865} {\bibfield  {journal} {\bibinfo  {journal} {arXiv preprint arXiv:2506.13865}\ } (\bibinfo {year} {2025})}\BibitemShut {NoStop}%
\bibitem [{\citenamefont {Holmes}\ \emph {et~al.}(2021)\citenamefont {Holmes}, \citenamefont {Arrasmith}, \citenamefont {Yan}, \citenamefont {Coles}, \citenamefont {Albrecht},\ and\ \citenamefont {Sornborger}}]{holmes2020barren}%
  \BibitemOpen
  \bibfield  {author} {\bibinfo {author} {\bibfnamefont {Z.}~\bibnamefont {Holmes}}, \bibinfo {author} {\bibfnamefont {A.}~\bibnamefont {Arrasmith}}, \bibinfo {author} {\bibfnamefont {B.}~\bibnamefont {Yan}}, \bibinfo {author} {\bibfnamefont {P.~J.}\ \bibnamefont {Coles}}, \bibinfo {author} {\bibfnamefont {A.}~\bibnamefont {Albrecht}},\ and\ \bibinfo {author} {\bibfnamefont {A.~T.}\ \bibnamefont {Sornborger}},\ }\bibfield  {title} {\bibinfo {title} {Barren plateaus preclude learning scramblers},\ }\href {https://doi.org/10.1103/PhysRevLett.126.190501} {\bibfield  {journal} {\bibinfo  {journal} {Physical Review Letters}\ }\textbf {\bibinfo {volume} {126}},\ \bibinfo {pages} {190501} (\bibinfo {year} {2021})}\BibitemShut {NoStop}%
\bibitem [{\citenamefont {Patti}\ \emph {et~al.}(2021)\citenamefont {Patti}, \citenamefont {Najafi}, \citenamefont {Gao},\ and\ \citenamefont {Yelin}}]{patti2020entanglement}%
  \BibitemOpen
  \bibfield  {author} {\bibinfo {author} {\bibfnamefont {T.~L.}\ \bibnamefont {Patti}}, \bibinfo {author} {\bibfnamefont {K.}~\bibnamefont {Najafi}}, \bibinfo {author} {\bibfnamefont {X.}~\bibnamefont {Gao}},\ and\ \bibinfo {author} {\bibfnamefont {S.~F.}\ \bibnamefont {Yelin}},\ }\bibfield  {title} {\bibinfo {title} {Entanglement devised barren plateau mitigation},\ }\href {https://doi.org/10.1103/PhysRevResearch.3.033090} {\bibfield  {journal} {\bibinfo  {journal} {Physical Review Research}\ }\textbf {\bibinfo {volume} {3}},\ \bibinfo {pages} {033090} (\bibinfo {year} {2021})}\BibitemShut {NoStop}%
\bibitem [{\citenamefont {Marrero}\ \emph {et~al.}(2021)\citenamefont {Marrero}, \citenamefont {Kieferov{\'a}},\ and\ \citenamefont {Wiebe}}]{marrero2020entanglement}%
  \BibitemOpen
  \bibfield  {author} {\bibinfo {author} {\bibfnamefont {C.~O.}\ \bibnamefont {Marrero}}, \bibinfo {author} {\bibfnamefont {M.}~\bibnamefont {Kieferov{\'a}}},\ and\ \bibinfo {author} {\bibfnamefont {N.}~\bibnamefont {Wiebe}},\ }\bibfield  {title} {\bibinfo {title} {Entanglement-induced barren plateaus},\ }\href {https://doi.org/10.1103/PRXQuantum.2.040316} {\bibfield  {journal} {\bibinfo  {journal} {PRX Quantum}\ }\textbf {\bibinfo {volume} {2}},\ \bibinfo {pages} {040316} (\bibinfo {year} {2021})}\BibitemShut {NoStop}%
\bibitem [{\citenamefont {Cerezo}\ \emph {et~al.}(2021)\citenamefont {Cerezo}, \citenamefont {Sone}, \citenamefont {Volkoff}, \citenamefont {Cincio},\ and\ \citenamefont {Coles}}]{cerezo2020cost}%
  \BibitemOpen
  \bibfield  {author} {\bibinfo {author} {\bibfnamefont {M.}~\bibnamefont {Cerezo}}, \bibinfo {author} {\bibfnamefont {A.}~\bibnamefont {Sone}}, \bibinfo {author} {\bibfnamefont {T.}~\bibnamefont {Volkoff}}, \bibinfo {author} {\bibfnamefont {L.}~\bibnamefont {Cincio}},\ and\ \bibinfo {author} {\bibfnamefont {P.~J.}\ \bibnamefont {Coles}},\ }\bibfield  {title} {\bibinfo {title} {Cost function dependent barren plateaus in shallow parametrized quantum circuits},\ }\href {https://doi.org/10.1038/s41467-021-21728-w} {\bibfield  {journal} {\bibinfo  {journal} {Nature {C}ommunications}\ }\textbf {\bibinfo {volume} {12}},\ \bibinfo {pages} {1} (\bibinfo {year} {2021})}\BibitemShut {NoStop}%
\bibitem [{\citenamefont {Letcher}\ \emph {et~al.}(2024)\citenamefont {Letcher}, \citenamefont {Woerner},\ and\ \citenamefont {Zoufal}}]{letcher2023tight}%
  \BibitemOpen
  \bibfield  {author} {\bibinfo {author} {\bibfnamefont {A.}~\bibnamefont {Letcher}}, \bibinfo {author} {\bibfnamefont {S.}~\bibnamefont {Woerner}},\ and\ \bibinfo {author} {\bibfnamefont {C.}~\bibnamefont {Zoufal}},\ }\bibfield  {title} {\bibinfo {title} {Tight and efficient gradient bounds for parameterized quantum circuits},\ }\href {https://quantum-journal.org/papers/q-2024-09-25-1484/} {\bibfield  {journal} {\bibinfo  {journal} {Quantum}\ }\textbf {\bibinfo {volume} {8}},\ \bibinfo {pages} {1484} (\bibinfo {year} {2024})}\BibitemShut {NoStop}%
\bibitem [{\citenamefont {Wang}\ \emph {et~al.}(2021)\citenamefont {Wang}, \citenamefont {Fontana}, \citenamefont {Cerezo}, \citenamefont {Sharma}, \citenamefont {Sone}, \citenamefont {Cincio},\ and\ \citenamefont {Coles}}]{wang2020noise}%
  \BibitemOpen
  \bibfield  {author} {\bibinfo {author} {\bibfnamefont {S.}~\bibnamefont {Wang}}, \bibinfo {author} {\bibfnamefont {E.}~\bibnamefont {Fontana}}, \bibinfo {author} {\bibfnamefont {M.}~\bibnamefont {Cerezo}}, \bibinfo {author} {\bibfnamefont {K.}~\bibnamefont {Sharma}}, \bibinfo {author} {\bibfnamefont {A.}~\bibnamefont {Sone}}, \bibinfo {author} {\bibfnamefont {L.}~\bibnamefont {Cincio}},\ and\ \bibinfo {author} {\bibfnamefont {P.~J.}\ \bibnamefont {Coles}},\ }\bibfield  {title} {\bibinfo {title} {Noise-induced barren plateaus in variational quantum algorithms},\ }\href {https://doi.org/10.1038/s41467-021-27045-6} {\bibfield  {journal} {\bibinfo  {journal} {Nature Communications}\ }\textbf {\bibinfo {volume} {12}},\ \bibinfo {pages} {1} (\bibinfo {year} {2021})}\BibitemShut {NoStop}%
\bibitem [{\citenamefont {Crognaletti}\ \emph {et~al.}(2024)\citenamefont {Crognaletti}, \citenamefont {Grossi},\ and\ \citenamefont {Bassi}}]{crognaletti2024estimates}%
  \BibitemOpen
  \bibfield  {author} {\bibinfo {author} {\bibfnamefont {G.}~\bibnamefont {Crognaletti}}, \bibinfo {author} {\bibfnamefont {M.}~\bibnamefont {Grossi}},\ and\ \bibinfo {author} {\bibfnamefont {A.}~\bibnamefont {Bassi}},\ }\bibfield  {title} {\bibinfo {title} {Estimates of loss function concentration in noisy parametrized quantum circuits},\ }\href {https://arxiv.org/abs/2410.01893} {\bibfield  {journal} {\bibinfo  {journal} {arXiv preprint arXiv:2410.01893}\ } (\bibinfo {year} {2024})}\BibitemShut {NoStop}%
\bibitem [{\citenamefont {Xiong}\ \emph {et~al.}(2025)\citenamefont {Xiong}, \citenamefont {Holmes}, \citenamefont {Angrisani}, \citenamefont {Suzuki}, \citenamefont {Chotibut},\ and\ \citenamefont {Thanasilp}}]{xiong2025role}%
  \BibitemOpen
  \bibfield  {author} {\bibinfo {author} {\bibfnamefont {W.}~\bibnamefont {Xiong}}, \bibinfo {author} {\bibfnamefont {Z.}~\bibnamefont {Holmes}}, \bibinfo {author} {\bibfnamefont {A.}~\bibnamefont {Angrisani}}, \bibinfo {author} {\bibfnamefont {Y.}~\bibnamefont {Suzuki}}, \bibinfo {author} {\bibfnamefont {T.}~\bibnamefont {Chotibut}},\ and\ \bibinfo {author} {\bibfnamefont {S.}~\bibnamefont {Thanasilp}},\ }\bibfield  {title} {\bibinfo {title} {Role of scrambling and noise in temporal information processing with quantum systems},\ }\href {https://arxiv.org/abs/2505.10080} {\bibfield  {journal} {\bibinfo  {journal} {arXiv preprint arXiv:2505.10080}\ } (\bibinfo {year} {2025})}\BibitemShut {NoStop}%
\bibitem [{\citenamefont {Thanaslip}\ \emph {et~al.}(2023)\citenamefont {Thanaslip}, \citenamefont {Wang}, \citenamefont {Nghiem}, \citenamefont {Coles},\ and\ \citenamefont {Cerezo}}]{thanaslip2021subtleties}%
  \BibitemOpen
  \bibfield  {author} {\bibinfo {author} {\bibfnamefont {S.}~\bibnamefont {Thanaslip}}, \bibinfo {author} {\bibfnamefont {S.}~\bibnamefont {Wang}}, \bibinfo {author} {\bibfnamefont {N.~A.}\ \bibnamefont {Nghiem}}, \bibinfo {author} {\bibfnamefont {P.~J.}\ \bibnamefont {Coles}},\ and\ \bibinfo {author} {\bibfnamefont {M.}~\bibnamefont {Cerezo}},\ }\bibfield  {title} {\bibinfo {title} {Subtleties in the trainability of quantum machine learning models},\ }\href {https://doi.org/10.1007/s42484-023-00103-6} {\bibfield  {journal} {\bibinfo  {journal} {Quantum Machine Intelligence}\ }\textbf {\bibinfo {volume} {5}},\ \bibinfo {pages} {21} (\bibinfo {year} {2023})}\BibitemShut {NoStop}%
\bibitem [{\citenamefont {Thanasilp}\ \emph {et~al.}(2024)\citenamefont {Thanasilp}, \citenamefont {Wang}, \citenamefont {Cerezo},\ and\ \citenamefont {Holmes}}]{thanasilp2022exponential}%
  \BibitemOpen
  \bibfield  {author} {\bibinfo {author} {\bibfnamefont {S.}~\bibnamefont {Thanasilp}}, \bibinfo {author} {\bibfnamefont {S.}~\bibnamefont {Wang}}, \bibinfo {author} {\bibfnamefont {M.}~\bibnamefont {Cerezo}},\ and\ \bibinfo {author} {\bibfnamefont {Z.}~\bibnamefont {Holmes}},\ }\bibfield  {title} {\bibinfo {title} {Exponential concentration in quantum kernel methods},\ }\href {https://doi.org/10.1038/s41467-024-49287-w} {\bibfield  {journal} {\bibinfo  {journal} {Nature Communications}\ }\textbf {\bibinfo {volume} {15}},\ \bibinfo {pages} {5200} (\bibinfo {year} {2024})}\BibitemShut {NoStop}%
\bibitem [{\citenamefont {K{\"u}bler}\ \emph {et~al.}(2021)\citenamefont {K{\"u}bler}, \citenamefont {Buchholz},\ and\ \citenamefont {Sch{\"o}lkopf}}]{kubler2021inductive}%
  \BibitemOpen
  \bibfield  {author} {\bibinfo {author} {\bibfnamefont {J.}~\bibnamefont {K{\"u}bler}}, \bibinfo {author} {\bibfnamefont {S.}~\bibnamefont {Buchholz}},\ and\ \bibinfo {author} {\bibfnamefont {B.}~\bibnamefont {Sch{\"o}lkopf}},\ }\bibfield  {title} {\bibinfo {title} {The inductive bias of quantum kernels},\ }\href {https://doi.org/10.5555/3540261.3541230} {\bibfield  {journal} {\bibinfo  {journal} {Advances in Neural Information Processing Systems}\ }\textbf {\bibinfo {volume} {34}},\ \bibinfo {pages} {12661} (\bibinfo {year} {2021})}\BibitemShut {NoStop}%
\bibitem [{\citenamefont {Suzuki}\ and\ \citenamefont {Li}(2023)}]{suzuki2023effect}%
  \BibitemOpen
  \bibfield  {author} {\bibinfo {author} {\bibfnamefont {Y.}~\bibnamefont {Suzuki}}\ and\ \bibinfo {author} {\bibfnamefont {M.}~\bibnamefont {Li}},\ }\bibfield  {title} {\bibinfo {title} {Effect of alternating layered ansatzes on trainability of projected quantum kernel},\ }\href {https://arxiv.org/abs/2310.00361} {\bibfield  {journal} {\bibinfo  {journal} {arXiv preprint arXiv:2310.00361}\ } (\bibinfo {year} {2023})}\BibitemShut {NoStop}%
\bibitem [{\citenamefont {Xiong}\ \emph {et~al.}(2023)\citenamefont {Xiong}, \citenamefont {Facelli}, \citenamefont {Sahebi}, \citenamefont {Agnel}, \citenamefont {Chotibut}, \citenamefont {Thanasilp},\ and\ \citenamefont {Holmes}}]{xiong2023fundamental}%
  \BibitemOpen
  \bibfield  {author} {\bibinfo {author} {\bibfnamefont {W.}~\bibnamefont {Xiong}}, \bibinfo {author} {\bibfnamefont {G.}~\bibnamefont {Facelli}}, \bibinfo {author} {\bibfnamefont {M.}~\bibnamefont {Sahebi}}, \bibinfo {author} {\bibfnamefont {O.}~\bibnamefont {Agnel}}, \bibinfo {author} {\bibfnamefont {T.}~\bibnamefont {Chotibut}}, \bibinfo {author} {\bibfnamefont {S.}~\bibnamefont {Thanasilp}},\ and\ \bibinfo {author} {\bibfnamefont {Z.}~\bibnamefont {Holmes}},\ }\bibfield  {title} {\bibinfo {title} {On fundamental aspects of quantum extreme learning machines},\ }\href {https://arxiv.org/abs/2312.15124} {\bibfield  {journal} {\bibinfo  {journal} {arXiv preprint arXiv:2312.15124}\ } (\bibinfo {year} {2023})}\BibitemShut {NoStop}%
\bibitem [{\citenamefont {Shaydulin}\ and\ \citenamefont {Wild}(2022)}]{shaydulin2021importance}%
  \BibitemOpen
  \bibfield  {author} {\bibinfo {author} {\bibfnamefont {R.}~\bibnamefont {Shaydulin}}\ and\ \bibinfo {author} {\bibfnamefont {S.~M.}\ \bibnamefont {Wild}},\ }\bibfield  {title} {\bibinfo {title} {Importance of kernel bandwidth in quantum machine learning},\ }\href {https://doi.org/10.1103/PhysRevA.106.042407} {\bibfield  {journal} {\bibinfo  {journal} {Physical Review A}\ }\textbf {\bibinfo {volume} {106}},\ \bibinfo {pages} {042407} (\bibinfo {year} {2022})}\BibitemShut {NoStop}%
\bibitem [{\citenamefont {Leone}\ \emph {et~al.}(2024)\citenamefont {Leone}, \citenamefont {Oliviero}, \citenamefont {Cincio},\ and\ \citenamefont {Cerezo}}]{leone2022practical}%
  \BibitemOpen
  \bibfield  {author} {\bibinfo {author} {\bibfnamefont {L.}~\bibnamefont {Leone}}, \bibinfo {author} {\bibfnamefont {S.~F.}\ \bibnamefont {Oliviero}}, \bibinfo {author} {\bibfnamefont {L.}~\bibnamefont {Cincio}},\ and\ \bibinfo {author} {\bibfnamefont {M.}~\bibnamefont {Cerezo}},\ }\bibfield  {title} {\bibinfo {title} {On the practical usefulness of the hardware efficient ansatz},\ }\href {https://doi.org/10.22331/q-2024-07-03-1395} {\bibfield  {journal} {\bibinfo  {journal} {Quantum}\ }\textbf {\bibinfo {volume} {8}},\ \bibinfo {pages} {1395} (\bibinfo {year} {2024})}\BibitemShut {NoStop}%
\bibitem [{\citenamefont {Tangpanitanon}\ \emph {et~al.}(2020)\citenamefont {Tangpanitanon}, \citenamefont {Thanasilp}, \citenamefont {Dangniam}, \citenamefont {Lemonde},\ and\ \citenamefont {Angelakis}}]{tangpanitanon2020expressibility}%
  \BibitemOpen
  \bibfield  {author} {\bibinfo {author} {\bibfnamefont {J.}~\bibnamefont {Tangpanitanon}}, \bibinfo {author} {\bibfnamefont {S.}~\bibnamefont {Thanasilp}}, \bibinfo {author} {\bibfnamefont {N.}~\bibnamefont {Dangniam}}, \bibinfo {author} {\bibfnamefont {M.-A.}\ \bibnamefont {Lemonde}},\ and\ \bibinfo {author} {\bibfnamefont {D.~G.}\ \bibnamefont {Angelakis}},\ }\bibfield  {title} {\bibinfo {title} {Expressibility and trainability of parametrized analog quantum systems for machine learning applications},\ }\href {https://doi.org/10.1103/PhysRevResearch.2.043364} {\bibfield  {journal} {\bibinfo  {journal} {Physical Review Research}\ }\textbf {\bibinfo {volume} {2}},\ \bibinfo {pages} {043364} (\bibinfo {year} {2020})}\BibitemShut {NoStop}%
\bibitem [{\citenamefont {Hirviniemi}\ \emph {et~al.}(2026)\citenamefont {Hirviniemi}, \citenamefont {Basheer},\ and\ \citenamefont {Cope}}]{hirviniemi2026preventing}%
  \BibitemOpen
  \bibfield  {author} {\bibinfo {author} {\bibfnamefont {O.}~\bibnamefont {Hirviniemi}}, \bibinfo {author} {\bibfnamefont {A.}~\bibnamefont {Basheer}},\ and\ \bibinfo {author} {\bibfnamefont {T.}~\bibnamefont {Cope}},\ }\bibfield  {title} {\bibinfo {title} {Preventing barren plateaus in continuous quantum generative models},\ }\href {https://arxiv.org/abs/2602.10049} {\bibfield  {journal} {\bibinfo  {journal} {arXiv preprint arXiv:2602.10049}\ } (\bibinfo {year} {2026})}\BibitemShut {NoStop}%
\bibitem [{\citenamefont {Mhiri}\ \emph {et~al.}(2025)\citenamefont {Mhiri}, \citenamefont {Puig}, \citenamefont {Lerch}, \citenamefont {Rudolph}, \citenamefont {Chotibut}, \citenamefont {Thanasilp},\ and\ \citenamefont {Holmes}}]{mhiri2025unifying}%
  \BibitemOpen
  \bibfield  {author} {\bibinfo {author} {\bibfnamefont {H.}~\bibnamefont {Mhiri}}, \bibinfo {author} {\bibfnamefont {R.}~\bibnamefont {Puig}}, \bibinfo {author} {\bibfnamefont {S.}~\bibnamefont {Lerch}}, \bibinfo {author} {\bibfnamefont {M.~S.}\ \bibnamefont {Rudolph}}, \bibinfo {author} {\bibfnamefont {T.}~\bibnamefont {Chotibut}}, \bibinfo {author} {\bibfnamefont {S.}~\bibnamefont {Thanasilp}},\ and\ \bibinfo {author} {\bibfnamefont {Z.}~\bibnamefont {Holmes}},\ }\bibfield  {title} {\bibinfo {title} {A unifying account of warm start guarantees for patches of quantum landscapes},\ }\bibfield  {journal} {\bibinfo  {journal} {arXiv preprint arXiv:2502.07889}\ }\href {https://doi.org/https://doi.org/10.48550/arXiv.2502.07889} {https://doi.org/10.48550/arXiv.2502.07889} (\bibinfo {year} {2025})\BibitemShut {NoStop}%
\bibitem [{\citenamefont {Puig}\ \emph {et~al.}(2025)\citenamefont {Puig}, \citenamefont {Drudis}, \citenamefont {Thanasilp},\ and\ \citenamefont {Holmes}}]{puig2024variational}%
  \BibitemOpen
  \bibfield  {author} {\bibinfo {author} {\bibfnamefont {R.}~\bibnamefont {Puig}}, \bibinfo {author} {\bibfnamefont {M.}~\bibnamefont {Drudis}}, \bibinfo {author} {\bibfnamefont {S.}~\bibnamefont {Thanasilp}},\ and\ \bibinfo {author} {\bibfnamefont {Z.}~\bibnamefont {Holmes}},\ }\bibfield  {title} {\bibinfo {title} {Variational quantum simulation: A case study for understanding warm starts},\ }\href {https://doi.org/10.1103/PRXQuantum.6.010317} {\bibfield  {journal} {\bibinfo  {journal} {PRX Quantum}\ }\textbf {\bibinfo {volume} {6}},\ \bibinfo {pages} {010317} (\bibinfo {year} {2025})}\BibitemShut {NoStop}%
\bibitem [{\citenamefont {Puig}\ \emph {et~al.}(2026)\citenamefont {Puig}, \citenamefont {Casas}, \citenamefont {Cervera-Lierta}, \citenamefont {Holmes},\ and\ \citenamefont {P{\'e}rez-Salinas}}]{puig2026warm}%
  \BibitemOpen
  \bibfield  {author} {\bibinfo {author} {\bibfnamefont {R.}~\bibnamefont {Puig}}, \bibinfo {author} {\bibfnamefont {B.}~\bibnamefont {Casas}}, \bibinfo {author} {\bibfnamefont {A.}~\bibnamefont {Cervera-Lierta}}, \bibinfo {author} {\bibfnamefont {Z.}~\bibnamefont {Holmes}},\ and\ \bibinfo {author} {\bibfnamefont {A.}~\bibnamefont {P{\'e}rez-Salinas}},\ }\bibfield  {title} {\bibinfo {title} {Warm starts, cold states: Exploiting adiabaticity for variational ground-states},\ }\bibfield  {journal} {\bibinfo  {journal} {arXiv preprint arXiv:2602.06137}\ }\href {https://doi.org/https://doi.org/10.48550/arXiv.2602.06137} {https://doi.org/10.48550/arXiv.2602.06137} (\bibinfo {year} {2026})\BibitemShut {NoStop}%
\bibitem [{\citenamefont {Dborin}\ \emph {et~al.}(2022)\citenamefont {Dborin}, \citenamefont {Barratt}, \citenamefont {Wimalaweera}, \citenamefont {Wright},\ and\ \citenamefont {Green}}]{dborin2022matrix}%
  \BibitemOpen
  \bibfield  {author} {\bibinfo {author} {\bibfnamefont {J.}~\bibnamefont {Dborin}}, \bibinfo {author} {\bibfnamefont {F.}~\bibnamefont {Barratt}}, \bibinfo {author} {\bibfnamefont {V.}~\bibnamefont {Wimalaweera}}, \bibinfo {author} {\bibfnamefont {L.}~\bibnamefont {Wright}},\ and\ \bibinfo {author} {\bibfnamefont {A.~G.}\ \bibnamefont {Green}},\ }\bibfield  {title} {\bibinfo {title} {Matrix product state pre-training for quantum machine learning},\ }\href {https://doi.org/10.1088/2058-9565/ac7073} {\bibfield  {journal} {\bibinfo  {journal} {Quantum Science and Technology}\ }\textbf {\bibinfo {volume} {7}},\ \bibinfo {pages} {035014} (\bibinfo {year} {2022})}\BibitemShut {NoStop}%
\bibitem [{\citenamefont {Goh}\ \emph {et~al.}(2025)\citenamefont {Goh}, \citenamefont {Larocca}, \citenamefont {Cincio}, \citenamefont {Cerezo},\ and\ \citenamefont {Sauvage}}]{goh2023lie}%
  \BibitemOpen
  \bibfield  {author} {\bibinfo {author} {\bibfnamefont {M.~L.}\ \bibnamefont {Goh}}, \bibinfo {author} {\bibfnamefont {M.}~\bibnamefont {Larocca}}, \bibinfo {author} {\bibfnamefont {L.}~\bibnamefont {Cincio}}, \bibinfo {author} {\bibfnamefont {M.}~\bibnamefont {Cerezo}},\ and\ \bibinfo {author} {\bibfnamefont {F.}~\bibnamefont {Sauvage}},\ }\bibfield  {title} {\bibinfo {title} {Lie-algebraic classical simulations for quantum computing},\ }\href {https://doi.org/10.1103/3y65-f5w6} {\bibfield  {journal} {\bibinfo  {journal} {Physical Review Research}\ }\textbf {\bibinfo {volume} {7}},\ \bibinfo {pages} {033266} (\bibinfo {year} {2025})}\BibitemShut {NoStop}%
\bibitem [{\citenamefont {Sauvage}\ \emph {et~al.}(2021)\citenamefont {Sauvage}, \citenamefont {Sim}, \citenamefont {Kunitsa}, \citenamefont {Simon}, \citenamefont {Mauri},\ and\ \citenamefont {Perdomo-Ortiz}}]{sauvage2021flip}%
  \BibitemOpen
  \bibfield  {author} {\bibinfo {author} {\bibfnamefont {F.}~\bibnamefont {Sauvage}}, \bibinfo {author} {\bibfnamefont {S.}~\bibnamefont {Sim}}, \bibinfo {author} {\bibfnamefont {A.~A.}\ \bibnamefont {Kunitsa}}, \bibinfo {author} {\bibfnamefont {W.~A.}\ \bibnamefont {Simon}}, \bibinfo {author} {\bibfnamefont {M.}~\bibnamefont {Mauri}},\ and\ \bibinfo {author} {\bibfnamefont {A.}~\bibnamefont {Perdomo-Ortiz}},\ }\bibfield  {title} {\bibinfo {title} {Flip: A flexible initializer for arbitrarily-sized parametrized quantum circuits},\ }\href {https://arxiv.org/abs/2103.08572} {\bibfield  {journal} {\bibinfo  {journal} {arXiv preprint arXiv:2103.08572}\ } (\bibinfo {year} {2021})}\BibitemShut {NoStop}%
\bibitem [{\citenamefont {Gibbs}\ \emph {et~al.}(2024)\citenamefont {Gibbs}, \citenamefont {Holmes},\ and\ \citenamefont {Stevenson}}]{gibbs2024exploiting}%
  \BibitemOpen
  \bibfield  {author} {\bibinfo {author} {\bibfnamefont {J.}~\bibnamefont {Gibbs}}, \bibinfo {author} {\bibfnamefont {Z.}~\bibnamefont {Holmes}},\ and\ \bibinfo {author} {\bibfnamefont {P.}~\bibnamefont {Stevenson}},\ }\bibfield  {title} {\bibinfo {title} {Exploiting symmetries in nuclear hamiltonians for ground state preparation},\ }\bibfield  {journal} {\bibinfo  {journal} {arXiv preprint arXiv:2402.10277}\ }\href {https://doi.org/10.48550/arXiv.2402.10277} {10.48550/arXiv.2402.10277} (\bibinfo {year} {2024})\BibitemShut {NoStop}%
\bibitem [{\citenamefont {Liu}\ \emph {et~al.}(2023)\citenamefont {Liu}, \citenamefont {Sun}, \citenamefont {Wu}, \citenamefont {Han},\ and\ \citenamefont {Guo}}]{liu2023mitigating}%
  \BibitemOpen
  \bibfield  {author} {\bibinfo {author} {\bibfnamefont {H.-Y.}\ \bibnamefont {Liu}}, \bibinfo {author} {\bibfnamefont {T.-P.}\ \bibnamefont {Sun}}, \bibinfo {author} {\bibfnamefont {Y.-C.}\ \bibnamefont {Wu}}, \bibinfo {author} {\bibfnamefont {Y.-J.}\ \bibnamefont {Han}},\ and\ \bibinfo {author} {\bibfnamefont {G.-P.}\ \bibnamefont {Guo}},\ }\bibfield  {title} {\bibinfo {title} {Mitigating barren plateaus with transfer-learning-inspired parameter initializations},\ }\href {https://doi.org/10.1088/1367-2630/acb58e} {\bibfield  {journal} {\bibinfo  {journal} {New Journal of Physics}\ }\textbf {\bibinfo {volume} {25}},\ \bibinfo {pages} {013039} (\bibinfo {year} {2023})}\BibitemShut {NoStop}%
\bibitem [{\citenamefont {Rudolph}\ \emph {et~al.}(2023)\citenamefont {Rudolph}, \citenamefont {Miller}, \citenamefont {Motlagh}, \citenamefont {Chen}, \citenamefont {Acharya},\ and\ \citenamefont {Perdomo-Ortiz}}]{rudolph2022synergy}%
  \BibitemOpen
  \bibfield  {author} {\bibinfo {author} {\bibfnamefont {M.~S.}\ \bibnamefont {Rudolph}}, \bibinfo {author} {\bibfnamefont {J.}~\bibnamefont {Miller}}, \bibinfo {author} {\bibfnamefont {D.}~\bibnamefont {Motlagh}}, \bibinfo {author} {\bibfnamefont {J.}~\bibnamefont {Chen}}, \bibinfo {author} {\bibfnamefont {A.}~\bibnamefont {Acharya}},\ and\ \bibinfo {author} {\bibfnamefont {A.}~\bibnamefont {Perdomo-Ortiz}},\ }\bibfield  {title} {\bibinfo {title} {Synergistic pretraining of parametrized quantum circuits via tensor networks},\ }\href {https://doi.org/10.1038/s41467-023-43908-6} {\bibfield  {journal} {\bibinfo  {journal} {Nature Communications}\ }\textbf {\bibinfo {volume} {14}},\ \bibinfo {pages} {8367} (\bibinfo {year} {2023})}\BibitemShut {NoStop}%
\bibitem [{\citenamefont {Grimsley}\ \emph {et~al.}(2023)\citenamefont {Grimsley}, \citenamefont {Mayhall}, \citenamefont {Barron}, \citenamefont {Barnes},\ and\ \citenamefont {Economou}}]{grimsley2022adapt}%
  \BibitemOpen
  \bibfield  {author} {\bibinfo {author} {\bibfnamefont {H.~R.}\ \bibnamefont {Grimsley}}, \bibinfo {author} {\bibfnamefont {N.~J.}\ \bibnamefont {Mayhall}}, \bibinfo {author} {\bibfnamefont {G.~S.}\ \bibnamefont {Barron}}, \bibinfo {author} {\bibfnamefont {E.}~\bibnamefont {Barnes}},\ and\ \bibinfo {author} {\bibfnamefont {S.~E.}\ \bibnamefont {Economou}},\ }\bibfield  {title} {\bibinfo {title} {Adaptive, problem-tailored variational quantum eigensolver mitigates rough parameter landscapes and barren plateaus},\ }\href {https://doi.org/10.1038/s41534-023-00681-0} {\bibfield  {journal} {\bibinfo  {journal} {npj Quantum Information}\ }\textbf {\bibinfo {volume} {9}},\ \bibinfo {pages} {19} (\bibinfo {year} {2023})}\BibitemShut {NoStop}%
\bibitem [{\citenamefont {Mele}\ \emph {et~al.}(2022)\citenamefont {Mele}, \citenamefont {Mbeng}, \citenamefont {Santoro}, \citenamefont {Collura},\ and\ \citenamefont {Torta}}]{mele2022avoiding}%
  \BibitemOpen
  \bibfield  {author} {\bibinfo {author} {\bibfnamefont {A.~A.}\ \bibnamefont {Mele}}, \bibinfo {author} {\bibfnamefont {G.~B.}\ \bibnamefont {Mbeng}}, \bibinfo {author} {\bibfnamefont {G.~E.}\ \bibnamefont {Santoro}}, \bibinfo {author} {\bibfnamefont {M.}~\bibnamefont {Collura}},\ and\ \bibinfo {author} {\bibfnamefont {P.}~\bibnamefont {Torta}},\ }\bibfield  {title} {\bibinfo {title} {Avoiding barren plateaus via transferability of smooth solutions in a {H}amiltonian variational ansatz},\ }\href {https://doi.org/10.1103/PhysRevA.106.L060401} {\bibfield  {journal} {\bibinfo  {journal} {Physical Review A}\ }\textbf {\bibinfo {volume} {106}},\ \bibinfo {pages} {L060401} (\bibinfo {year} {2022})}\BibitemShut {NoStop}%
\bibitem [{\citenamefont {Grant}\ \emph {et~al.}(2019)\citenamefont {Grant}, \citenamefont {Wossnig}, \citenamefont {Ostaszewski},\ and\ \citenamefont {Benedetti}}]{grant2019initialization}%
  \BibitemOpen
  \bibfield  {author} {\bibinfo {author} {\bibfnamefont {E.}~\bibnamefont {Grant}}, \bibinfo {author} {\bibfnamefont {L.}~\bibnamefont {Wossnig}}, \bibinfo {author} {\bibfnamefont {M.}~\bibnamefont {Ostaszewski}},\ and\ \bibinfo {author} {\bibfnamefont {M.}~\bibnamefont {Benedetti}},\ }\bibfield  {title} {\bibinfo {title} {An initialization strategy for addressing barren plateaus in parametrized quantum circuits},\ }\href {https://doi.org/10.22331/q-2019-12-09-214} {\bibfield  {journal} {\bibinfo  {journal} {Quantum}\ }\textbf {\bibinfo {volume} {3}},\ \bibinfo {pages} {214} (\bibinfo {year} {2019})}\BibitemShut {NoStop}%
\bibitem [{\citenamefont {Volkoff}\ and\ \citenamefont {Coles}(2021)}]{volkoff2021large}%
  \BibitemOpen
  \bibfield  {author} {\bibinfo {author} {\bibfnamefont {T.}~\bibnamefont {Volkoff}}\ and\ \bibinfo {author} {\bibfnamefont {P.~J.}\ \bibnamefont {Coles}},\ }\bibfield  {title} {\bibinfo {title} {Large gradients via correlation in random parameterized quantum circuits},\ }\href {https://doi.org/10.1088/2058-9565/abd89} {\bibfield  {journal} {\bibinfo  {journal} {Quantum Science and Technology}\ }\textbf {\bibinfo {volume} {6}},\ \bibinfo {pages} {025008} (\bibinfo {year} {2021})}\BibitemShut {NoStop}%
\bibitem [{\citenamefont {Zhang}\ \emph {et~al.}(2022)\citenamefont {Zhang}, \citenamefont {Liu}, \citenamefont {Hsieh},\ and\ \citenamefont {Tao}}]{zhang2022escaping}%
  \BibitemOpen
  \bibfield  {author} {\bibinfo {author} {\bibfnamefont {K.}~\bibnamefont {Zhang}}, \bibinfo {author} {\bibfnamefont {L.}~\bibnamefont {Liu}}, \bibinfo {author} {\bibfnamefont {M.-H.}\ \bibnamefont {Hsieh}},\ and\ \bibinfo {author} {\bibfnamefont {D.}~\bibnamefont {Tao}},\ }\bibfield  {title} {\bibinfo {title} {Escaping from the barren plateau via {G}aussian initializations in deep variational quantum circuits},\ }in\ \href {https://openreview.net/forum?id=jXgbJdQ2YIy} {\emph {\bibinfo {booktitle} {Advances in Neural Information Processing Systems}}}\ (\bibinfo {year} {2022})\BibitemShut {NoStop}%
\bibitem [{\citenamefont {Wang}\ \emph {et~al.}(2024)\citenamefont {Wang}, \citenamefont {Qi}, \citenamefont {Ferrie},\ and\ \citenamefont {Dong}}]{wang2023trainability}%
  \BibitemOpen
  \bibfield  {author} {\bibinfo {author} {\bibfnamefont {Y.}~\bibnamefont {Wang}}, \bibinfo {author} {\bibfnamefont {B.}~\bibnamefont {Qi}}, \bibinfo {author} {\bibfnamefont {C.}~\bibnamefont {Ferrie}},\ and\ \bibinfo {author} {\bibfnamefont {D.}~\bibnamefont {Dong}},\ }\bibfield  {title} {\bibinfo {title} {Trainability enhancement of parameterized quantum circuits via reduced-domain parameter initialization},\ }\href {https://doi.org/10.1103/PhysRevApplied.22.054005} {\bibfield  {journal} {\bibinfo  {journal} {Physical Review Applied}\ }\textbf {\bibinfo {volume} {22}},\ \bibinfo {pages} {054005} (\bibinfo {year} {2024})}\BibitemShut {NoStop}%
\bibitem [{\citenamefont {Park}\ and\ \citenamefont {Killoran}(2024)}]{park2023hamiltonian}%
  \BibitemOpen
  \bibfield  {author} {\bibinfo {author} {\bibfnamefont {C.-Y.}\ \bibnamefont {Park}}\ and\ \bibinfo {author} {\bibfnamefont {N.}~\bibnamefont {Killoran}},\ }\bibfield  {title} {\bibinfo {title} {Hamiltonian variational ansatz without barren plateaus},\ }\href {https://doi.org/10.22331/q-2024-02-01-1239} {\bibfield  {journal} {\bibinfo  {journal} {Quantum}\ }\textbf {\bibinfo {volume} {8}},\ \bibinfo {pages} {1239} (\bibinfo {year} {2024})}\BibitemShut {NoStop}%
\bibitem [{\citenamefont {Park}\ \emph {et~al.}(2024)\citenamefont {Park}, \citenamefont {Kang},\ and\ \citenamefont {Huh}}]{park2024hardware}%
  \BibitemOpen
  \bibfield  {author} {\bibinfo {author} {\bibfnamefont {C.-Y.}\ \bibnamefont {Park}}, \bibinfo {author} {\bibfnamefont {M.}~\bibnamefont {Kang}},\ and\ \bibinfo {author} {\bibfnamefont {J.}~\bibnamefont {Huh}},\ }\bibfield  {title} {\bibinfo {title} {Hardware-efficient ansatz without barren plateaus in any depth},\ }\href {https://arxiv.org/abs/2403.04844} {\bibfield  {journal} {\bibinfo  {journal} {arXiv preprint arXiv:2403.04844}\ } (\bibinfo {year} {2024})}\BibitemShut {NoStop}%
\bibitem [{\citenamefont {Shi}\ and\ \citenamefont {Shang}(2024)}]{shi2024avoiding}%
  \BibitemOpen
  \bibfield  {author} {\bibinfo {author} {\bibfnamefont {X.}~\bibnamefont {Shi}}\ and\ \bibinfo {author} {\bibfnamefont {Y.}~\bibnamefont {Shang}},\ }\bibfield  {title} {\bibinfo {title} {Avoiding barren plateaus via {G}aussian mixture model},\ }\href {https://arxiv.org/abs/2402.13501} {\bibfield  {journal} {\bibinfo  {journal} {arXiv preprint arXiv:2402.13501}\ } (\bibinfo {year} {2024})}\BibitemShut {NoStop}%
\bibitem [{\citenamefont {Zou}\ \emph {et~al.}(2025{\natexlab{b}})\citenamefont {Zou}, \citenamefont {Rahm}, \citenamefont {Kockum},\ and\ \citenamefont {Olsson}}]{zou2025generative}%
  \BibitemOpen
  \bibfield  {author} {\bibinfo {author} {\bibfnamefont {H.}~\bibnamefont {Zou}}, \bibinfo {author} {\bibfnamefont {M.}~\bibnamefont {Rahm}}, \bibinfo {author} {\bibfnamefont {A.~F.}\ \bibnamefont {Kockum}},\ and\ \bibinfo {author} {\bibfnamefont {S.}~\bibnamefont {Olsson}},\ }\bibfield  {title} {\bibinfo {title} {Generative flow-based warm start of the variational quantum eigensolver},\ }\href {https://doi.org/10.1038/s41534-025-01159-x} {\bibfield  {journal} {\bibinfo  {journal} {npj Quantum Information}\ }\textbf {\bibinfo {volume} {12}},\ \bibinfo {pages} {5} (\bibinfo {year} {2025}{\natexlab{b}})}\BibitemShut {NoStop}%
\bibitem [{\citenamefont {Ravi}\ \emph {et~al.}(2022)\citenamefont {Ravi}, \citenamefont {Gokhale}, \citenamefont {Ding}, \citenamefont {Kirby}, \citenamefont {Smith}, \citenamefont {Baker}, \citenamefont {Love}, \citenamefont {Hoffmann}, \citenamefont {Brown},\ and\ \citenamefont {Chong}}]{ravi2022cafqa}%
  \BibitemOpen
  \bibfield  {author} {\bibinfo {author} {\bibfnamefont {G.}~\bibnamefont {Ravi}}, \bibinfo {author} {\bibfnamefont {P.}~\bibnamefont {Gokhale}}, \bibinfo {author} {\bibfnamefont {Y.}~\bibnamefont {Ding}}, \bibinfo {author} {\bibfnamefont {W.}~\bibnamefont {Kirby}}, \bibinfo {author} {\bibfnamefont {K.}~\bibnamefont {Smith}}, \bibinfo {author} {\bibfnamefont {J.}~\bibnamefont {Baker}}, \bibinfo {author} {\bibfnamefont {P.}~\bibnamefont {Love}}, \bibinfo {author} {\bibfnamefont {H.}~\bibnamefont {Hoffmann}}, \bibinfo {author} {\bibfnamefont {K.}~\bibnamefont {Brown}},\ and\ \bibinfo {author} {\bibfnamefont {F.}~\bibnamefont {Chong}},\ }\bibfield  {title} {\bibinfo {title} {Cafqa: A classical simulation bootstrap for variational quantum algorithms},\ }\href {https://arxiv.org/abs/2202.12924} {\bibfield  {journal} {\bibinfo  {journal} {arXiv preprint arXiv:2202.12924}\ } (\bibinfo {year} {2022})}\BibitemShut {NoStop}%
\bibitem [{\citenamefont {Mitarai}\ \emph {et~al.}(2022)\citenamefont {Mitarai}, \citenamefont {Suzuki}, \citenamefont {Mizukami}, \citenamefont {Nakagawa},\ and\ \citenamefont {Fujii}}]{mitarai2022quadratic}%
  \BibitemOpen
  \bibfield  {author} {\bibinfo {author} {\bibfnamefont {K.}~\bibnamefont {Mitarai}}, \bibinfo {author} {\bibfnamefont {Y.}~\bibnamefont {Suzuki}}, \bibinfo {author} {\bibfnamefont {W.}~\bibnamefont {Mizukami}}, \bibinfo {author} {\bibfnamefont {Y.~O.}\ \bibnamefont {Nakagawa}},\ and\ \bibinfo {author} {\bibfnamefont {K.}~\bibnamefont {Fujii}},\ }\bibfield  {title} {\bibinfo {title} {Quadratic clifford expansion for efficient benchmarking and initialization of variational quantum algorithms},\ }\href {https://doi.org/10.1103/PhysRevResearch.4.033012} {\bibfield  {journal} {\bibinfo  {journal} {Physical Review Research}\ }\textbf {\bibinfo {volume} {4}},\ \bibinfo {pages} {033012} (\bibinfo {year} {2022})}\BibitemShut {NoStop}%
\bibitem [{\citenamefont {Wilson}\ \emph {et~al.}(2019)\citenamefont {Wilson}, \citenamefont {Stromswold}, \citenamefont {Wudarski}, \citenamefont {Hadfield}, \citenamefont {Tubman},\ and\ \citenamefont {Rieffel}}]{wilson2019optimizing}%
  \BibitemOpen
  \bibfield  {author} {\bibinfo {author} {\bibfnamefont {M.}~\bibnamefont {Wilson}}, \bibinfo {author} {\bibfnamefont {S.}~\bibnamefont {Stromswold}}, \bibinfo {author} {\bibfnamefont {F.}~\bibnamefont {Wudarski}}, \bibinfo {author} {\bibfnamefont {S.}~\bibnamefont {Hadfield}}, \bibinfo {author} {\bibfnamefont {N.~M.}\ \bibnamefont {Tubman}},\ and\ \bibinfo {author} {\bibfnamefont {E.}~\bibnamefont {Rieffel}},\ }\bibfield  {title} {\bibinfo {title} {Optimizing quantum heuristics with meta-learning},\ }\href {https://arxiv.org/abs/1908.03185} {\bibfield  {journal} {\bibinfo  {journal} {arXiv preprint arXiv:1908.03185}\ } (\bibinfo {year} {2019})}\BibitemShut {NoStop}%
\bibitem [{\citenamefont {Zhou}\ \emph {et~al.}(2020)\citenamefont {Zhou}, \citenamefont {Wang}, \citenamefont {Choi}, \citenamefont {Pichler},\ and\ \citenamefont {Lukin}}]{zhou2020quantum}%
  \BibitemOpen
  \bibfield  {author} {\bibinfo {author} {\bibfnamefont {L.}~\bibnamefont {Zhou}}, \bibinfo {author} {\bibfnamefont {S.-T.}\ \bibnamefont {Wang}}, \bibinfo {author} {\bibfnamefont {S.}~\bibnamefont {Choi}}, \bibinfo {author} {\bibfnamefont {H.}~\bibnamefont {Pichler}},\ and\ \bibinfo {author} {\bibfnamefont {M.~D.}\ \bibnamefont {Lukin}},\ }\bibfield  {title} {\bibinfo {title} {Quantum approximate optimization algorithm: Performance, mechanism, and implementation on near-term devices},\ }\href {https://doi.org/10.1103/PhysRevX.10.021067} {\bibfield  {journal} {\bibinfo  {journal} {Physical Review X}\ }\textbf {\bibinfo {volume} {10}},\ \bibinfo {pages} {021067} (\bibinfo {year} {2020})}\BibitemShut {NoStop}%
\bibitem [{\citenamefont {Egger}\ \emph {et~al.}(2021)\citenamefont {Egger}, \citenamefont {Mare{\v{c}}ek},\ and\ \citenamefont {Woerner}}]{egger2021warm}%
  \BibitemOpen
  \bibfield  {author} {\bibinfo {author} {\bibfnamefont {D.~J.}\ \bibnamefont {Egger}}, \bibinfo {author} {\bibfnamefont {J.}~\bibnamefont {Mare{\v{c}}ek}},\ and\ \bibinfo {author} {\bibfnamefont {S.}~\bibnamefont {Woerner}},\ }\bibfield  {title} {\bibinfo {title} {Warm-starting quantum optimization},\ }\href {https://doi.org/10.22331/q-2021-06-17-479} {\bibfield  {journal} {\bibinfo  {journal} {{Quantum}}\ }\textbf {\bibinfo {volume} {5}},\ \bibinfo {pages} {479} (\bibinfo {year} {2021})}\BibitemShut {NoStop}%
\bibitem [{\citenamefont {Verdon}\ \emph {et~al.}(2019)\citenamefont {Verdon}, \citenamefont {Broughton}, \citenamefont {McClean}, \citenamefont {Sung}, \citenamefont {Babbush}, \citenamefont {Jiang}, \citenamefont {Neven},\ and\ \citenamefont {Mohseni}}]{verdon2019learning}%
  \BibitemOpen
  \bibfield  {author} {\bibinfo {author} {\bibfnamefont {G.}~\bibnamefont {Verdon}}, \bibinfo {author} {\bibfnamefont {M.}~\bibnamefont {Broughton}}, \bibinfo {author} {\bibfnamefont {J.~R.}\ \bibnamefont {McClean}}, \bibinfo {author} {\bibfnamefont {K.~J.}\ \bibnamefont {Sung}}, \bibinfo {author} {\bibfnamefont {R.}~\bibnamefont {Babbush}}, \bibinfo {author} {\bibfnamefont {Z.}~\bibnamefont {Jiang}}, \bibinfo {author} {\bibfnamefont {H.}~\bibnamefont {Neven}},\ and\ \bibinfo {author} {\bibfnamefont {M.}~\bibnamefont {Mohseni}},\ }\bibfield  {title} {\bibinfo {title} {Learning to learn with quantum neural networks via classical neural networks},\ }\href {https://arxiv.org/abs/1907.05415} {\bibfield  {journal} {\bibinfo  {journal} {arXiv preprint arXiv:1907.05415}\ } (\bibinfo {year} {2019})}\BibitemShut {NoStop}%
\bibitem [{\citenamefont {Tate}\ \emph {et~al.}(2023)\citenamefont {Tate}, \citenamefont {Moondra}, \citenamefont {Gard}, \citenamefont {Mohler},\ and\ \citenamefont {Gupta}}]{tate2021classically}%
  \BibitemOpen
  \bibfield  {author} {\bibinfo {author} {\bibfnamefont {R.}~\bibnamefont {Tate}}, \bibinfo {author} {\bibfnamefont {J.}~\bibnamefont {Moondra}}, \bibinfo {author} {\bibfnamefont {B.}~\bibnamefont {Gard}}, \bibinfo {author} {\bibfnamefont {G.}~\bibnamefont {Mohler}},\ and\ \bibinfo {author} {\bibfnamefont {S.}~\bibnamefont {Gupta}},\ }\bibfield  {title} {\bibinfo {title} {Warm-started qaoa with custom mixers provably converges and computationally beats goemans-williamson's max-cut at low circuit depths},\ }\href {https://doi.org/10.22331/q-2023-09-26-1121} {\bibfield  {journal} {\bibinfo  {journal} {Quantum}\ }\textbf {\bibinfo {volume} {7}},\ \bibinfo {pages} {1121} (\bibinfo {year} {2023})}\BibitemShut {NoStop}%
\bibitem [{\citenamefont {Wurtz}\ and\ \citenamefont {Lykov}(2021)}]{wurtz2021fixed}%
  \BibitemOpen
  \bibfield  {author} {\bibinfo {author} {\bibfnamefont {J.}~\bibnamefont {Wurtz}}\ and\ \bibinfo {author} {\bibfnamefont {D.}~\bibnamefont {Lykov}},\ }\bibfield  {title} {\bibinfo {title} {Fixed-angle conjectures for the quantum approximate optimization algorithm on regular maxcut graphs},\ }\href {https://doi.org/10.1103/PhysRevA.104.052419} {\bibfield  {journal} {\bibinfo  {journal} {Physical Review A}\ }\textbf {\bibinfo {volume} {104}},\ \bibinfo {pages} {052419} (\bibinfo {year} {2021})}\BibitemShut {NoStop}%
\bibitem [{\citenamefont {Akshay}\ \emph {et~al.}(2021)\citenamefont {Akshay}, \citenamefont {Rabinovich}, \citenamefont {Campos},\ and\ \citenamefont {Biamonte}}]{akshay2021parameter}%
  \BibitemOpen
  \bibfield  {author} {\bibinfo {author} {\bibfnamefont {V.}~\bibnamefont {Akshay}}, \bibinfo {author} {\bibfnamefont {D.}~\bibnamefont {Rabinovich}}, \bibinfo {author} {\bibfnamefont {E.}~\bibnamefont {Campos}},\ and\ \bibinfo {author} {\bibfnamefont {J.}~\bibnamefont {Biamonte}},\ }\bibfield  {title} {\bibinfo {title} {Parameter concentrations in quantum approximate optimization},\ }\href {https://doi.org/10.1103/PhysRevA.104.L010401} {\bibfield  {journal} {\bibinfo  {journal} {Phys. Rev. A}\ }\textbf {\bibinfo {volume} {104}},\ \bibinfo {pages} {L010401} (\bibinfo {year} {2021})}\BibitemShut {NoStop}%
\bibitem [{\citenamefont {Niu}\ \emph {et~al.}(2023)\citenamefont {Niu}, \citenamefont {Zhang},\ and\ \citenamefont {Bao}}]{niu2023warm}%
  \BibitemOpen
  \bibfield  {author} {\bibinfo {author} {\bibfnamefont {Y.-F.}\ \bibnamefont {Niu}}, \bibinfo {author} {\bibfnamefont {S.}~\bibnamefont {Zhang}},\ and\ \bibinfo {author} {\bibfnamefont {W.-S.}\ \bibnamefont {Bao}},\ }\bibfield  {title} {\bibinfo {title} {Warm starting variational quantum algorithms with near clifford circuits},\ }\href {https://doi.org/10.3390/electronics12020347} {\bibfield  {journal} {\bibinfo  {journal} {Electronics}\ }\textbf {\bibinfo {volume} {12}},\ \bibinfo {pages} {347} (\bibinfo {year} {2023})}\BibitemShut {NoStop}%
\bibitem [{\citenamefont {Kirsch}(2019)}]{kirsch2019curie}%
  \BibitemOpen
  \bibfield  {author} {\bibinfo {author} {\bibfnamefont {W.}~\bibnamefont {Kirsch}},\ }\bibfield  {title} {\bibinfo {title} {The curie-weiss model--an approach using moments},\ }\bibfield  {journal} {\bibinfo  {journal} {arXiv preprint arXiv:1909.05612}\ }\href {https://doi.org/https://doi.org/10.48550/arXiv.1909.05612} {https://doi.org/10.48550/arXiv.1909.05612} (\bibinfo {year} {2019})\BibitemShut {NoStop}%
\bibitem [{\citenamefont {Arous}\ and\ \citenamefont {Zeitouni}(1999)}]{arous1999increasing}%
  \BibitemOpen
  \bibfield  {author} {\bibinfo {author} {\bibfnamefont {G.~B.}\ \bibnamefont {Arous}}\ and\ \bibinfo {author} {\bibfnamefont {O.}~\bibnamefont {Zeitouni}},\ }\bibfield  {title} {\bibinfo {title} {Increasing propagation of chaos for mean field models},\ }in\ \href {https://www.numdam.org/item/?id=AIHPB_1999__35_1_85_0} {\emph {\bibinfo {booktitle} {Annales de l'Institut Henri Poincare (B) Probability and Statistics}}},\ Vol.~\bibinfo {volume} {35}\ (\bibinfo {organization} {Elsevier},\ \bibinfo {year} {1999})\ pp.\ \bibinfo {pages} {85--102}\BibitemShut {NoStop}%
\bibitem [{\citenamefont {{\v{S}}amaj}(1988)}]{vsamaj1988improved}%
  \BibitemOpen
  \bibfield  {author} {\bibinfo {author} {\bibfnamefont {L.}~\bibnamefont {{\v{S}}amaj}},\ }\bibfield  {title} {\bibinfo {title} {An improved theory of higher-order correlations for an ising model with general spin},\ }\href {https://doi.org/https://doi.org/10.1016/0378-4371(88)90239-7} {\bibfield  {journal} {\bibinfo  {journal} {Physica A: Statistical Mechanics and its Applications}\ }\textbf {\bibinfo {volume} {153}},\ \bibinfo {pages} {530} (\bibinfo {year} {1988})}\BibitemShut {NoStop}%
\bibitem [{\citenamefont {Nietner}\ \emph {et~al.}(2025)\citenamefont {Nietner}, \citenamefont {Ioannou}, \citenamefont {Sweke}, \citenamefont {Kueng}, \citenamefont {Eisert}, \citenamefont {Hinsche},\ and\ \citenamefont {Haferkamp}}]{nietner2025average}%
  \BibitemOpen
  \bibfield  {author} {\bibinfo {author} {\bibfnamefont {A.}~\bibnamefont {Nietner}}, \bibinfo {author} {\bibfnamefont {M.}~\bibnamefont {Ioannou}}, \bibinfo {author} {\bibfnamefont {R.}~\bibnamefont {Sweke}}, \bibinfo {author} {\bibfnamefont {R.}~\bibnamefont {Kueng}}, \bibinfo {author} {\bibfnamefont {J.}~\bibnamefont {Eisert}}, \bibinfo {author} {\bibfnamefont {M.}~\bibnamefont {Hinsche}},\ and\ \bibinfo {author} {\bibfnamefont {J.}~\bibnamefont {Haferkamp}},\ }\bibfield  {title} {\bibinfo {title} {On the average-case complexity of learning output distributions of quantum circuits},\ }\href {https://doi.org/10.22331/q-2025-10-13-1883} {\bibfield  {journal} {\bibinfo  {journal} {Quantum}\ }\textbf {\bibinfo {volume} {9}},\ \bibinfo {pages} {1883} (\bibinfo {year} {2025})}\BibitemShut {NoStop}%
\bibitem [{\citenamefont {Bermejo}\ \emph {et~al.}(2024)\citenamefont {Bermejo}, \citenamefont {Braccia}, \citenamefont {Rudolph}, \citenamefont {Holmes}, \citenamefont {Cincio},\ and\ \citenamefont {Cerezo}}]{bermejo2024quantum}%
  \BibitemOpen
  \bibfield  {author} {\bibinfo {author} {\bibfnamefont {P.}~\bibnamefont {Bermejo}}, \bibinfo {author} {\bibfnamefont {P.}~\bibnamefont {Braccia}}, \bibinfo {author} {\bibfnamefont {M.~S.}\ \bibnamefont {Rudolph}}, \bibinfo {author} {\bibfnamefont {Z.}~\bibnamefont {Holmes}}, \bibinfo {author} {\bibfnamefont {L.}~\bibnamefont {Cincio}},\ and\ \bibinfo {author} {\bibfnamefont {M.}~\bibnamefont {Cerezo}},\ }\bibfield  {title} {\bibinfo {title} {Quantum convolutional neural networks are (effectively) classically simulable},\ }\bibfield  {journal} {\bibinfo  {journal} {arXiv preprint arXiv:2408.12739}\ }\href {https://doi.org/10.48550/arXiv.2408.12739} {10.48550/arXiv.2408.12739} (\bibinfo {year} {2024})\BibitemShut {NoStop}%
\bibitem [{\citenamefont {Angrisani}\ \emph {et~al.}(2025)\citenamefont {Angrisani}, \citenamefont {Schmidhuber}, \citenamefont {Rudolph}, \citenamefont {Cerezo}, \citenamefont {Holmes},\ and\ \citenamefont {Huang}}]{angrisani2024classically}%
  \BibitemOpen
  \bibfield  {author} {\bibinfo {author} {\bibfnamefont {A.}~\bibnamefont {Angrisani}}, \bibinfo {author} {\bibfnamefont {A.}~\bibnamefont {Schmidhuber}}, \bibinfo {author} {\bibfnamefont {M.~S.}\ \bibnamefont {Rudolph}}, \bibinfo {author} {\bibfnamefont {M.}~\bibnamefont {Cerezo}}, \bibinfo {author} {\bibfnamefont {Z.}~\bibnamefont {Holmes}},\ and\ \bibinfo {author} {\bibfnamefont {H.-Y.}\ \bibnamefont {Huang}},\ }\bibfield  {title} {\bibinfo {title} {Classically estimating observables of noiseless quantum circuits},\ }\href {https://doi.org/10.1103/lh6x-7rc3} {\bibfield  {journal} {\bibinfo  {journal} {Phys. Rev. Lett.}\ }\textbf {\bibinfo {volume} {135}},\ \bibinfo {pages} {170602} (\bibinfo {year} {2025})}\BibitemShut {NoStop}%
\bibitem [{\citenamefont {Lerch}\ \emph {et~al.}(2024)\citenamefont {Lerch}, \citenamefont {Puig}, \citenamefont {Rudolph}, \citenamefont {Angrisani}, \citenamefont {Jones}, \citenamefont {Cerezo}, \citenamefont {Thanasilp},\ and\ \citenamefont {Holmes}}]{lerch2024efficient}%
  \BibitemOpen
  \bibfield  {author} {\bibinfo {author} {\bibfnamefont {S.}~\bibnamefont {Lerch}}, \bibinfo {author} {\bibfnamefont {R.}~\bibnamefont {Puig}}, \bibinfo {author} {\bibfnamefont {M.}~\bibnamefont {Rudolph}}, \bibinfo {author} {\bibfnamefont {A.}~\bibnamefont {Angrisani}}, \bibinfo {author} {\bibfnamefont {T.}~\bibnamefont {Jones}}, \bibinfo {author} {\bibfnamefont {M.}~\bibnamefont {Cerezo}}, \bibinfo {author} {\bibfnamefont {S.}~\bibnamefont {Thanasilp}},\ and\ \bibinfo {author} {\bibfnamefont {Z.}~\bibnamefont {Holmes}},\ }\bibfield  {title} {\bibinfo {title} {Efficient quantum-enhanced classical simulation for patches of quantum landscapes},\ }\bibfield  {journal} {\bibinfo  {journal} {arXiv preprint arXiv:2411.19896}\ }\href {https://doi.org/10.48550/arXiv.2411.19896} {10.48550/arXiv.2411.19896} (\bibinfo {year} {2024})\BibitemShut {NoStop}%
\end{thebibliography}%
\onecolumngrid
\newpage

\appendix
{\huge \noindent \textbf{Appendix}}

\bigskip

\part{}
\parttoc 

\bigskip

\section{Notation table}
\label{app:notation-table}
\noindent
\begin{tabular}{ |p{2.4cm}|p{13.1cm}| }
\hline
\multicolumn{2}{|c|}{\textbf{Core objects (used throughout)}} \\
\hline
Symbol & Definition \\
\hline

$[n]$ & Index set $\{1,2,\dots,n\}$. \\

$A\subseteq [n]$ & Subset of qubit indices (support of an observable/correlator). \\

$A^c$ & Complement of $A$ in $[n]$, i.e. $[n]\setminus A$. \\

$\boldsymbol{z}\in\{0,1\}^n$ & Bit string (typically summed over {\rm mod} $2$). \\
$Z_i$ & Single Pauli-$Z$ on qubit $i\in[n]$ \\
$Z_A$ & Pauli-$Z$ string supported on $A$: $Z_A=\bigotimes_{i\in A} Z_i$. \\

$X_{\boldsymbol{b}}$ & Pauli-$X$ string labelled by $\boldsymbol{b}\in\{0,1\}^n$:
$X_{\boldsymbol{b}}=\bigotimes_{i=1}^n X_i^{b_i}$. \\

$U(\boldsymbol{\theta})$ & Parametrized IQP circuit with parameters $\boldsymbol{\theta}$. \\

$\theta_{\boldsymbol{b}}$ & Parameter associated with generator $X_{\boldsymbol{b}}$. \\
$m$ & Number of parameters (length of the parameter vector $\thv$).\\
$\mathcal{A}_A$ & Set of generator labels that anti-commute with $Z_A$ (architecture-dependent). \\

$z_A(\boldsymbol{\theta})$ & Model correlator for subset $A$:
$z_A(\boldsymbol{\theta})=\langle \mathbf{0}|U(\boldsymbol{\theta}) Z_A U^\dagger(\boldsymbol{\theta})|\mathbf{0}\rangle$. For a single qubit $j\in[n]$, we write $z_j(\thv):=z_{\{j\}}(\thv)$.\\

$t_A$ & Target correlator for subset $A$: $t_A=\Ebb_{\xv\sim p_{\rm data}}\left[(-1)^{\sum_{i\in A}x_i}\right]$. For a single qubit $j\in[n]$, we write $t_j:=t_{\{j\}}$. \\

$g_A^{(\alpha)}(\boldsymbol{\theta})$ & Partial derivative of the correlator w.r.t.\ parameter $\theta_\alpha$:
$g_A^{(\alpha)}(\boldsymbol{\theta})=\partial z_A(\boldsymbol{\theta})/\partial\theta_\alpha$.
(When unambiguous, written $g_A(\boldsymbol{\theta})$.) \\

$\mathcal{L}(\boldsymbol{\theta})$ & MMD loss:
$\mathcal{L}(\boldsymbol{\theta})=\sum_{A\subseteq[n]} w_A\,(z_A(\boldsymbol{\theta})-t_A)^2$. \\

$w_A$ & MMD weight for subset $A$, usually depends only on $|A|$.
(e.g. $w_A=p_\sigma^{|A|}(1-p_\sigma)^{n-|A|}$.) \\

$p_\sigma$ & Locality/kernel parameter used in the weights $w_A$: $p_{\sigma}=\frac{1-e^{-\frac{1}{2\sigma^2}}}{2}$. \\

$r$ & Patch size:
$\theta_{\boldsymbol{b}}\sim\theta_{\bv}^*+\mathrm{Unif}(-r,r)$. \\
$\thv^*$ & initialization parameters.
\\

\hline
\end{tabular}

\bigskip

\noindent
\begin{tabular}{ |p{2.4cm}|p{13.1cm}| }
\hline
\multicolumn{2}{|c|}{\textbf{Graph/topology notation (when using 1- and 2-qubit gates)}} \\
\hline
Symbol & Definition \\
\hline

$E$ & Edge set specifying allowed two-qubit interactions (connectivity graph). \\

$N_E(A)$ & External neighborhood of $A$:
$N_E(A)=\{k\in A^c:\exists j\in A \text{ s.t. } (j,k)\in E\}$. \\

$d_A$ & Effective Pauli light cone: $d_A = |A| + |N_E(A)|$. \\

\hline
\end{tabular}

\bigskip

\noindent
\begin{tabular}{ |p{2.4cm}|p{13.1cm}| }
\hline
\multicolumn{2}{|c|}{\textbf{Dictionary (subset vs bit string support)}} \\
\hline
Symbol & Definition \\
\hline

$\boldsymbol{\alpha}\leftrightarrow A$ &
We identify a subset $A\subseteq[n]$ with its indicator bit string
$\boldsymbol{\alpha}\in\{0,1\}^n$ defined by $\alpha_i=1 \Leftrightarrow i\in A$.
Conversely, $A=\{i:\alpha_i=1\}$. \\

\hline
\end{tabular}

\bigskip

\section{Preliminary notation and useful identities}
\label{app:preliminary-notation-identities}

This section serves as a preliminary discussion aimed at introducing notation and proving a lemma used in the later appendices. The arguments are straightforward and included mainly for completeness, so readers who are already familiar with the setup may prefer to skip this section and refer directly to the notation table in Appendix~\ref{app:notation-table}.

\subsection{General IQP notation}

Let us first consider the IQP circuits that include higher-order interaction terms, where each indicator bit string $\bv\in\{0,1\}^n$ is associated with the Pauli-$X$ rotation
\begin{equation}
    U_{\bv}(\theta_{\bv})=e^{i\th_{\bv} X_{\bv}}\;,
\end{equation}
where
\begin{equation}
    X_{\bv}=\bigotimes_{i=1}^nX_i^{b_i}\;\;.
\end{equation}
Notice that there are $2^n-1$ non-trivial rotations. Later we will restrict our attention to single-qubit and two-qubit rotations, but for now let us keep things generally. In this general set-up, the IQP circuit can be expressed as
\begin{equation}
\label{eq:IQP-circuit-general}
    U(\thv)=\prod_{\bv \in\{0,1\}^n} e^{i\th_{\bv} X_{\bv}}\;.
\end{equation}
The state evolved with the IQP circuit from an initial state $\ket{\vec{0}}$  is
\begin{equation}
    \ket{\psi(\thv)}=U^{\dagger}(\thv)\ket{\vec{0}}\;.
\end{equation}
We measure Pauli-$Z$ observables indexed by subsets of qubits $A\subseteq [n] = \{1,2,\dots,n\}$ on which they act non-trivially:
\begin{equation}
    Z_A=\bigotimes_{i\in A}Z_i\;,
\end{equation}
where $Z_i$ is the Pauli-$Z$ acting on qubit $i$. 
The model Pauli-$Z_A$ correlators are given by
\begin{equation}
    z_A(\thv)=\langle\psi(\thv)|Z_A|\psi(\thv)\rangle\; \;.
\end{equation}
We will simplify the notation for single qubit indices, i.e. when $A=\{j\}$, assuming $z_{\{j\}}(\thv)=z_j(\thv)$.
\subsection{Fourier expansion of model correlators}

We introduce the following lemma which shows the exact expression of the Pauli-$Z_A$ correlator.
\begin{lemma}
\label{lemma:correlators-fourier-expansion}
The correlator $z_A(\thv)$ with a general IQP circuit as in Eq.~\eqref{eq:IQP-circuit-general} can be written as
\begin{equation}
\label{eq:correlator-expansion-fourier-lemma}
z_A(\thv)
=\langle \vec{0}|U(\thv)Z_AU^\dagger(\thv)|\vec{0}\rangle=
\frac{1}{2^n}
\sum_{\zv\in\{0,1\}^n}
\exp\!\left(
2i\sum_{\bv\in\mathcal{A}_A}
\theta_{\bv} (-1)^{\bv\cdot\zv}
\right)\;\;,
\end{equation}
where $\mathcal{A}_A$ denotes the set of indicator strings corresponding to rotations that anti-commute with $Z_A$, i.e. $\{X_{\bv},Z_A\}=0$.
\end{lemma}

\begin{proof}
First, note that for a bit string $\bv\in\mathcal{A}_A$, the anti-commutation property implies
\begin{align}
    U_{\bv}(\th_{\bv})Z_AU^{\dagger}_{\bv}(\th_{\bv})
    &=e^{i\th_{\bv} X_{\bv}}Z_Ae^{-i\th_{\bv} X_{\bv}}\\
    &=e^{2i\th_{\bv} X_{\bv}}Z_A\;,
\end{align}
whereas if $\bv\notin\mathcal{A}_A$ (i.e., they commute) the corresponding rotation leaves $Z_A$ invariant. Therefore, we have
\begin{align}
    \langle \vec{0}|U(\thv)Z_AU^\dagger(\thv)|\vec{0}\rangle
    &= \langle\vec{0}|\left(\prod_{\bv\in\mathcal{A}_A}e^{2i\theta_{\bv}X_{\bv}}\right)Z_A|\vec{0}\rangle\\
    &=\langle\vec{0}|\left(\prod_{\bv\in\mathcal{A}_A}e^{2i\theta_{\bv}X_{\bv}}\right)|\vec{0}\rangle\;,\label{eq:rewrite_iqp_expval}
\end{align}
where we used $Z_A\ket{\vec{0}}=\ket{\vec{0}}$.

Now, using the Hadamard gate identity $X=HZH$, we write
\begin{equation}\label{eq:XtoZhaddamard}
    \prod_{\bv\in\mathcal{A}_A}e^{2i\theta_{\bv}X_{\bv}}
    =
    H^{\otimes n}\left(\prod_{\bv\in\mathcal{A}_A}e^{2i\theta_{\bv}Z_{\bv}} \right)H^{\otimes n}\;.
\end{equation}
Recalling that the action of the Hadamard gate into the zero state is 
\begin{equation}
    H^{\otimes n}\ket{\vec{0}}=|+\rangle^{\otimes n}
    =\frac{1}{2^{n/2}}\sum_{\zv\in\{0,1\}^n} |\zv\rangle\;,
\end{equation}
we then obtain
\begin{align}
    \langle \vec{0}|U(\thv)Z_AU^\dagger(\thv)|\vec{0}\rangle
    &=\langle\vec{0}|\left(\prod_{\bv\in\mathcal{A}_A}e^{2i\theta_{\bv}X_{\bv}}\right)|\vec{0}\rangle \label{eq:correlation_as_as_product_ofexp1}\\
    &=\frac{1}{2^n}\sum_{\yv,\zv\in\{0,1\}^n}\langle \yv|\left(\prod_{\bv\in\mathcal{A}_A}e^{2i\theta_{\bv}Z_{\bv}}\right)|\zv\rangle\\
    &=\frac{1}{2^n}\sum_{\yv,\zv\in\{0,1\}^n}\langle \yv|\zv\rangle\left(\prod_{\bv\in\mathcal{A}_A}e^{2i\theta_{\bv}(-1)^{\zv\cdot\bv}}\right)\\
    &=\frac{1}{2^n}\sum_{\zv\in\{0,1\}^n}\prod_{\bv\in\mathcal{A}_A}e^{2i\theta_{\bv}(-1)^{\zv\cdot\bv}}\label{eq:correlation_as_as_product_ofexp2}\;,
\end{align}
where in the first equality we used the identity derived in Eq.~\eqref{eq:rewrite_iqp_expval}, in the second equality we rewrote the $X_{\vec{b}}$ rotation in terms of Hadamard and $Z_{\vec{b}}$ rotations as in Eq.~\eqref{eq:XtoZhaddamard}, we also used $Z_{\bv}|\zv\rangle=(-1)^{\zv\cdot\bv}|\zv\rangle$ in the third equality  and $\langle \yv|\zv\rangle=\delta_{\yv,\zv}$ in the last equality.

Finally, we can rewrite the result as
\begin{equation}
    \prod_{\bv\in\mathcal{A}_A}e^{2i\theta_{\bv}(-1)^{\zv\cdot\bv}}
    =
    \exp\!\left(
    2i\sum_{\bv\in\mathcal{A}_A}\theta_{\bv}(-1)^{\bv\cdot\zv}
    \right)\;,
\end{equation}
which completes the proof.
\end{proof}

\section{General first and second moments of model correlators}

In this section, we compute the first and second moments of correlators $z_A(\thv)$ according to IQP circuits with $m$ parameters drawn as $\thv\sim {\rm Unif}\left([-r,r]^m\right)$, centered at the initialization point $\thv^* = \vec{0}$. These expressions are used in the next section to show that the full-angle range $r=\pi/2$ exhibits a barren plateau.

\subsection{First order expectation}

Let us consider a general IQP circuit architecture and let $\mathcal{A}_A$ be the set of generator labels that anti-commutes with $Z_A$ i.e. $\mathcal{A}_A=\{\bv\in\{0,1\}^n\,:\,\{X_{\bv},Z_A\}=0\}$. From Lemma~\ref{lemma:correlators-fourier-expansion}, the model correlator can be written as

\begin{equation} \label{eq:correlator-expansion-fourier}
z_A(\thv)
=\langle \vec{0}|U(\thv)Z_AU^\dagger(\thv)|\vec{0}\rangle=
\frac{1}{2^n}
\sum_{\zv\in\{0,1\}^n}
\exp\!\left(
2i\sum_{\bv\in\mathcal{A}_A}
\theta_{\bv} (-1)^{\bv\cdot\zv}
\right).
\end{equation}

 Let us consider each parameter is initialised around $0$ such that $\theta_{\bv}\sim{\rm Unif}(-r,r)$. The average of the correlator becomes
\begin{align}
    \Ebb[z_A(\thv)]&=\frac{1}{2^n}\sum_{\zv\in\{0,1\}^n}\Ebb\left[\exp\left(2i\sum_{\bv\in\mathcal{A}_A}\theta_{\bv}(-1)^{\bv\cdot\zv}\right)\right]\\
    &=\frac{1}{2^n}\sum_{\zv\in\{0,1\}^n}\prod_{\bv\in\mathcal{A}_A}\sinc(2r)\\
    &=\prod_{\bv\in\mathcal{A}_A}\sinc(2r)\;, \label{eq:general-first-oder-moment-correlator}
\end{align}
where averaging over parameters uses
\begin{equation}
\mathbb{E}_\theta[e^{i c \theta}] = \mathrm{sinc}(rc),
\qquad
\mathrm{sinc}(x)=\frac{\sin x}{x},
\end{equation}
together with even parity of $\sinc$ function which mean the $(-1)^{\bv\cdot\zv}$ dependence vanishes. So, for $r=\pi/2$, we have $\sinc(\pi)=0$ so all $z_A(\thv)$ expectations vanishes (except the identity which is not considered).
\subsection{Second order expectation}
For the second moment, we restrict our attention to IQP circuits of all $1$-qubit gates and $2$-qubit gates specified by a graph $E\subseteq\{(i,j)|i,j\in[n]\;{\rm and}\;i\neq j\}$:

\begin{equation}
    \label{eq:IQP-circuit-2qubits-appendix} U(\thv)=\exp\left(i\sum_{j=1}^n\theta_jX_j+i\sum_{\substack{j<k \\(j,k)\in E}}\theta_{jk}X_jX_k\right)\;.
\end{equation}
In this case, the gate generators that anti-commute with $Z_A$ are all $1$-qubit gates in $A$ i.e. $j\in A$ and all $2$-qubit gates in $E$ such that one qubit is in $A$ and the other in $A^c$ i.e. all $(i,j)\in E$ such  that $i\in A$ and $j\in A^c$. Therefore, from Lemma~\ref{lemma:correlators-fourier-expansion}, we can rewrite
\begin{align}
    z_A(\thv)^2&=\frac{1}{2^{2n}}\sum_{\boldsymbol{z},\boldsymbol{z'}\in\{0,1\}^n }\exp\left(2i\sum_{j\in A}\theta_{j}((-1)^{z_j}+(-1)^{z_j'})+2i\sum_{\substack{ j\in A, k\in A^c\\ (j,k)\in E}}\theta_{jk}((-1)^{z_j+z_k} +(-1)^{z_j'+z_k'})\right)\;.
\end{align}
Therefore, averaging over $\boldsymbol{\theta}$ leads to
\begin{align}
    \mathbb{E}_{\boldsymbol{\theta}}[z_A(\thv)^2]
    &=\frac{1}{2^{2n}}\sum_{\boldsymbol{z},\boldsymbol{z'}\in\{0,1\}^n }\prod_{j\in A}\text{sinc}(2r((-1)^{z_j}+(-1)^{z_j'}))\prod_{\substack{j\in A, k\in A^c \\ (j,k)\in E}}\text{sinc}(2r((-1)^{z_j+z_k} +(-1)^{z_j'+z_k'})) \\
    &=\frac{1}{2^{2n}}\sum_{\boldsymbol{z},\boldsymbol{z'}\in\{0,1\}^n }\prod_{j\in A}\text{sinc}(2r(1+(-1)^{z_j+z_j'}))\prod_{\substack{j\in A, k\in A^c \\ (j,k)\in E}}\text{sinc}(2r(1 +(-1)^{z_j+z_k+z_j'+z_k'})) \\
    &=\frac{1}{2^{n}}\sum_{\boldsymbol{z}\in\{0,1\}^n }\prod_{j\in A}\text{sinc}(2r(1+(-1)^{z_j}))\prod_{\substack{j\in A, k\in A^c \\ (j,k)\in E}}\text{sinc}(2r(1 +(-1)^{z_j+z_k}))\;,
\end{align}
where the second equality is due to the parity of sinc ($
\sinc(x) = \sinc(-x)$).
For the last equality, notice that for any $z_j'\in\{0,1\}$ we have $\sum_{z_j\in\{0,1\}}f(z_j+z_j' \text{ mod }  2)=\sum_{z_j\in\{0,1\}}f(z_j \text{ mod }  2)$.  Therefore, for any $\zv'\in\{0,1\}^{n}$, we have $\sum_{\zv}f(\zv+\zv' \;{\rm mod} \;2)=\sum_{\zv} f(\zv)$ since the change of variable $\zv+\zv'\to \zv$ is is a bijection of $\{0,1\}^{n}$ with addition mod $2$.  

Now, given that $\sinc(2r(1+(-1)^b))=\delta_{b,1}+\sinc(4r)\delta_{b,0}$ for $b\in\{0,1\}$, we have that 
\begin{equation}
    \prod_{j\in A}\text{sinc}(2r(1+(-1)^{z_j}))=\prod_{j\in A}\sinc(4r)^{\delta_{z_j,0}}\;,
\end{equation}
similarly for the second product we have 
\begin{equation}
   \prod_{\substack{j\in A, k\in A^c \\ (j,k)\in E}}\text{sinc}(2r(1 +(-1)^{z_j+z_k}))=\prod_{\substack{j\in A, k\in A^c \\ (j,k)\in E}}\sinc(4r)^{\delta_{z_k,z_j}}\;.
\end{equation}
Therefore, we can rewrite 
\begin{equation}\label{eq:second-moment-correlator-small-angle-sum-general}
\mathbb{E}_{\boldsymbol{\theta}}[z_A(\thv)^2]=\frac{1}{2^n}\sum_{\zv\in\{0,1\}^n}\left(\prod_{j\in A}\sinc(4r)^{\delta_{z_j,0}}\right)\left( \prod_{\substack{j\in A, k\in A^c \\ (j,k)\in E}}\sinc(4r)^{\delta_{z_k,z_j}}\right) \;.
\end{equation}

\section{Correlators analysis with random initialization}
\label{app:full-angle-correlator-variance}

\subsection{Summary of the correlator variance results: Topology dependence and concentration}

We now specialize our analysis to the full-angle regime ($r=\pi/2$), which corresponds to a standard random initialization. To ground this in a concrete architecture, we focus on the most common IQP circuit structure involving $1$- and $2$-qubit gates, as specified in Eq.~\eqref{eq:IQP-circuit-2qubits-appendix}.

We first compute the variance of the correlators $z_A(\thv)$ over the full-angle random initialization. Crucially, the correlator depends solely on the circuit architecture $U(\thv)$ and is independent of the target distribution $p_{\rm data}(\vec{z})$. For IQP circuits, we find that the correlator variance is strongly influenced by the topology of the two-body interaction graph underlying the circuit.

To characterize this relationship more precisely, consider an interaction graph $E$ associated with the set of two-qubit rotations in the IQP circuit defined in Eq.~\eqref{eq:IQP-circuit}. For any two qubits $i\neq j$, we say that $(i,j)\in E$ if qubits $i$ and $j$ are connected in the interaction graph. For any subset $A\subseteq[n]$, we define the external neighborhood of $A$ as
\begin{equation}\label{eq:external-neighbourhood}
    N_E(A):=\{i\in A^c:\exists j\in A \text{ such that } (i,j)\in E\}\;\;,
\end{equation}
where $A^c$ denotes the complement of $A$. Concretely, $N_E(A)$ consists of the qubits in $A^c$ that are connected to, or equivalently interact with, at least one qubit in $A$. This is illustrated with a simple example in Fig.~\ref{fig:circuit}, where the external neighbourhood of $A$ is highlighted in red. 

\begin{figure}
    \centering
    \includegraphics[width=0.8\linewidth]{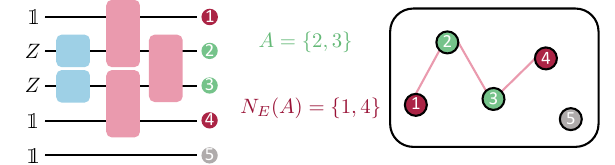}
    \caption{Illustration of the effective Pauli light cone. The subset $A$ is shown in green, and its external neighborhood $N_E(A)$ in red. The pink gates represent the interactions (equivalent to the edges of graph $E$), namely the $2$-qubits gates in the IQP circuit. Gray qubits lie outside $A \cup N_E(A)$. In this $5$ qubits example, $A=\{2,3\}$ and $N_E(A)=\{1,4\}$.}
    \label{fig:circuit}
\end{figure}
For a non-trivial Pauli string $Z_A$ where $A\neq\emptyset$, we define the \emph{effective Pauli light cone} as
\begin{equation}
    d_A := |A|+|N_E(A)| \;\;.
\end{equation}
Intuitively, $d_A$ represents the light-cone of $Z_A$ under Heisenberg evolution through the IQP circuit, counting the number of qubits that enter the causal support of the operator.

The following proposition gives the exact correlator moments in the full-angle random-initialization regime. 
\begin{proposition}[Exact correlator variance under full-angle random initialization]
\label{prop:correlator-variance-random-initialization-exact}
Let the parameters of the IQP circuit in Eq.~\eqref{eq:IQP-circuit} be independently drawn from ${\rm Unif}([-\pi/2,\pi/2])$.
For any nontrivial Pauli string $Z_A$ ($A\neq\emptyset$),
\begin{align}
    \Ebb_{\thv}[z_A(\thv)] & = 0 \;\;\;\;\;\;\;\;, \\
    \Var_{\thv}[z_A(\thv)] & = 2^{-d_A}\;\;\,,
\end{align}
where $d_A=|A|+|N_E(A)|$ is the effective Pauli light cone.
\end{proposition}
The proof of the proposition is given in Appendix~\ref{app:exact-correlator-variance-full-angle}.
One can see that the variance is dictated purely by the reach of the interaction graph. As an immediate consequence, we can bound this variance for $K$-regular interaction graphs by bounding $d_A$ in terms of $K$ leading to
\begin{corollary}[Architecture-dependent bounds for $K$-regular graphs]
\label{coro:correlator-variance-random-initialization-k-regular}
Under the same setting as in Proposition~\ref{prop:correlator-variance-random-initialization-exact}, if $E$ is a $K$-regular graph, then for every non-trivial $A\subseteq[n]$,
\begin{equation}
    2^{-\min(n,(K+1)|A|)}
    \le
    \Var_{\thv}[z_A(\thv)]
    \le
    2^{-\max(|A|,K+1)}.
\end{equation}
\end{corollary} 
\begin{figure}
    \centering
    \includegraphics[width=1\linewidth]{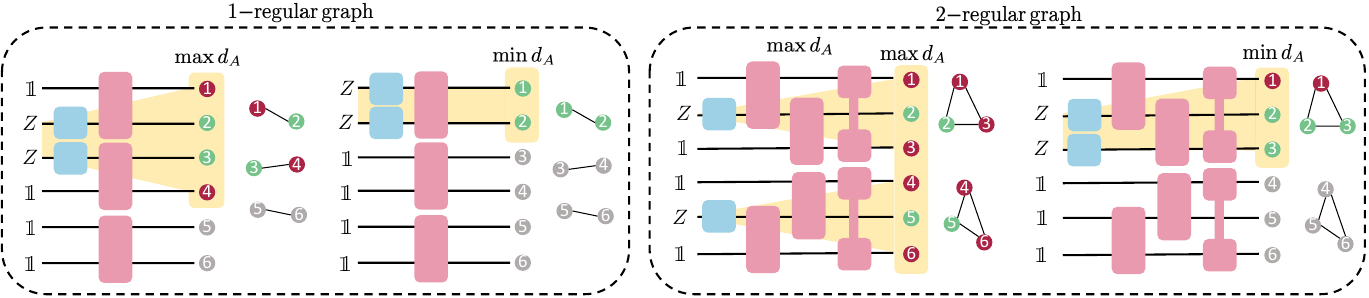}
    \caption{Representation of a $1$-regular and a $2$-regular graph with different observables such that the maximum and minimum $d_A$ are obtained.}
    \label{fig:kregularfigure}
\end{figure}
A representation of $1$-regular and $2$-regular graphs is shown in Fig.~\ref{fig:kregularfigure}.
That the correlator variance depends solely on $d_A$ implies the inevitable emergence of exponential concentration in highly connected graphs. For example, if the light-cone encompasses $\Theta(n)$ qubits, the variance vanishes exponentially with the system size. This is precisely the case for all-to-all connectivity, where the effective Pauli light cone saturates to $d_A = n$ for every non-trivial $A$, leading to immediate and uniform exponential concentration across all correlators. This effect is akin to the globality induced barren plateaus for unstructured circuits previously observed in Ref.~\cite{cerezo2020cost}. This is formalized in the following corollary.
\begin{corollary}[All-to-all full-angle correlator concentration]
\label{coro:correlator-variance-all-to-all-random-initialization}
Under the same setting as in Proposition~\ref{prop:correlator-variance-random-initialization-exact}, if $E$ is all-to-all, then for every non-trivial $A\subseteq[n]$,
\begin{equation}
    \Var_{\thv}[z_A(\thv)]=2^{-n}.
\end{equation}
\end{corollary}
\noindent The proof of Corollaries~\ref{coro:correlator-variance-random-initialization-k-regular} and~\ref{coro:correlator-variance-all-to-all-random-initialization} are given in Appendix~\ref{app:architecture-dependent-bound}.

Lastly, we show that under full-angle random initialization the cross term average between different correlators disappear. That is, we have
\begin{proposition}[Vanishing correlator cross terms]\label{prop:cross-term-vanish} Under the same setting as in Proposition~\ref{prop:correlator-variance-random-initialization-exact}, we have
\begin{align} \label{eq:cross-ter-zAzB}
     \Ebb_{\thv}[z_A(\thv)z_B(\thv)]
    =
    \delta_{A,B}\Ebb_{\thv}[z_A(\thv)^2],
\end{align}
where $\delta_{A,B}$ is the Kronecker delta, and
\begin{align}\label{eq:cross-term-zAzAzB}
    \Ebb_{\thv}[z_A(\thv)^2z_B(\thv)]=0 \;\;,
\end{align}
for every non-empty $B$.
\end{proposition}
\noindent The proof of this is given in Appendix~\ref{app:cross-term-vanish-proof}.

\subsection{Proof of Proposition~\ref{prop:correlator-variance-random-initialization-exact}: Exact correlator variance under full-angle random initialization}
\label{app:exact-correlator-variance-full-angle}

\begin{proof}
Considering full rotations for each parameter i.e. $r=\pi/2$, we have $\sinc(4r)=\sinc(2\pi)=0$. Therefore, Eq.~\eqref{eq:general-first-oder-moment-correlator} shows that the first order expectation vanishes for $\mathcal{A}_A\neq\emptyset$ (which is equivalent to $A\neq\emptyset$ here). This can be seen by noting that each single qubit gates acting on $A$ anti-commutes with $Z_A$ so $\mathcal{A}_A\neq\emptyset$ and thus
\begin{equation}
    \Ebb_{\thv}[z_A(\thv)]=0\;.
\end{equation}
For the second order expectation, the only non-vanishing terms in Eq.~\eqref{eq:second-moment-correlator-small-angle-sum-general} satisfy
\begin{equation}
   \sum_{j\in A}\delta_{z_j,0} = 0 \,\,\, \text{and}   \,\, \sum_{\substack{j\in A, k\in A^c \\ (j,k)\in E}}\delta_{z_j,z_k}=0\; .
\end{equation}
The first constraint means that each element $j\in A$ needs to satisfy  $z_j = 1$ (i.e. $\delta_{z_j,0}=0$). If that is fulfilled then, from the second constraint, we also need
\begin{equation}
    \sum_{\substack{j\in A, k\in A^c \\ (j,k)\in E}}\delta_{1,z_k}=0\;,
\end{equation}
which means that the elements in the complementary $k\in A^c$ such that $(j,k)\in E$ for some $j\in A$ needs to be $z_k = 0$. 

Let us count how many strings satisfy this condition. First, we must impose $z_j=1$ for every $j\in A$. Moreover, for every $k\in N_E(A)$ (i.e., every $k\in A^c$ such that there exists $j\in A$ with $(j,k)\in E$), we must have $z_k=0$. 
Hence the only indices that remain free are those in
\begin{equation}
A^c\setminus N_E(A)=\{\,k\in A^c\,:\,\ k\notin N_E(A)\,\}.
\end{equation}
All other coordinates are fixed, so the number of binary strings $\zv \in\{0,1\}^n$ satisfying the constraints is
\begin{equation}
2^{|A^c\setminus N_E(A)|}=2^{n-|A|-|N_E(A)|}\,,
\end{equation}
where we used that the number of qubits is $n$ and that if a qubit is in $A^c\setminus N_E(A)$ it cannot be on $A$ nor on $N_E(A)$ as well as the fact that these two sets have no overlap. Because of this, $|A^c\setminus N_E(A)| = n - |A|-|N_E(A)|$. Therefore, the variance of any given correlator (with $A\neq\emptyset$) is given by 
\begin{equation}
\Var_{\thv}[z_A(\thv)]=\mathbb{E}_{\boldsymbol{\theta}}[z_A(\thv)^2]=2^{-|A|-|N_E(A)|}=2^{-d_A}\;,
\end{equation}
where we define $d_A=|A|+|N_E(A)|$  (effective Pauli light cone). 
\end{proof}

\subsection{Proof of 
Corollaries~\ref{coro:correlator-variance-random-initialization-k-regular} and~\ref{coro:correlator-variance-all-to-all-random-initialization}: Architecture-dependent consequence}

\label{app:architecture-dependent-bound}
\begin{proof}
\underline{Let us first prove Corollary~\ref{coro:correlator-variance-random-initialization-k-regular}.}
Assume each qubit has $K$ nearest neighbours (i.e., the graph is $K$-regular): \begin{equation}
    |\{j\in [n]|(k,j)\in E\}|=K\;,\quad \forall k\in [n]\;.
    \end{equation}
    Fig.~\ref{fig:kregularfigure} shows a representation of $1$-regular and $2$-regular graphs for different observables.
     We first bound $|N_E(A)|$ in terms of $K$ in order to bound $d_A$. 

The largest value of $|N_E(A)|$ arises when each element of $A$ is connected to $K$ elements in $A^c$ which are all different such that $|N_E(A)|=K|A|$ provided $K|A|\leq |A^c|$ since $N_E(A)\subseteq A^c $ and equivalently $|N_E(A)|\leq |A^c|$. Therefore, we have $|N_E(A)|\leq \min(K|A|,|A^c|)$ which gives the following upper-bound for $d_A$:
\begin{equation}
    d_A\leq |A|+\min(K|A|,|A^c|)=\min((K+1)|A|,n)\;,
\end{equation}
where we used $|A^c|=n-|A|$ in the last equality.

For the lower-bound on $d_A$, we aim at lower-bounding the size of $|N_E(A)|$. If $|A|\leq K$ each element of $A$ can be connected to the remaining $|A|-1$ in $A$ so it is connected to at least $K-(|A|-1)$ elements in $|A^c|$. These elements can, in principle, be the same for each element of $A$, and therefore $|N_E(A)|\geq K-|A|+1$. If $|A|>K$, we could have $|N_E(A)|=0$ (case where elements in $A$ are not connected to any element in $A^c$). So, we have $|N_E(A)|\geq \max(0,K-|A|+1)$. Equivalently for $d_A$ we have
\begin{equation}
    d_A\geq |A|+\max(0,K-|A|+1)=\max(|A|,K+1)\;.
\end{equation}

Proposition~\ref{prop:correlator-variance-random-initialization-exact} then implies $\Var_{\thv}[z_A(\thv)]=2^{-d_A}$, so a lower-bound on $d_A$ gives an upper-bound on the variance, while an upper-bound on $d_A$ yields a lower-bound on the variance. This leads to the result of Corollary~\ref{coro:correlator-variance-random-initialization-k-regular}:
\begin{equation}
    2^{-n}\leq 2^{-\min(n,(K+1)|A|)}\leq \text{Var}_{\thv}[z_A(\thv)]=\Ebb_{\thv}[z_A(\thv)^2]\leq 2^{-\max(|A|,K+1)}\leq 2^{-|A|}\;.
\end{equation}
\underline{The proof of Corollary~\ref{coro:correlator-variance-all-to-all-random-initialization}} is quite straight-forward since for the all-to-all case we have $K=n-1$, so previous expression (i.e. Corollary~\ref{coro:correlator-variance-random-initialization-k-regular}) directly implies
\begin{equation}
    \Var_{\thv}[z_A(\thv)]=2^{-n}\;.
\end{equation} 
Alternatively, we can show it using Proposition~\ref{prop:correlator-variance-random-initialization-exact}. For all-to-all topology, every vertex in $A^c$ is connected to each vertex in $A$ (for $A\neq\emptyset$), so $N_E(A)=A^c$ and thus $|N_E(A)|=n-|A|$ (equivalently, $d_A=n$). 
\end{proof}

\subsection{Proof of Proposition~\ref{prop:cross-term-vanish}: Vanishing cross-expectations over full angle range}\label{app:cross-term-vanish-proof}

Before starting the formal proof, let us note that we could see that the expected values are zero just out of the parity of the functions. Indeed, $z_A(\thv)$ is a product of trigonometric functions with different angles. Because $\int_{-\pi/2}^{\pi/2} \cos(2\th)=0$ and $\int_{-\pi/2}^{\pi/2}\sin(2\th) = 0$, but the squares do not, the condition $\Ebb[z_A(\thv)z_B(\thv)]\not=0$ implies that $z_A(\thv)$ and $z_B(\thv)$ need to contain exactly the same elements, which means that $A = B$. Similarly, $\Ebb[z_A(\thv)^2z_B(\thv)]\not=0$ is equivalent to having integrals of the third power of the trigonometric functions, which also vanish. In that case, $z_B(\thv)$ can never compensate for this, so it will always be zero. We proceed now to rigorously prove this.

\begin{proof} 
\underline{Derivation of Eq.~\eqref{eq:cross-ter-zAzB}.} We consider $A\neq B$ and show that $\Ebb_{\thv}[z_A(\thv)z_B(\thv)]=0$. 

We write the correlators as in Lemma~\ref{lemma:correlators-fourier-expansion} to obtain
\begin{align}
  \Ebb_{\thv}[z_A(\thv)z_B(\thv)]&=\frac{1}{2^{2n}}\sum_{\zv,\zv'}  \Ebb_{\thv}\left[\exp\left(2i\left(\sum_{\bv\in\mathcal{A}_A}\theta_{\bv}(-1)^{\zv\cdot\bv}+\sum_{\bv\in\mathcal{A}_B}\theta_{\bv}(-1)^{\zv'\cdot\bv}\right)\right)\right]\label{eq:toreference_below_expvalue_to_sinc_prop3}\\
  &= \frac{1}{2^{2n}}\sum_{\zv,\zv'}  \prod_{\bv\in \mathcal{A}_A\cap\mathcal{A}_B}\Ebb_{\theta_{\bv}}\left[\exp\left(2i\theta_{\bv}\left((-1)^{\zv\cdot\bv}+(-1)^{\zv'\cdot\bv}\right)\right)\right]\\
  &\cdot\prod_{\bv\in \mathcal{A}_A\backslash\mathcal{A}_A\cap\mathcal{A}_B}\Ebb_{\theta_{\bv}}\left[\exp\left(2i\theta_{\bv}(-1)^{\zv\cdot\bv}\right)\right]\prod_{\bv\in \mathcal{A}_B\backslash\mathcal{A}_A\cap\mathcal{A}_B}\Ebb_{\theta_{\bv}}\left[\exp\left(2i\theta_{\bv}(-1)^{\zv'\cdot\bv}\right)\right]\\
  &=\frac{1}{2^{2n}}\sum_{\zv,\zv'}  \prod_{\bv\in \mathcal{A}_A\cap\mathcal{A}_B} \sinc(\pi(1+(-1)^{\bv\cdot(\zv+\zv')}))\prod_{\bv\in \mathcal{A}_A\cup\mathcal{A}_B\backslash \mathcal{A}_A\cap\mathcal{A}_B}\sinc(\pi)\label{eq:toreference_below_expvalue_to_sinc_prop3_v2}\\
  &=0\;, \label{eq:cross-term-vanishes-proofbelow-lasteq}
\end{align}
where in the second equality we used that $e^{i\sum_j x_j} = \prod_j e^{i x_j}$ and the fact that all parameters are independent. In the third equality, we used: 
\begin{align}
    \Ebb_{x\in[-r,r]}\left[e^{2i x(-1)^y}\right] &= \sinc(2r)\;,\\
    \Ebb_{x\in[-r,r]}\left[e^{2i x((-1)^y+(-1)^{y'}}\right] &= \sinc(2r(1+(-1)^{y+y'}))\;,
\end{align}
where $y,y'\in\{0,1\}$ for the terms in $\mathcal{A}_A\cup\mathcal{A}_B\backslash \mathcal{A}_A\cap\mathcal{A}_B$ and $\mathcal{A}_A\cap\mathcal{A}_B$ respectively, and we take $r = \pi/2$ because we are here considering the full angle range.
In Eq.~\eqref{eq:cross-term-vanishes-proofbelow-lasteq}, we use that $\sinc(\pi)=0$ so if there exists an element in $\mathcal{A}_A\cup\mathcal{A}_B\backslash \mathcal{A}_A\cap\mathcal{A}_B$, the cross expectation vanishes - and this has to be true because we have by assumption $A\neq B$.  Note that this does not require all-to-all topology, but single qubit rotations on each qubits, with a random $\thv$ in full angle range, is enough. If $A=B$, this corresponds to the second order moment which is computed above.

\medskip

\underline{Derivation of Eq.~\eqref{eq:cross-term-zAzAzB}.} Consider the following quantity
\begin{align}
  \Ebb_{\thv}[z_A(\thv)^2z_B(\thv)]&=\frac{1}{2^{3n}}\sum_{\zv,\zv',\zv''}  \Ebb_{\thv}\left[\exp\left(2i\left(\sum_{\bv\in\mathcal{A}_A}\theta_{\bv}((-1)^{\zv\cdot\bv}+(-1)^{\bv\cdot\zv'})+\sum_{\bv\in\mathcal{A}_B}\theta_{\bv}(-1)^{\zv''\cdot\bv}\right)\right)\right]\;.
\end{align}
Note that this quantity has a very similar form as the one above in Eq.~\eqref{eq:toreference_below_expvalue_to_sinc_prop3}. Thus we can use the same technique to analyze it. Indeed, if there exists an element $\bv\in\mathcal{A}_B\backslash \mathcal{A}_A$ it gives automatically $0$ after averaging w.r.t  to $\theta_{\bv}$, because this corresponds to the case in which we obtained a $\sinc(\pi)$ in the previous proof (Eq.~\eqref{eq:toreference_below_expvalue_to_sinc_prop3_v2}). Moreover, the terms that are in $\mathcal{A}_A\cap \mathcal{A}_B$ will also give $0$ since
\begin{align}
    \Ebb[e^{2i\theta_{\bv}(-1)^{\bv\cdot \zv''}}]&=0\;,\\
    \Ebb[e^{2i\theta_{\bv}((-1)^{\zv\cdot\bv}+(-1)^{\bv\cdot\zv'}+(-1)^{\bv\cdot \zv''})}]&=0\;,
\end{align}
where the second equality is valid since $((-1)^{\zv\cdot\bv}+(-1)^{\bv\cdot\zv'}+(-1)^{\bv\cdot \zv''})\in\{-3,-1,1,3\}$ so that $\sinc(\pi((-1)^{\zv\cdot\bv}+(-1)^{\bv\cdot\zv'}+(-1)^{\bv\cdot \zv''})=0$ always. So, the cross-expectation is non-vanishing if and only if $\mathcal{A}_A\cap \mathcal{A}_B=\emptyset$ and $\mathcal{A}_B\backslash \mathcal{A}_A=\emptyset$ i.e. if $B=\emptyset$ which is not allowed.

\end{proof}

\section{From correlator variance to MMD variance with full-angle random initialization}
\label{app:from-correlator-variance-to-MMD-variance-full-angle}
In this section, we relate the correlator variance results of the previous section to the variance of the MMD loss under full-angle random initialization. While we would intuitively think that this connection is straightforward, the non-linearity and data-dependent structure of the MMD makes the formal proof more fiddly. We first show that, for all-to-all IQP circuits, the exponential concentration of individual correlators implies exponential concentration of the full MMD loss. We then briefly comment on graph-dependent regimes in which full-angle random initialization need not lead to barren plateaus, although these lie outside the main scope of the present work.

\subsection{Exponential concentration of the MMD loss for random initialization and all-to-all IQP topology}
\label{app:exponential-concentration-MMD}
We focus on the all-to-all setting and show that the correlator concentration established in Corollary~\ref{coro:correlator-variance-all-to-all-random-initialization} implies exponential concentration of the MMD loss itself. This shows that full-angle random initialization over the entire parameter landscape is generically untrainable in this regime due to the barren plateau phenomenon.

\begin{theoremappendix}[Exponential concentration of the MMD loss, formal]
\label{thm:mmd-loss-concentration-random-alltoall-appendix}
For IQP circuits with all-to-all topology where parameters are independently and uniformly drawn from ${\rm Unif}([-\pi/2,\pi/2])$, the MMD loss satisfies
\begin{equation}
    \Var_{\thv}[\mathcal{L}(\thv)]\leq \frac{5}{2^n}\;.
\end{equation}
\end{theoremappendix}

\begin{proof}
We recall Corollary~\ref{coro:correlator-variance-all-to-all-random-initialization} that for random initialization and all-to-all IQP topology, we have that any correlator $z_A(\thv)$ exponentially concentrates
\begin{align}
    \Var_{\thv} [z_A(\thv)] = 2^{-n}\;,
\end{align}
for any non-trivial $A$.
Now, let us consider the full MMD loss that we recall
\begin{equation}
  \LC(\thv)  = \sum_{A \subseteq[n]} w_{A} (z_A(\thv) - t_A)^2\;,  
\end{equation}
where the trivial term $A=\emptyset$ is zero since $z_{\emptyset}(\thv)=t_{\emptyset}=1$.
Therefore we can compute the variance as,
\begin{align}
    \Var_{\thv}[\LC(\thv)] =& \Var_{\thv}[\sum_{A \subseteq[n]} w_{A} (z_A(\thv) - t_A)^2 ]\\
    =&\Var_{\thv}[\sum_{\substack{A \subseteq[n]\\ A\neq\emptyset}} w_{A} (z_A(\thv) - t_A)^2]\\
    =& \Var_{\thv}\left[\sum_{\substack{A \subseteq[n]\\ A\neq\emptyset}} w_{A} (z_A(\thv)^2 -2 t_Az_A(\thv))\right]\,,
\end{align}
were we just got rid of the constants by using that $\Var[f(x)+ c] = \Var[f(x)]$ (where $c$ is a constant). We can pick the expression up from here 
\begin{align}
    \Var_{\thv}[\LC(\thv)]&= \Var_{\thv}\left[\sum_{\substack{A \subseteq[n]\\ A\neq\emptyset}} w_{A} (z_A(\thv)^2 -2 t_Az_A(\thv))\right] \\
    &=\sum_{\substack{A,B \subseteq[n]\\ A,B\neq\emptyset}}w_Aw_B \left(\Cov[z_A(\thv)^2,z_B(\thv)^2]+4t_At_B\Cov[z_A(\thv),z_B(\thv)]-4t_B\Cov[z_A(\thv)^2,z_B(\thv)]\right)\\
    &=\sum_{\substack{A,B \subseteq[n]\\ A,B\neq\emptyset}}w_Aw_B \left(\Cov[z_A(\thv)^2,z_B(\thv)^2]+4t_A^2\Ebb[z_A(\thv)^2]\delta_{A,B}\right)\;, \label{eq:MMD-variance-full-angle-general-without-cross-terms}
\end{align}
where in the second equality we expanded the variance as a sum of covariances, i.e. $\Var[\sum_j X_j]=\sum_{i,j}\Cov[X_i,X_j]$. We used Proposition~\ref{prop:cross-term-vanish} together with $\Ebb_{\thv}[z_A(\thv)]=0$ in the last equality to get rid of the crossed terms. Now from Cauchy-Schwartz, we have 
\begin{equation}   \Cov[z_A(\thv)^2,z_B(\thv)^2]\leq\sqrt{\Var[z_A(\thv)^2]}\sqrt{\Var[z_B(\thv)^2]}\;.
\end{equation}
Additionally, we upper-bound the fourth moment by the second $\Ebb[z_A(\thv)^4]\leq \Ebb[z_A(\thv)^2]$ since $z_A(\thv)^2\leq 1$ so that $\Var[z_A(\thv)^2]\leq\Ebb[z_A(\thv)^4]\leq \Ebb[z_A(\thv)^2]=1/2^n$ and previous expression can be bounded as 
\begin{equation}\label{eq:upperbound_covar_squared}
  \Cov[z_A(\thv)^2,z_B(\thv)^2]\leq\sqrt{\Ebb[z_A(\thv)^2]}\sqrt{\Ebb[z_B(\thv)^2]}=1/2^n\;.
\end{equation}

Therefore, we can upper-bound the MMD variance as follows 
\begin{align}
    \Var_{\thv}[\LC(\thv)] &=\sum_{\substack{A,B \subseteq[n]\\ A,B\neq\emptyset}}w_Aw_B \left(\Cov[z_A(\thv)^2,z_B(\thv)^2]+4t_A^2\Ebb[z_A(\thv)^2]\delta_{A,B}\right)\\
    &\leq \sum_{\substack{A,B \subseteq[n]\\ A,B\neq\emptyset}}w_Aw_B \left(\frac{1}{2^n}+4\frac{1}{2^n}\delta_{A,B}\right)\\
    &=\frac{1}{2^n}\left((1-w_{\emptyset})^2+4\sum_{\substack{A \subseteq[n]\\ A\neq\emptyset}}w_A^2\right)\\
    &\leq \frac{5}{2^n}\;,
\end{align}
where in the first inequality we use that $t_A^2\leq 1$ as well as the upper-bound on the covariances in Eq.~\eqref{eq:upperbound_covar_squared}. In the second equality, we perform the sum over the weight which correspond to a probability distribution, so the sum of each element except $A=\emptyset$ is upper-bounded by $1$ as well as the sum of the square weights since $w_A^2\leq w_A$.  This completes the proof of the theorem.
\end{proof}

\subsection{Polynomially large MMD variance for full-angle regime with restricted topology}
\label{app:mmd-variance-lower-bound-full-angle-random}
We observe from Corollary~\ref{coro:correlator-variance-random-initialization-k-regular} that if $K|A|\in\mathcal{O}(\log(\poly(n)))$, then we can lower-bound the variance of the loss as
\begin{equation}
    \Var_{\thv}[z_A(\thv)]\in \Omega\!\left(\frac{1}{\poly(n)}\right).
\end{equation}

Indeed, for sufficiently sparse topologies and low-body observables, the correlator variance need not be exponentially small. To understand why this is true, let us consider as an example: In a $d$-dimensional lattice topology, where we have $K=2d$, if $d\in\order{1}$ (such as in $1$D and $2$D, i.e.  $K=2$ and $K=4$), and $|A|=\mathcal{O}(\log n)$-body correlators, then the variance under random initialization will be inverse-polynomial, rather than exponentially small. Intuitively, in the full-angle regime and for low-expressivity circuits, only low-body correlators carry non-negligible signal. Consequently, in this regime we should not expect the low-body MMD loss to exhibit exponential concentration for data whose low-order correlations remain inverse-polynomial in $n$. In particular, we derive the following theorem:

\begin{proposition}[MMD variance for random initialization under restricted topology]
\label{thm:mmd-trainability-random-kregular}
Consider an $n$-qubit IQP ansatz whose interaction graph is $K$-regular with
\begin{equation}
K \in \mathcal{O}(\log(\mathrm{poly}(n))).
\end{equation}
Let $\mathcal{L}(\boldsymbol{\theta})$ denote the low-body MMD loss using a Gaussian kernel with bandwidth
\begin{equation}
\sigma \in \Theta(\sqrt{n}).
\end{equation}
Assume that the target distribution has at least one non-vanishing low-body correlator: there exists a subset
$A \subseteq [n]$ with $|A|\in\mathcal{O}(1)$ such that the corresponding target correlator satisfies
\begin{equation}
|t_A| \in \Omega\!\left(\frac{1}{\mathrm{poly}(n)}\right).
\end{equation}
If the circuit parameters $\boldsymbol{\theta}$ are drawn i.i.d.\ uniformly from the full-angle range
$\mathrm{Unif}([-\pi/2,\pi/2])$, then the loss variance under random
initialization is lower bounded by an inverse polynomial:
\begin{equation}
\Var_{\boldsymbol{\theta}}\!\big[\mathcal{L}(\boldsymbol{\theta})\big]
\in \Omega\!\left(\frac{1}{\mathrm{poly}(n)}\right).
\end{equation}
\end{proposition}

\begin{proof}
Let us start by recalling that Corollary~\ref{coro:correlator-variance-random-initialization-k-regular} shows that the effect of full-angle random initialization depends strongly on the underlying interaction graph. This observation can be extended to obtain graph-dependent lower bounds on the MMD variance in corresponding regimes. 
For interested reader, notice that the MMD variance in Eq.~\eqref{eq:MMD-variance-full-angle-general-without-cross-terms} can be lower-bounded as follows

\begin{align}
        \Var_{\thv}[\LC(\thv)] &=\sum_{\substack{A,B \subseteq[n]\\ A,B\neq\emptyset}}w_Aw_B \left(\Cov[z_A(\thv)^2,z_B(\thv)^2]+4t_A^2\Ebb[z_A(\thv)^2]\delta_{A,B}\right)\\
        &\geq \sum_{\substack{A \subseteq[n]\\ A\neq\emptyset}}4w_A^2 t_A^2\Ebb[z_A(\thv)^2]\;, \label{eq:mmd-variance-lower-bound-full-angle-random}
\end{align}
where the lower-bound is obtained using $\Var[\sum_j X_j]=\sum_{i,j}\Cov[X_i,X_j]$ such that
\begin{equation}
    \sum_{\substack{A,B \subseteq[n]\\ A,B\neq\emptyset}}w_Aw_B \Cov[z_A(\thv)^2,z_B(\thv)^2]=\Var\left[\sum_{\substack{A \subseteq[n]\\ A\neq\emptyset}}w_Az_A(\thv)^2\right]\geq 0\;.
\end{equation}
 Notice that all terms in Eq.~\eqref{eq:mmd-variance-lower-bound-full-angle-random} are nonnegative, so it suffices to lower-bound a single dominant term.
 
 Since, in the low-body MMD regime (i.e. $\sigma\in\Theta(\sqrt{n})$), the weights satisfy $w_A\in \Theta(n^{-|A|})$, we restrict attention to subsets $A$ with $|A|\in\OC(1)$, for which $w_A\in\Theta(1/\poly(n))$.
Moreover, if $K|A|\in\order{\log({\rm poly}(n))}$, Corollary~\ref{coro:correlator-variance-random-initialization-k-regular} implies
\begin{equation}
 \Ebb[z_A(\thv)^2]\geq 2^{-\min(n,(K+1)|A|)}\in\Omega(1/\poly(n))\;.
\end{equation}
Hence, under the assumption$K\in\order{\log({\rm poly}(n))}$, it suffices to consider terms with $|A|\in\OC(1)$. If there exists at least one such subset $A$ for which $|t_A|\in\Omega(1/\poly(n))$, 
then the corresponding contribution to Eq.~\eqref{eq:mmd-variance-lower-bound-full-angle-random} is still inverse-polynomial in $n$. Therefore,
\begin{align}
    \Var_{\thv}[\LC(\thv)] \in\Omega\left(\frac{1}{{\poly}(n)}\right)\,,
\end{align}
since every term contributes nonnegatively to the variance. This concludes the proof of the theorem.
\end{proof}
Note that this proof only holds for low-body MMD since high-body terms are exponentially suppressed. Thus the variance guarantees for this case are only directly applicable to learning low-body data sets (or low-body approximations of a data set) and would likely fail to capture information about higher order correlations. 
\section{Correlator variance for identity initialization strategy}
\label{app:correlator-variance-restricted-angle-identity-initialization}

Having established that IQP-QCBMs exhibit exponential concentration over the full loss landscape, we now turn to the analysis of local subregions. In particular,
concentration over a large domain does not preclude the
existence of smaller regions with high curvature and only
inverse-polynomial loss variance, in which gradient signals
remain accessible. 

We first illustrate this phenomenon at the correlator level by examining patches centered at the identity. This serves mainly as an illustrative example to motivate further study on the MMD loss itself (see Section~\ref{sec:restricted-patch-can-avoid-BP} in the main and coming Appendix~\ref{app:patch-variance-curvature}).

Consider the identity strategy where $\thv^* = \vec{0}$ and each $\theta_\alpha$ is sampled independently from $\mathrm{Unif}([-r, r])$. We provide a lower-bound on the correlator variance around identity initialization as follows
\begin{proposition}[Non-vanishing correlator variance with identity initialization]
\label{prop:correlator-variance-around-identity}
Let $m_A$ denote the number of circuit generators that anti-commute with $Z_A$. If the patch half-width satisfies
\begin{equation}
    r^2\in\OC(\log(m_A)/m_A)\quad {\rm and}\quad r^2\in\Omega(1/\poly(m_A)),
\end{equation}
then
\begin{equation}
    \Var_{\thv}[z_A(\thv)]\in\Omega\!\left(\frac{1}{\poly(m_A)}\right).
\end{equation}
In particular, for all-to-all topology,
\begin{equation}
    m_A=|A|(n+1-|A|)\le \left(\frac{n+1}{2}\right)^2.
\end{equation}
\end{proposition}
Proposition~\ref{prop:correlator-variance-around-identity} demonstrates that the exponential concentration observed in the full-angle regime is not an intrinsic property of the IQP architecture alone. Instead, it  depends critically also on the initialization patch size. In sufficiently small patches around the identity, correlators can remain initially trainable.

Notably, by proving this proposition directly for the IQP architecture, we obtain a guaranteed width that is strictly larger than what is achievable through standard curvature techniques. While the latter provide more general bounds, this architecture-specific approach yields a tighter, more practical characterization of the initialization landscape.

\subsection{Proof of Proposition~\ref{prop:correlator-variance-around-identity}: Non-vanishing correlator variance with identity initialization}

Before delving into the main proof, let us consider the following Lemma which serves as a preliminary.
\begin{lemma} \label{lemma:lower-bound-product-cosine-variance}
    Let $r\in [0,1]$ and $m\in \mathbb{N}$. Therefore, we have
    \begin{equation}
        \left(\frac{1+\sinc(4r)}{2}\right)^{m}-\left(\sinc(2r)\right)^{2m}\geq \frac{16r^4m}{45}\exp(-cmr^2)\;,
    \end{equation}
    where $c=-2\ln(\sinc(2))\approx 1.58$.
\end{lemma}
\begin{proof}
 Let us identify $k_+=\left(\frac{1+\sinc(4r)}{2}\right)$ and $k_-=\sinc(2r)^2$  for simplification and notice that $k_+\geq k_-$ since $k_+-k_-=\Var_{\theta\sim \DC(0,r) }[\cos(2\theta)]\geq 0$ (recall that $\DC(0,r)= {\rm Unif}[ \VC(0,r)]$, as in Eq.~\eqref{eq:param-drawn}). We first lower-bound Eq.~\eqref{eq:var_za_around_identity-sinc-form} using geometric series as follow
 \begin{align}
     k_+^m-k_-^m&=(k_+-k_-)\sum_{k=0}^{m-1}k_+^{m-1-k}k_-^k\\
     &\geq (k_+-k_-)k_-^{m-1}m\;,\label{eq:eqtolowerbound_with_kpkm}
 \end{align}
 where the inequality is obtained by lower-bounding each term in the sum by the smallest one (i.e. the term $k=m-1$).
 Now, we use the Taylor expansion of sines and cosines together with the definitions of $k_\pm$: 
 \begin{align}
    k_+&=\frac{1+\sinc(4r)}{2}=1+\frac{1}{2}\sum_{k=1}^{\infty}\frac{(-1)^k(4r)^{2k}}{(2k+1)!}\;, \\
     k_-&=\frac{\sin^2(2r)}{4r^2}=\frac{1-\cos(4r)}{8r^2}=1+\sum_{k=1}^\infty \frac{(-1)^k(4r)^{2k}}{(2k+1)!(k+1)}\;.
 \end{align}
 So assuming $r\leq 1$, we first have
 \begin{equation}
     k_+-k_-=\sum_{k=2}^{\infty}\frac{(-1)^k(4r)^{2k}}{(2k+1)!}\left(\frac{1}{2}-\frac{1}{k+1}\right)=\sum_{k=0}^{\infty}\frac{(-1)^k(4r)^{2k+4}}{(2k+5)!}\left(\frac{k+1}{2(k+3)}\right)\geq \frac{16r^4}{45}\left(1-\frac{4r^2}{7}\right)\;,\label{eq:bound_kp_minus_km}
 \end{equation}
 where we respectively keep the first positive ($k=0$) and negative terms ($k=1$) for the lower-bound given that the magnitude of summands decreases and alternating sign. Specifically, we have that summing terms $k=2q$ and $k=2q+1$ for any $q\in \mathbb{N}$  gives a positive contribution:
 \begin{equation}
     \sum_{k=2q}^{2q+1}\frac{(-1)^k(4r)^{2k+4}}{(2k+5)!}\left(\frac{k+1}{2(k+3)}\right)=\frac{(4r)^{4q+4}(2q+1)}{(4q+5)!2(2q+3)}\left(1-\frac{(4r)^2(2q+2)(2q+3)}{(4q+6)(4q+7)(2q+1)(2q+4)}\right)\geq 0 \;,
 \end{equation}
 where in the inequality, we used the fact that for $r\leq 1$ and $q\geq 0$ we have
 \begin{equation}
     \frac{(4r)^2(2q+2)(2q+3)}{(4q+6)(4q+7)(2q+1)(2q+4)}\leq\frac{16(2q+2)(2q+3)}{42(2q+1)(2q+4)}=\frac{8}{21}\left(1+\frac{1}{2q^2+5q+2}\right)\leq \frac{4}{7}<1\;.
 \end{equation}
 where we used that $\max_{r\leq1, q\geq0}\frac{(4r)^2}{(4q+6)(4q+7)} = \frac{16}{42} = \frac{8}{21}$.

Now we can lower-bound Eq.~\eqref{eq:eqtolowerbound_with_kpkm} using Eq.~\eqref{eq:bound_kp_minus_km} to obtain
\begin{align}
    k_+^m - k_-^m \geq & (k_+-k_-)k_-^{m-1}m\\
    \geq & \frac{16r^4 m}{45}\left(1-\frac{4r^2}{7}\right)k_-^{m-1}\\\
    \geq &  \frac{16r^4 m}{45}\left(1-\frac{4r^2}{7}\right)e^{-c (m-1) r^2}
\end{align}
where in the last inequality we used that $k_-\geq e^{-r^2 c}$ for $x\in[0,1]$, and with $c=\max_{x\in[0,1]} -\frac{2\ln({\rm sinc}(2x))}{x^2} = -2\ln[\sinc(2)]\approx 1.58$. Similarly, to have a compact expression we can see that $1-\frac{4x^2}{7}\geq e^{-\tilde{c} x^2}$, and we find that $\tilde{c} = \max_{x\in[0,1]}-\frac{\ln[1-4x^2/7]}{x^2} = \ln[7/3]\approx0.85$. Thus we can finally lower bound the previous expression as
\begin{align}
    k_+^m - k_-^m \geq &  \frac{16r^4 m}{45}\left(1-\frac{4r^2}{7}\right)e^{-c m r^2}\\
    \geq & \frac{16r^4 m}{45}e^{-c m r^2}
\end{align}
where we used that $\tilde{c}<c$ and thus $e^{-cx^2}<e^{-\tilde{c}x^2}$ which completes the proof of the Lemma.
\end{proof}

We are now ready to prove Proposition~\ref{prop:correlator-variance-around-identity}

\begin{proof} 
The first steps are similar to those in Lemma~\ref{lemma:correlators-fourier-expansion}, but we restate each individual step here for completeness. We begin with a compact expression for a general circuit, and then specialize to the case of single-qubit and two-qubit gates with all-to-all topology. For a given $A$, we identify bit string $\av$ s.t. $a_j=1$ for all $j\in A$ otherwise $a_j=0$. So that all the gates that anti-commutes with $Z_A$ are analogously identified by $\bv$ such that $\av\cdot \bv=1$ and let $\mathcal{A}_A$ be the label set of anti-commuting elements i.e. $\mathcal{A}_A=\{\bv|\av\cdot \bv=1\}$. Using this notation we can write the correlator $z_A(\thv)$ as:

\begin{align}
    \langle \vec{0}|U(\thv)Z_AU^\dagger(\thv)|\vec{0}\rangle&= \langle\vec{0}|\left(\prod_{\bv\in\mathcal{A}_A}e^{2i\theta_{\bv}X_{\bv}}\right)Z_A|\vec{0}\rangle\\
    &=\langle\vec{0}|\left(\prod_{\bv\in\mathcal{A}_A}e^{2i\theta_{\bv}X_{\bv}}\right)|\vec{0}\rangle\\
    &=\langle\vec{0}|\prod_{\bv\in\mathcal{A}_A}\left(\cos(2\theta_{\bv})+i\sin(2\theta_{\bv})X_{\bv}\right)|\vec{0}\rangle\\
    &=\sum_{S\subseteq \mathcal{A}_A}\left(\prod_{\bv\in \mathcal{A}_A\backslash S} \cos(2\theta_{\bv})\right)\left(\prod_{\bv\in S}\sin(2\theta_{\bv})\right)i^{|S|}\langle\vec{0}|X_{\sum_{\bv\in S}\bv}|\vec{0}\rangle\\
    &=\sum_{S\subseteq \mathcal{A}_A}\left(\prod_{\bv\in \mathcal{A}_A\backslash S} \cos(2\theta_{\bv})\right)\left(\prod_{\bv\in S}\sin(2\theta_{\bv})\right)i^{|S|}\prod_{j=1}^n\left(\frac{1+(-1)^{\sum_{\bv\in S}b_j}}{2}\right)\;, \label{eq:z_a-as-sine-cosine-series}
\end{align}
 where we first consider anti-commuting relations, then $Z_A|\vec{0}\rangle =|\vec{0}\rangle$. Then, we expand the exponential as a sum of sine and cosine and rewrite the product of the sum as a sum over the product as we did in Eqs.~(\ref{eq:correlation_as_as_product_ofexp1}-\ref{eq:correlation_as_as_product_ofexp2}). Finally, in the last equality we notice the terms that are non-zero are those the exponent of the $-1$ term is zero (that corresponds to the product of $X$ being equal to identity), i.e.
 \begin{equation}
     \prod_{j=1}^n\left(\frac{1+(-1)^{\sum_{\bv\in S}b_j}}{2}\right)=\begin{cases}
         1 &{\rm if }\;  \sum_{\bv\in S}\bv=\vec{0}\; {\rm mod}\; 2\,,\\
         0 & {\rm else}\,.
     \end{cases}
 \end{equation}
 so remaining terms must satisfy $\sum_{\bv\in S}b_j=0$ mod $2$ for all $j\in [n]$.

Notice that $|S|$ is even for all $S\subseteq\mathcal{A}_A$ satisfying $\sum_{\bv\in S}\bv=\vec{0}$ mod $2$. Now, let us consider $1$- and $2$-qubits gate IQP circuit with all-to-all topology. For the single qubit gates, we denote $x_j\in\{0,1\}$ as the bit indicating whether $S$ contains the single qubit gate that acts on the qubit $j$ or not. Notice that we only consider anti-commuting gates since $S\subseteq\mathcal{A}_A$ so $j\in A$. Similarly for the $2$-qubit gates, $x_{jk}$ indicates whether $S$ contains the $2$-qubit gate associated with qubits $j\in A$ and $k\in A^c$ or not (since anti-commuting $2$-qubit gates needs one element in $A$ and one in $A^c$). Let us make this more explicit by defining $\vec{e}_j$ as the unit vector with component $1$ in the $j$-th element (so other components are set to $0$). Therefore, $\forall j\in A$ we have $x_j=1$ if $\vec{e}_j\in S$ and $x_j=0$ otherwise. Similarly for the $2$-qubit gate, $\forall j\in A$ and $\forall k\in A^c$ we have $x_{jk}=1$ if $\vec{e}_j+\vec{e}_k\in S$  and $x_{jk}=0$ otherwise.
Therefore, we have
\begin{equation} \label{eq:sum_b_in_s-proof}
    \vec{v}:=\sum_{\bv\in S}\bv=\sum_{j\in A}x_j\vec{e}_j+\sum_{j\in A}\sum_{k\in A^c}x_{jk}(\vec{e}_j+\vec{e}_k)=\vec{0} \text{ mod }2\;,
\end{equation}
where we use the notation $\vec{v}$ for ease of notation and $\vec{v}$ must be equal to $\vec{0}$ mod $2$. This leads to a system of $n$ equations obtained by taking the scalar product with each $\vec{e}_j$ for all $j\in [n]$ (i.e. one equation per component in Eq.~\eqref{eq:sum_b_in_s-proof}). Therefore, for each $\vec{e}_j$ with $j\in A$, this reads as
\begin{equation}
   \vec{e}_j\cdot\vec{v}= x_j+\sum_{k\in A^c}x_{jk}=0 \text{ mod } 2 \;.
\end{equation} 
And for each $\vec{e}_k$, $k\in A^c$, we have 
\begin{equation}
    \vec{e}_k\cdot\vec{v}=\sum_{j\in A}x_{jk}=0\text{ mod } 2\;.
\end{equation}
Thus, the cardinality of $S$, which is simply the sum of all $x_j$ with $j\in A$ and $x_{jk}$ for all $j\in A$ and all $k\in A^c$, is even:
\begin{equation}
    |S|=\sum_{j\in A}x_j+\sum_{j\in A}\sum_{k\in A^c}x_{jk}=0 \text{ mod } 2\;.
\end{equation} 
Therefore, $i^{|S|}\in\{-1,+1\}$ and we only have real terms in the sum of Eq.~\eqref{eq:z_a-as-sine-cosine-series}. 
 
 Now, we have a sum of orthogonal trigonometric functions (Pauli paths orthogonality with symmetric distribution) with coefficients $\pm1$ so the second moment with perturbation $r$ around $\vec{0}$ can be written as  
 \begin{equation}
     \Ebb_{\thv}[z_A(\thv)^2]=\sum_{S\subseteq \mathcal{A}_A}\left( \Ebb_{\theta}[\cos^2(2\theta)]\right)^{|\mathcal{A}_A|-|S|}\left(\Ebb_{\theta}[\sin^2(2\theta)]\right)^{|S|} \prod_{j=1}^n\left(\frac{1+(-1)^{\sum_{\bv\in S}b_j}}{2}\right)\;,
 \end{equation}
 and the first order expectation kills all the sines terms i.e. only $S=\emptyset$ survives:
 \begin{equation}
    \Ebb_{\thv}[z_A(\thv)] = \Ebb_{\theta}[\cos(2\theta)]^{|\mathcal{A}_A|}.
 \end{equation}
 Now, notice that $\Ebb_\theta[\cos^2(2\theta)]\geq \Ebb_\theta[\cos(2\theta)]^2$ (due to positivity of variance). Moreover, all the terms in the second moment are positive, so by keeping only the term $S=\emptyset$ (which is the largest with identity initialization) in the second moment, we can still get a positive variance lower-bound given by
 \begin{align}
     \Var_{\thv}[z_A(\thv)]&\geq \Ebb_\theta[\cos^2(2\theta)]^{|\mathcal{A}_A|}-\Ebb_{\theta}[\cos(2\theta)]^{2|\mathcal{A}_A|}\\
     &= \left(\frac{1+\sinc(4r)}{2}\right)^{|\mathcal{A}_A|}-\left(\sinc(2r)^2\right)^{|\mathcal{A}_A|}\;. \label{eq:var_za_around_identity-sinc-form}
 \end{align}

Therefore, we can use  Lemma~\ref{lemma:lower-bound-product-cosine-variance} given that $r\leq 1$ to obtain
\begin{equation}
    \Var_{\thv}[z_A(\thv)]\geq \frac{16r^4m_A}{45}\exp(-cm_Ar^2)\;,
\end{equation}
where we defined $m_A=|\mathcal{A}_A|$ as the number of gates affecting $Z_A$ non-trivially.
Now, from the assumptions on $r$, we have
\begin{equation}
   r^2\in\OC(\log(m_A)/m_A)\implies \exp(-cm_Ar^2)\in\Omega(1/\poly(m_A))\;,
\end{equation}
together with
\begin{equation}
     r^2\in\Omega(1/\poly(m_A))\implies r^4m_A\in\Omega(1/\poly(m_A))\;, 
\end{equation}
which finally leads to:
\begin{equation}
    \Var_{\thv}[z_A(\thv)]\in\Omega(1/\poly(m_A))\;.
\end{equation}

Although this bound is still quite general, for the all-to-all case, we will have $|A|$ single qubit gates and $|A|\cdot|A^c|=|A|(n-|A|)$ two qubit gates that anti-commute with a given $Z_A$. Therefore, we have $m_A=|\mathcal{A}_A|=|A|(1+n-|A|)\leq\left(\frac{n+1}{2}\right)^2$ which completes the proof of the proposition. 
\end{proof}

\section{Evading barren plateaus with small-angle initializations: General non-linear loss and MMD loss variance lower bounds}
\label{app:patch-variance-curvature}
A unifying approach to investigating the loss variance in the vicinity of a fixed point is to analyze its local curvature, i.e., the second-order derivatives at that point. This analytical framework has been introduced and effectively deployed for loss functions based on standard expectation values in Ref.~\cite{mhiri2025unifying} analyzing various circuit architectures. However, for certain applications (including those in quantum machine learning) we need to consider non-linear losses. 

In this section, we extend curvature-based analyses to a broad class of general non-linear loss functions, with a specific focus on the MMD loss used in generative modeling. In particular, we provide a formal version of Theorem~\ref{thm:MMD-variance-patch-lower-bound-curvature-informal} together with its proof. 

\begin{theoremappendix}[Lower bound guarantee of an arbitrary non-linear loss and MMD loss with sufficient curvature, formal]\label{thm-sup:patch-variance-double-derivative}
Consider an arbitrary non-linear loss $\LC(\thv)$ and any quantum process parametrized with the set of parameters $\thv$. Further suppose that the parameters are uniformly distributed over a hypercube patch with the patch center at $\thv^*$ and the patch half-width $r$
\begin{equation}
\thv\sim \thv^*+\mathrm{Unif}\big([-r,r]^m\big) \;.
\end{equation}

Assume that there exists a parameter in the set $\theta_\alpha \in \thv$ such that the following derivative conditions are satisfied
\begin{enumerate}
    \item For all $k\ge 1$ the $2k$-order derivatives withe respect to $\theta_\alpha$ satisfies
    \begin{align}\label{eq:assumption-even-deriv-bound-thm}
\left|\left.\left(\frac{\partial^{2k}\mathcal{L}(\thv)}{\partial\theta_\alpha^{2k}}\right)\right|_{\theta_\alpha=\theta_\alpha^*}\right|
\le a\,\gamma^{2k} \;\;,
    \end{align}
    for some constants $a>0$ and $\gamma>0$ such that $r\le \frac{3}{2\gamma}$.
    \item The mixed fourth derivatives follow
\begin{equation}\label{eq:assumption-fourth-deriv-bound-thm}
\left|\frac{\partial^4 \mathcal{L}(\thv)}{\partial\theta_\alpha^2\,\partial\theta_j^2}\right|
\le \gamma_j\;\;,
\end{equation}
for all $j\neq \alpha$ with some constants $\gamma_j$.
\end{enumerate}

Denote the local curvature at the patch center as
\begin{equation}
c_\alpha(\thv^*):=\left|\left.\left(\frac{\partial^2 \mathcal{L}(\thv)}{\partial\theta_\alpha^2}\right)\right|_{\thv=\thv^*}\right|\;\;,
\end{equation}
and further denote 
\begin{equation}
\beta_1=\sum_{j\neq\alpha}\frac{\gamma_j}{6}\quad {\rm and}\quad \beta_2 := \frac{a^2\gamma^6}{6} \;\;.
\end{equation}

\smallskip

If $r$ satisfies
\begin{equation}\label{eq:r-condition-thm}
r^2 \le \Delta\,\frac{c_\alpha(\thv^*)^2}{2\beta_1 c_\alpha(\thv^*)+\beta_2},
\end{equation}
with some constant $\Delta\in(0,1)$, then the loss variance can be lower bounded as\begin{equation}\label{eq:variance-lb-thm}
\Var_{\thv}\!\big[\mathcal{L}(\thv)\big]
\;\ge\; \Ebb_{\thv}\!\big[\Var_{\theta_\alpha}[\mathcal{L}(\thv)]\big]
\;\ge\; (1-\Delta)\,\frac{r^4}{45}\,c_\alpha(\thv^*)^2.
\end{equation}

In particular, the (local) MMD loss in Eq.~\eqref{eq:MMD-loss} and the IQP circuit in Eq.~\eqref{eq:IQP-circuit} with single- and two-qubit gates satisfy the derivative conditions in Eq.~\eqref{eq:assumption-even-deriv-bound-thm} and Eq.~\eqref{eq:assumption-fourth-deriv-bound-thm} with
\begin{equation}
a=4p_\sigma(1-p_\sigma),\qquad \gamma=4,\qquad 
\beta_1=192(m-1)p_\sigma^2(1-p_\sigma)^2,\qquad
\beta_2=\frac{a^2\gamma^6}{6},\qquad 0\leq \Delta\leq 3/8 \;\;.
\end{equation}
Lastly, if the curvature is at least polynomially large
\begin{align}
    c_\alpha(\thv^*) \in \Omega\left(\frac{1}{\poly(n)}\right) \;\;,
\end{align}
then there exists a patch with half-width $r\in \mathcal{O}(1/\poly(n))$ such that
\begin{equation}
\Var_{\thv}[\mathcal{L}(\thv)]\in \Omega\!\left(\frac{1}{\poly(n)}\right).
\end{equation}

\end{theoremappendix}

For the rest of this appendix, we provide the proof of the theorem. The technical development of the proof is structured below.
\begin{itemize}
    \item Appendix~\ref{app:proof-prelimm-non-linear-loss} (Technical Preliminaries): We establish the necessary analytical foundations. This includes both restated results from Ref.~\cite{mhiri2025unifying} and original technical statements required for the non-linear extension.
    \item Appendix~\ref{app:proof_patch-variance-double-derivative} (General Non-Linear Case): We present the core proof of the variance-curvature relationship for arbitrary non-linear loss functions.
    \item Appendix~\ref{app:proof_mmd-patch-variance} (Application to MMD and IQP): We specialize the proof to the MMD loss architecture, specifically focusing on the structure of IQP circuits to obtain tighter variance bounds.
\end{itemize}

\subsection{Preliminaries: Key technical statements}\label{app:proof-prelimm-non-linear-loss}

Here we provide key technical necessarily to prove the variance lower bound of general non-linear and MMD losses. Some of these statements are derived in details in Ref.~\cite{mhiri2025unifying} and here we are merely restating them.

\begin{proposition}[Multiple-variable variance decomposition; Proposition~4 in Ref.~\cite{mhiri2025unifying}]\label{prop:var_decomp}
Consider a multivariable function $f(\vec{\th})=f(\th_1,\dots,\th_\nparams)$ depending on $\nparams$ parameters such that $f: \mathbb{R}^\nparams \rightarrow \mathbb{R}$. We assume that each parameter is sampled independently from some distribution $\PC$ i.e., $\thv \sim \PC^{\otimes m}$. Then for any permutation $\pi$, 
the variance of the function $f$ can be expressed as 
\begin{align}
    \Var_{\vec{\th} \sim \mathcal{P}^{\otimes \nparams}}\left[f(\vec{\th})\right] = \sum_{k=1}^\nparams \Ebb_{\pi(\nparams),\dots,\pi(k+1)}[\Var_{\pi(k)}[\Ebb_{\pi(k-1),\dots,\pi(1)}[f(\vec{\th})]]] \;\;.
\end{align}
\end{proposition}
\begin{proposition}[Variance lower bound of a function with bounded derivative; Proposition~3 in Appendix~B of Ref.~\cite{mhiri2025unifying}]
\label{prop:variance-lower-bound-with-bounded-derivative-prop3warmstart}

  Let $f(\thv)$ be a function of $m$ parameters $\thv$. If there exist 
positive constants $a$ and $\gamma$ such that
\begin{equation}   \label{eq:assumption-beta-warm-start-1st-constant} \left|\left.\frac{\partial^{2k} f(\thv)}{\partial\theta_{\alpha}^{2k}}\right|_{\theta_{\alpha}=\theta_{\alpha}^*}\right|\leq a\gamma^{2k}\;,
\end{equation}
 for all $k\geq1$ and if $r$ satisfies $r\leq\frac{3}{2\gamma}$ , then we have
\begin{equation} \label{eq:bound-single-param-variance-general-warm-start-2}
    \Var_{\theta_\alpha}[f(\thv)]\geq \frac{r^4}{45}\left(\left.\frac{\partial^{2} f(\thv)}{\partial\theta_{\alpha}^{2}}\right|_{\theta_{\alpha}=\theta_{\alpha}^*}\right)^2-\frac{a^2\gamma^6r^6}{270}\;.
\end{equation}
\end{proposition}

Notably, we introduce the following proposition, which constitutes an original derivation central to our analysis of the initialization landscape.

\begin{proposition}[Distance between average and fixed point of parametrized functions]
\label{prop:distance-average-fixed-point-Taylor}
Let $f(\thv)=f(\theta_j,\theta_{j+1},..,\theta_m)$ be a parametrized function such that
\begin{equation}
    \left|\frac{\partial^2 f(\thv)}{\partial\theta_j^2}\right|\leq \gamma_j\;,
\end{equation}
for all $l\in\{j,j+1,\dots,m\}$, and let the parameters be independently distributed as $\th_l\sim\th_l^*+{\rm Unif}[-r,r]$. Then, the following upper-bound holds
\begin{equation}
   |\Ebb_{j,j+1,..,m}[f(\thv)]-f(\thv^*)| \leq \sum_{l=j}^m\gamma_l\frac{r^2}{6}\;.
\end{equation}
\end{proposition}
\begin{proof}
    First, let us show that 
    \begin{equation}\label{eq:distance-average-fixed-point-propproof-single}
        |\Ebb_{j}[f(\thv)]-f(\theta_j^*,\theta_{j+1},..,\theta_m)|\leq \frac{r^2}{6}\gamma_j\;,
    \end{equation}
    provided that $\gamma_j$ satisfies 
    \begin{equation} \label{eq:double-deriv-condition-gammaj-prop-distance-average}
        \left|\frac{\partial^2 f(\thv)}{\partial\theta_j^2}\right|\leq \gamma_j\;.
    \end{equation}
    Indeed, using Taylor expansion up to second order of $f(\thv)$ around $\theta_j=\theta_j^*$,
    \begin{align}
        f(\thv)-f(\theta_j^*,\theta_{j+1},..,\theta_m)=(\th_j-\th_j^*)\partial_jf(\theta_j^*,\theta_{j+1},..,\theta_m)+ \frac{(\th_j-\th_j^*)^2}{2}\partial_j^2f(\tilde{\theta}_j,\theta_{j+1},..,\theta_m)\;,
    \end{align}
    where $\tilde{\theta}_j=\th_j^*+c(\th_j-\th_j^*)$ for some constant $c\in (0,1)$.
    Thus, by taking the average with respect to $\th_j$, the first order term vanishes and we have 
    \begin{align}
        |\Ebb_j[f(\thv)]-f(\theta_j^*,\theta_{j+1},..,\theta_m)|&= \left|\Ebb_j\left[\frac{(\th_j-\th_j^*)^2}{2}\partial_j^2f(\tilde{\theta}_j,\theta_{j+1},..,\theta_m)\right]\right|\\
        &\leq\Ebb_j\left[\frac{(\th_j-\th_j^*)^2}{2}\gamma_j\right]\\
        &=\frac{r^2}{6}\gamma_j\;,
    \end{align}
    where we used Eq.~\eqref{eq:double-deriv-condition-gammaj-prop-distance-average} for the inequality. So, given that Eq.~\eqref{eq:double-deriv-condition-gammaj-prop-distance-average} is satisfied, we indeed recover Eq.~\eqref{eq:distance-average-fixed-point-propproof-single}.
    Now, for a function $f(\thv)=f(\theta_j,\theta_{j+1},..,\theta_m)$, we can recursively use triangle inequality as follows 
\begin{align}
   |\Ebb_{j,j+1,..,m}[f(\thv)]-f(\thv^*)|&\leq \underbrace{|\Ebb_{j,j+1,..,m}[f(\thv)]-\Ebb_{j+1,..,m}[f(\theta_j^*,\theta_{j+1},..,\theta_m)]|}_{\leq \gamma_j\, r^2/6 }+|\Ebb_{j+1,..,m}[f(\theta_j^*,\theta_{j+1},..,\theta_m)]-f(\thv^*)|\;,
\end{align}
together with Eq.~\eqref{eq:distance-average-fixed-point-propproof-single} such that
\begin{equation}
   |\Ebb_{j,j+1,..,m}[f(\thv)]-f(\thv^*)| \leq \sum_{l=j}^m\gamma_l\frac{r^2}{6}\;.
\end{equation}
\end{proof}

\subsection{Proof of Theorem~\ref{thm-sup:patch-variance-double-derivative}: Lower bound guarantee of an arbitrary non-linear loss and MMD loss with sufficient curvature}

\subsubsection{General non-linear loss part of the theorem}\label{app:proof_patch-variance-double-derivative}

\begin{proof}
The variance of $\LC(\thv)$ with respect to the set of $m$ parameters, $\thv$, can be decomposed using  Proposition~\ref{prop:var_decomp} as
\begin{equation}
    \Var_{\thv}[\mathcal{L}(\thv)]=\sum_{k=1}^m\Ebb_{\pi(m),...,\pi(k+1)}[\Var_{\pi(k)}[\Ebb_{\pi(k-1),...,\pi(1)}[\mathcal{L}(\thv)]]]\;,
\end{equation}
where $\Ebb_{i_1, \dots, i_l}$ (and similarly $\Var_{i_1, \dots, i_l}$) denotes the expectation (and variance) over the parameter subset $\{\theta_{i_1}, \dots, \theta_{i_l}\}$. Here, $\pi(\cdot)$ represents a permutation operator over the original index set $\{1, \dots, m\}$; importantly, the total variance remains invariant under any such relabelling.

Since each term in this summation is non-negative, we can derive a valid lower bound by isolating a single term $\Ebb_{\bar{\thv}}[\Var_{\theta_{\alpha}}[\mathcal{L}(\thv)]]$ corresponding to an arbitrary parameter $\theta_\alpha$ where $\bar{\thv}$ is the set of parameters excluding $\theta_\alpha$. Due to the permutation invariance of the indices, $\alpha$ can be chosen to represent any parameter in the set, yielding
\begin{equation}
    \Var_{\thv}[\mathcal{L}(\thv)]\geq \Ebb_{\thv}[\Var_{\theta_{\alpha}}[\mathcal{L}(\thv)]] \;\;.
\end{equation}

To further analyze the quantity $\Ebb_{\thv}[\Var_{\theta_{\alpha}}[\mathcal{L}(\thv)]]$, we follow three primary technical steps.

\noindent\underline{Step~1: Bounding the variance.} We first find a lower bound on $\Var_{\theta_\alpha}[\mathcal{L}(\thv)]$ using Proposition~\ref{prop:variance-lower-bound-with-bounded-derivative-prop3warmstart} which states that if we can find some 
positive constants $a$ and $\gamma$ such that
\begin{equation}   \label{eq:assumption-beta-warm-start-1st-constant-2} \left|\left.\left(\frac{\partial^{2k} \mathcal{L}(\thv)}{\partial\theta_{\alpha}^{2k}}\right)\right|_{\theta_{\alpha}=\theta_{\alpha}^*}\right|\leq a\gamma^{2k}\;\;,
\end{equation}
 for all $k\geq1$ and if $r$ satisfies $r\leq\frac{3}{2\gamma}$ , then we have
\begin{equation} \label{eq:bound-single-param-variance-general-warm-start}
    \Var_{\theta_\alpha}[\mathcal{L}(\thv)]\geq \frac{r^4}{45}\left[\left(\left.\frac{\partial^{2} \mathcal{L}(\thv)}{\partial\theta_{\alpha}^{2}}\right)\right|_{\theta_{\alpha}=\theta_{\alpha}^*}\right]^2-\frac{a^2\gamma^6r^6}{270}\;.
\end{equation}

\medskip

\noindent\underline{Step~2: Bounding the expectation of the variance term.}
We perform the average over remaining parameters in Eq.~\eqref{eq:bound-single-param-variance-general-warm-start}, leading to
\begin{align}
\Ebb_{\bar{\thv}}[\Var_{\theta_{\alpha}}[\mathcal{L}(\thv)]] & \geq \Ebb_{\bar{\thv}} \left[ \frac{r^4}{45}\left[\left(\left.\frac{\partial^{2} \mathcal{L}(\thv)}{\partial\theta_{\alpha}^{2}}\right)\right|_{\theta_{\alpha}=\theta_{\alpha}^*}\right]^2-\frac{a^2\gamma^6r^6}{270}\right] \\
& = \frac{r^4}{45} \Ebb_{\bar{\thv}} \left[ \left[\left(\left.\frac{\partial^{2} \mathcal{L}(\thv)}{\partial\theta_{\alpha}^{2}}\right)\right|_{\theta_{\alpha}=\theta_{\alpha}^*}\right]^2 \right] - \frac{a^2\gamma^6r^6}{270} \\
& \geq \frac{r^4}{45} \left[ \Ebb_{\bar{\thv}} \left(\left.\frac{\partial^{2} \mathcal{L}(\thv)}{\partial\theta_{\alpha}^{2}}\right)\right|_{\theta_{\alpha}=\theta_{\alpha}^*}  \right]^2 - \frac{a^2\gamma^6r^6}{270} \;\;, 
\end{align}
where the last inequality is due to the positivity of the variance.

\medskip

\noindent\underline{Step~3: Expressing the bound in terms of local curvature.} To obtain the final bound in terms of the double derivative at the fixed point $\thv^*$ (i.e., the local curvature), we bound the difference between the above average and the curvature at the fixed point using reverse triangle inequality as follows 
\begin{align} \label{eq:bound-average-double-deriv-as-distance-with-fixed-point-warm-start}
    \left[ \Ebb_{\bar{\thv}} \left(\left.\frac{\partial^{2} \mathcal{L}(\thv)}{\partial\theta_{\alpha}^{2}}\right)\right|_{\theta_{\alpha}=\theta_{\alpha}^*}  \right]^2 \geq \left[\left|\left(\left.\frac{\partial^{2} \mathcal{L}(\thv)}{\partial\theta_{\alpha}^{2}}\right)\right|_{\thv=\thv^*}\right|-\left|\Ebb_{\thv}\left(\left.\frac{\partial^{2} \mathcal{L}(\thv)}{\partial\theta_{\alpha}^{2}}\right)\right|_{\theta_{\alpha}=\theta_{\alpha}^*}-\left.\left(\frac{\partial^{2} \mathcal{L}(\thv)}{\partial\theta_{\alpha}^{2}}\right)\right|_{\thv=\thv^*}\right|\right]^2\;\;.
\end{align}
Note that since $\left(\left.\frac{\partial^{2} \mathcal{L}(\thv)}{\partial\theta_{\alpha}^{2}}\right)\right|_{\theta_\alpha=\theta_\alpha^*}$ has the parameter $\theta_\alpha$ evaluated already at $\theta_\alpha^*$, we have $\Ebb_{\bar{\thv}} \left(\left.\frac{\partial^{2} \mathcal{L}(\thv)}{\partial\theta_{\alpha}^{2}}\right)\right|_{\theta_\alpha=\theta_\alpha^*} = \Ebb_{\thv} \left(\left.\frac{\partial^{2} \mathcal{L}(\thv)}{\partial\theta_{\alpha}^{2}}\right)\right|_{\theta_\alpha=\theta_\alpha^*}$.

Then, we use Proposition~\ref{prop:distance-average-fixed-point-Taylor} to bound the distance between the average and the fixed point $\thv^*$. In our case, we identify the function in Proposition~\ref{prop:distance-average-fixed-point-Taylor} as the double derivative of the loss with the parameter indexed $\alpha$ evaluated as $\theta_\alpha = \theta^*_\alpha$. That is, we have $f(\thv) = \left.\frac{\partial^2\LC(\thv)}{\partial\theta_\alpha^2}\right|_{\theta_\alpha = \theta_\alpha^*}$. Hence, by invoking the proposition, if there exists some constant $\gamma_j$ that satisfies
\begin{equation} \label{eq:bounded-4st-derivative-gamma-j-assumption-warm-start}
    \left|\left.\left(\frac{\partial^{4} \mathcal{L}(\thv)}{\partial\theta_{j}^{2}\partial\theta_{\alpha}^{2}}\right)\right|_{\theta_\alpha = \theta_\alpha^*}\right|\leq \gamma_j\;,
\end{equation}
then we have
\begin{equation}\label{eq:proof-var-curvature-00}
    \left|\Ebb_{\thv}\left[\left.\left(\frac{\partial^{2} \mathcal{L}(\thv)}{\partial\theta_{\alpha}^{2}}\right)\right|_{\theta_{\alpha}=\theta_{\alpha}^*}\right]-\left(\left.\frac{\partial^{2} \mathcal{L}(\thv)}{\partial\theta_{\alpha}^{2}}\right)\right|_{\thv=\thv^*}\right|\leq \sum_{j\neq\al} \gamma_j\frac{r^2}{6}\;.
\end{equation}

Therefore, combining this Eq.~\eqref{eq:proof-var-curvature-00} together with Eq.~\eqref{eq:bound-average-double-deriv-as-distance-with-fixed-point-warm-start} into Eq.~\eqref{eq:bound-single-param-variance-general-warm-start} leads to

\begin{equation}
\Var_{\thv}[\mathcal{L}(\thv)]\geq \frac{r^4}{45}\left(\left|\left(\left.\frac{\partial^{2} \mathcal{L}(\thv)}{\partial\theta_{\alpha}^{2}}\right)\right|_{\thv=\thv^*}\right|-\sum_{j\neq\alpha}\gamma_j \frac{r^2}{6}\right)^2-\frac{a^2\gamma^6r^6}{270}\;\;,
\end{equation}
provided that the following condition is satisfied
\begin{equation}
    \left|\left.\left(\frac{\partial^{2} \mathcal{L}(\thv)}{\partial\theta_{\alpha}^{2}}\right)\right|_{\thv=\thv^*}\right|\geq \sum_{j\neq\alpha}\gamma_j \frac{r^2}{6}\;\;.
\end{equation}

Before going into the final step, let us introduce some notation to simplify the expression. Denote $c_{\alpha}(\thv^*)=\left|\left(\left.\frac{\partial^{2} \mathcal{L}(\thv)}{\partial\theta_{\alpha}^{2}}\right)\right|_{\thv=\thv^*}\right|$, \,$\beta_1=\sum_{j\neq\alpha} \frac{\gamma_j}{6}$ and $\beta_2=\frac{a^2\gamma^6}{6}$. Now the above expressions can be written as: Provided that $c_{\alpha}(\thv^*)\geq r^2\beta_1$, the variance lower-bound follows
\begin{align}
    \Var_{\theta_\alpha}[\mathcal{L}(\thv)]&\geq\frac{r^4}{45}\left(\left(c_{\alpha}(\thv^*)-r^2\beta_1\right)^2-r^2\beta_2\right) \\
    &\geq\frac{r^4}{45}\left(c_{\alpha}(\thv^*)^2-(2\beta_1c_{\alpha}(\thv^*)+\beta_2)r^2\right)\;\;,\label{eq:proof-var-curvature-01}
\end{align}
where the last inequality comes from expanding the square and throwing away the last positive term in the expansion.

Lastly, to further ensure that the lower bound is informative, we require the difference in Eq.~\eqref{eq:proof-var-curvature-01} is non-negative, leading to the condition on $r$ as follows
\begin{align}\label{eq:proof-var-curvature-02} 
     r^2\leq \Delta\frac{ c_{\alpha}(\thv^*)^2}{2\beta_1c_{\alpha}(\thv^*)+\beta_2}\;\;,
\end{align}
for some arbitrary constant $0<\Delta<1$.

Crucially, if $r$ is chosen such that it respects the condition in Eq.~\eqref{eq:proof-var-curvature-02}, it also automatically satisfies satisfies $\beta_1r^2\leq c_{\alpha}(\thv^*)$ since 
\begin{equation}
  r^2\leq \Delta\frac{ c_{\alpha}(\thv^*)^2}{2\beta_1c_{\alpha}(\thv^*)+\beta_2}\leq \Delta\frac{ c_{\alpha}(\thv^*)}{2\beta_1} \leq \frac{ c_{\alpha}(\thv^*)}{\beta_1} \;\;.
\end{equation}

Altogether, given that $r$ satisfies Eq.~\eqref{eq:proof-var-curvature-02} and also $r\leq \frac{3}{2\gamma}$ (i.e., the condition from \textit{Step 1}), then we have the following lower-bound on the variance
\begin{equation}
    \Var_{\thv}[\mathcal{L}(\thv)]\geq \Ebb_{\thv}[\Var_{\theta_{\alpha}}[\mathcal{L}(\thv)]]\geq (1-\Delta)\frac{r^4}{45}c_{\alpha}(\thv^*)^2\;\;.
\end{equation} 
Hence, this completes the proof for the general non-linear loss of the theorem.
\end{proof}

\subsubsection{MMD loss part of the theorem}\label{app:proof_mmd-patch-variance}
\begin{proof} We specify the non-linear loss $\LC(\thv)$ in Theorem~\ref{thm-sup:patch-variance-double-derivative}  to be the MMD loss in Eq.~\eqref{eq:MMD-loss} 
\begin{align}
    \mathcal{L}(\thv)=\sum_{A\subseteq[n]} w_A \big(z_A(\thv)-t_A\big)^2 \;\;,
\end{align}
where, for the Gaussian kernel of bandwidth $\sigma$, the weights are given by $w_A=(1-p_\sigma)^{n-|A|}p_\sigma^{|A|}$ with $p_\sigma=\frac{1-e^{-1/(2\sigma^2)}}{2}$. We also further consider the IQP circuit architecture as defined in Eq.~\eqref{eq:IQP-circuit-general}.

\medskip

Crucially, we show that for the (local) MMD loss and the IQP circuit the required conditions in Eq.~\eqref{eq:assumption-even-deriv-bound-thm} and Eq.~\eqref{eq:assumption-fourth-deriv-bound-thm} are satisfied. In particular, we provide the constants $a$, $\gamma$, $\gamma_j$ and $\Delta$ for this specific setting.

\bigskip

\noindent\underline{Let us first find constants $a$ and $\gamma$ that satisfies Eq.~\eqref{eq:assumption-even-deriv-bound-thm}.} Consider the correlator $z_A(\thv)$ with a general IQP circuit in Eq.~\eqref{eq:correlator-expansion-fourier-lemma} from Lemma~\ref{lemma:correlators-fourier-expansion} 
\begin{align}
    z_A(\thv)
=\langle \vec{0}|U(\thv)Z_AU^\dagger(\thv)|\vec{0}\rangle=
\frac{1}{2^n}
\sum_{\zv\in\{0,1\}^n}
\exp\!\left(
2i\sum_{\bv\in\mathcal{A}_A}
\theta_{\bv} (-1)^{\bv\cdot\zv}
\right)\;\;,
\end{align}
where $\mathcal{A}_A$ denotes the set of indicator strings corresponding to rotations that anti-commute with $Z_A$, i.e. $\{X_{\bv},Z_A\}=0$.

One can observe that if the generator $X_{\alpha}$ associated with the parameter $\theta_\alpha$ (-- can be either one- or two-qubit gates) anti-commutes with a certain Pauli string  $Z_A$ then, by a direct computation we have 
\begin{equation} \label{eq:double-derivative-za-gives-4za}
    \frac{\partial^2 z_A(\thv)}{\partial\theta_{\alpha}^2}=-4z_A(\thv)\;,
\end{equation}
which could also be understood from $-[X_{\alpha},[X_{\alpha},Z_A]]=-4Z_A$. Therefore, we can bound the even order derivatives as follows
\begin{equation}   \left|\frac{\partial^{2k} z_A(\thv)}{\partial{\theta_{\alpha}^{2k}}}\right|=|(-4)^kz_A(\thv)|\leq 2^{2k}\;,
\end{equation}
where we used $|z_A(\thv)|\leq 1$ in the inequality.

Now, for the terms $z_A^2(\thv)$ in MMD, we use the product rule for derivatives as follows
\begin{align}
  \frac{\partial^{2k} z_A^2(\thv)}{\partial{\theta_{\alpha}^{2k}}}  
  &=\sum_{l=0}^{2k}\binom{2k}{l}\frac{\partial^{2k-l} z_A(\thv)}{\partial{\theta_{\alpha}^{2k-l}}}\frac{\partial^{l} z_A(\thv)}{\partial{\theta_{\alpha}^{l}}}\\
  &\leq \sum_{l=0}^{2k}\binom{2k}{l}2^{2k-l}2^{2k}\\
  &=4^{2k}\;,
\end{align}
 where the inequality uses the previous bound on single correlator derivatives and the last equality is obtained by recognising the binomial sum.

The higher-order derivatives of the MMD loss (with $k\geq 1$) are then bounded as follows 
\begin{align}
    \frac{\partial^{2k} \mathcal{L}(\thv)}{\partial\theta_{\alpha}^{2k}}&=\sum_{\substack{A\subseteq{[n]}\\
    \{X_{\alpha},Z_A\}= 0}}w_A\left(\frac{\partial^{2k} z_A^2(\thv)}{\partial{\theta_{\alpha}^{2k}}}-2t_A\frac{\partial^{2k} z_A(\thv)}{\partial{\theta_{\alpha}^{2k}}}\right)\\
    &\leq\sum_{\substack{A\subseteq{[n]}\\
    \{X_{\alpha},Z_A\}= 0}}w_A\left(4^{2k}+2|t_A|2^{2k}\right)\\
    &\leq \left(4^{2k}+2^{2k+1}\right)\sum_{\substack{A\subseteq{[n]}\\
    \{X_{\alpha},Z_A\}= 0}}w_A\\
    &\leq 2\cdot 4^{2k}\sum_{\substack{A\subseteq{[n]}\\
    \{X_{\alpha},Z_A\}= 0}}w_A\;,
\end{align}
where in the last inequality we used $2^{2k+1}\leq 4^{2k}$ to group terms together and identify the constant $\gamma=4$ satisfying Eq.~\eqref{eq:assumption-beta-warm-start-1st-constant}. 

Now, let us bound the remaining sum. If $\alpha$ labels a single qubit gate we have to sum all the weight $w_A$ such that $\alpha\in A$ that is $p_{\sigma}$ (which can be interpreted as the probability/weight of having $Z$ on a qubit $\alpha$). For two-qubits gates, assuming $\alpha=(i,j)$, we need either $i\in A$ and $j\in A^c$, or inversely $i\in A^c$ and $j\in A$ which summed up to $p_{\sigma}(1-p_{\sigma})$ each so the total sum of weight anti-commuting with 2-qubit gates is $2p_{\sigma}(1-p_{\sigma})$. In this case, the sum is slightly larger than for single qubits, but it is important to notice that it scales with $p_{\sigma}$ which can be small for low-body MMD.  Therefore , we have
\begin{align}\label{eq:even-derivative-bound-single-param-warm-start}
    \frac{\partial^{2k} \mathcal{L}(\thv)}{\partial\theta_{\alpha}^{2k}}&\leq \underbrace{4p_{\sigma}(1-p_{\sigma}) }_{=a}\,\underbrace{4^{2k}}_{=\gamma^{2k}}\;,
\end{align}
where we identify constants $a=4p_{\sigma}(1-p_{\sigma})$ and $\gamma=4$ which satisfies Eq.~\eqref{eq:bound-single-param-variance-general-warm-start} provided that $r\leq 3/8$ (the condition that we will show we satisfy later in the proof).
Therefore, we have 
\begin{equation}
    \beta_2=\frac{a^2\gamma^6}{6}\,,\quad \text{with}\quad a=4p_{\sigma}(1-p_{\sigma})\quad \text{and}\quad \gamma=4\;.
\end{equation}

\bigskip

\noindent \underline{Let us next find $\gamma_j$ which satisfy the condition in Eq.~\eqref{eq:assumption-fourth-deriv-bound-thm} of the theorem.} $\gamma_j$ are obtained by upper-bounding $\frac{\partial^{4} \mathcal{L}(\thv)}{\partial\theta_{\alpha}^{2}\partial\theta_{j}^{2}}$ for all $j\neq \alpha$. This is similar as the case of derivatives with respect to single parameter since Eq.~\eqref{eq:double-derivative-za-gives-4za} can be applied twice for correlators that are affected by both gates. Indeed, given that some $Z_A$ anti-commutes with both generators ($j$ and $\alpha$), we have $[X_j,[X_j,[X_{\alpha},[X_{\alpha},Z_A]]]]=16Z_A$. By a similar argument as previously, mixed fourth-order derivatives of a the correlator follow
\begin{align}
    \left|\frac{\partial^{4} z_A(\thv)}{\partial\theta_{\alpha}^{2}\partial\theta_{j}^{2}} \right| \leq 2^4 \;\;,
\end{align}
and, mixed fourth-order derivatives of the MMD loss are

\begin{align}
  \left|\frac{\partial^{4} \mathcal{L}(\thv)}{\partial\theta_{\alpha}^{2}\partial\theta_{j}^{2}} \right|&\leq \left(4^{4}+2^{5}\right)\sum_{\substack{A\subseteq{[n]}\\
    \{X_{\alpha},Z_A\}= 0\\ \{X_{j},Z_A\}= 0}}w_A\; \;\;,
\end{align}
for any $\theta_j \in \thv$.

Next, we compute the sum of the remaining weight. If both $j$ and $\alpha$ are single qubit labels, then the sum of the remaining weight is $p_{\sigma}^2$ (i.e. probability weight of having $Z$ on two given qubits). Now, if $\alpha$ is a two-qubits label $(\alpha_1,\alpha_2)$ and $j$ is a single qubit, we need either $j,\a_1\in A$ and $\a_2\in A^c$ which gives $p_\sigma^2(1-p_\sigma)$ (probability of having $Z$ on two given qubits and $I$ on a third given one) or $\a_1\in A^c$ and $j,\a_2\in A$ which gives the same contribution so $2p_\sigma^2(1-p_\sigma)$ in total. Now, if both are two qubit gates, we have two pairs of two indices that should corresponds to different Paulis ($I$ and $Z$) so there are four such choices with weight $p_\sigma^2(1-p_\sigma)^2$ each which is the largest possible case so that 
\begin{align}
  \left|\frac{\partial^{4} \mathcal{L}(\thv)}{\partial\theta_{\alpha}^{2}\partial\theta_{j}^{2}} \right|&\leq \left(4^{4}+2^{5}\right)4p_{\sigma}^2(1-p_\sigma)^2=1152p_{\sigma}^2(1-p_\sigma)^2=\gamma_j\;. 
\end{align}
Therefore, with $m$ parameters in total, we have
\begin{equation}
    \beta_1=\sum_{j\neq \al}\frac{\gamma_j}{6}=192(m-1)p_{\sigma}^2(1-p_\sigma)^2\;.
\end{equation}

\bigskip

\noindent\underline{Lastly, let us identify the constant $\Delta$} that satisfies the condition $r\leq \frac{3}{2\gamma}=\frac{3}{8}$ which is listed in \textit{Step 1} of the proof of the theorem. From Eq.~\eqref{eq:even-derivative-bound-single-param-warm-start}, we have $c_{\alpha}(\thv^*)\leq 4^3p_{\sigma}(1-p_{\sigma})$. Therefore, using this together with the definition of $\beta_2$ leads to the following upper bound for Eq.~\eqref{eq:proof-var-curvature-02}: 
\begin{equation}
    r^2\leq \Delta\frac{ c_{\alpha}(\thv^*)^2}{2\beta_1c_{\alpha}(\thv^*)+\beta_2}\leq\Delta\frac{ c_{\alpha}(\thv^*)^2}{\beta_2}\leq \Delta \frac{3}{8}\;,
\end{equation}
which always satisfies the condition $r\leq \frac{3}{8}$ provided that $\Delta \leq \frac{3}{8}$. This completes the final part of the proof of the theorem.
\end{proof}

\section{Curvature mechanisms at initialization: identity, unbiased, and data-dependent centers}
\label{app:curvature-mechanisms}

This appendix compares the local curvature of the MMD loss at several initialization centers relevant to the discussion in the main text: the identity center, an unbiased data-independent reference center, and a data-dependent marginal-matching center.
Our goal is to isolate the mechanisms that produce non-vanishing local curvature (mismatch versus model sensitivity), and to identify regimes in which the curvature remains inverse-polynomial in $n$.
Whenever such an inverse-polynomial curvature lower bound is established (i.e., there is in our terminology \textit{polynomially large} curvature), the corresponding inverse-polynomial patch-variance guarantee follows directly from Theorem~\ref{thm-sup:patch-variance-double-derivative}.

\subsection{General curvature decomposition}
\label{app:general-curvature-decomposition}

The double derivative of the MMD with respect to a single-qubit gate parameter $\theta_{\alpha}$ can be written as
\begin{align} \label{eq:MMD-double-derivative-single-parameter}
    \frac{\partial^2 \mathcal{L}(\thv)}{\partial\theta_{\al}^2}&= \sum_{\substack{A\subseteq [n]\\ \alpha \in A}} 2w_A  (4\underbrace{{z_A(\thv)(t_A-z_A(\thv))}}_{\rm Mismatch}+\underbrace{g_A^{(\al)}(\thv)^2)}_{\rm Model \; sensitivity}  \;,
\end{align}
where we used the fact that $\frac{\partial^2 z_A(\thv)}{\partial{\theta_{\alpha}}^2}=-4z_A(\thv)$ since $z_A(\thv)$ is a trigonometric series with product of sines and cosines with parameters $2\theta_k$ (see Eq.~\eqref{eq:z_a-as-sine-cosine-series}). A similar expression holds for a two-qubit gate parameter, although the anti-commuting set differs (i.e., the terms retained in the sum are determined by the corresponding anti-commutation relations). In Eq.~\eqref{eq:MMD-double-derivative-single-parameter}, we can clearly distinguish two contributions: one arising from the \textit{mismatch} between the target and the model, weighted by the value of the model correlator, and another corresponding to what we term \textit{model sensitivity}, given by the squared magnitude of the model correlators. The latter is, in principle, independent of the target distribution, except in the data-dependent initialization setting. This decomposition will be used below to compare the mechanisms underlying local trainability for different initialization centers.

\subsection{Data-agnostic initialization strategies}\label{app:data-agnostic-initialization-curvature}

Here we provide a formal version of Theorem~\ref{thm:variance-guarantee-agnostic-initialisation-main}, i.e., a non-vanishing variance guarantee for the data-agnostic initialization strategies, together with its proof. In particular, we focus on the identity initialization strategy which is mismatch-driven, and the unbiased data-independent initialization which is sensitivity-driven.

\begin{theoremappendix}[Variance guarantee for agnostic initialization strategies, formal]
\label{thm:variance-guarantee-agnostic-initialisation-appendix}
Consider the low-body MMD loss with $\sigma\in\Theta(\sqrt{n})$ (so that $w_A\in\Theta(n^{-|A|})$ for $A\subseteq[n]$), the IQP circuit and the following initialization scenarios, either
\begin{enumerate}
    \item The initialization center is at identity $\thv^*=\vec{0}$, and further assume that there exists at least one single bit marginal that is system size independent, that is, there exists at least one qubit indexed by $\al$ such that $1-t_{\al}\in\Theta(1)$, or:
    \item The unbiased initialization center such that single qubit gates parameters are set to $\pi/4$ and two-qubit gates parameters are set to $0$.
\end{enumerate}     
Then there exists at least one single-qubit parameter index $\alpha$ such that
\begin{equation}
\left|\frac{\partial^2 \mathcal{L}(\thv^*)}{\partial\theta_\alpha^2}\right|
\in \Theta\!\left(\frac{1}{n}\right) \;\;.
\end{equation}
Consequently, by Theorem~\ref{thm-sup:patch-variance-double-derivative}, there exists a patch half-width
\begin{equation}
r\in\OC\!\left(\frac{1}{\poly(n)}\right) \;\;,
\end{equation}
such that, for $\thv\sim\thv^*+\mathrm{Unif}([-r,r]^m)$ 
\begin{equation}
\Var_{\thv}[\mathcal{L}(\thv)]\in\Omega\!\left(\frac{1}{\poly(n)}\right) \;\;.
\end{equation}
\end{theoremappendix}

\medskip

Before going into the proof of the theorem (see Appendix~\ref{app:identity-initialization-curvature} for the identity initialization and Appendix~\ref{app:unbiased-initialization-curvature} for the unbiased initialization), we provide some technical notes and interpretation regarding the data-agnostic initialization strategies.

\medskip

\underline{Identity initialization corresponds to a local maximum.} We show later on that the double derivatives of the identity initialization are negative
in Eq.~\eqref{eq:iden-local-max}. Notice that for 2-qubit gates only the number of remaining terms $A$ in the sum is different, therefore each double derivative is also negative. Now, we also have that the gradient and the partial derivatives with respect to two different parameters cancels. 
Indeed, the gradient components are given by
\begin{equation}\label{eq:MMD-gradient-component-appendix-identity}
    \frac{\partial \mathcal{L}(\thv^*)}{\partial\theta_{\alpha}}=\sum_{A\subseteq[n]} 2w_A\, g_A^{(\al)}\!(\thv^*)(z_A(\thv^*)-t_A)\;,
\end{equation}
which vanishes for $\thv^*=\vec{0}$ due to $g_A^{(\al)}\!(\vec{0})=0$. Now, if we derive with respect to another parameter with label $\beta$, we have
\begin{equation}
    \frac{\partial^2 \mathcal{L}(\thv^*)}{\partial\theta_{\alpha}\partial\theta_{\beta}}=\sum_{A\subseteq[n]} 2w_A\, \left(\frac{\partial^2z_A(\thv^*)}{\partial\theta_{\alpha}\partial\theta_{\beta}}(z_A(\thv^*)-t_A)+g_A^{(\al)}\!(\thv^*)g_A^{(\beta)}\!(\thv^*)\right)\;,
\end{equation}
which also vanishes for $\thv^*=\vec{0}$ because if $\al,\beta\in\mathcal{A}_A$, then  $\frac{\partial^2z_A(\thv^*)}{\partial\theta_{\alpha}\partial\theta_{\beta}}=-4\sin(2\theta_\al)\sin(2\theta_\beta)=0$ (from Eq.~\eqref{eq:z_a-as-sine-cosine-series}).
Therefore, the Hessian at $\thv=\vec{0}$ is diagonal with negative entries and the gradient vanishes so we are at a local maximum which makes sense given that we are at the boundary of the space of correlations. Notice that the true maximum is attained by fixing $\theta_j=0$ if $t_j<0$, $\theta_j=\pi/2$ (i.e. such that $z_j(\thv)=-1$) if $t_j>0$, and both choices are valid for $t_j=0$ (leading to a degenerate global maximum).

\medskip

\underline{Geometric interpretation.}
In the correlator representation $\{z_A\}_{A\neq\emptyset}$, the identity center corresponds to an extremal aligned configuration ($z_A=1$ for all $A$), whereas the unbiased $\pi/4$ center sits at the origin ($z_A=0$ for all non-trivial $A$). As opposed to identity initialization, the unbiased initialization does not correspond to a local maximum, but a region which usually have non-negligible gradient (except if all $t_j=0$ or are exponentially vanishing). Indeed,  Eq.~\eqref{eq:MMD-gradient-component-appendix-identity} with this unbiased initialization leads to:
\begin{equation}
    \frac{\partial \mathcal{L}(\thv^*)}{\partial\theta_{\alpha}}=-2\sum_{A\subseteq[n]} w_A\, g_A^{(\al)}\!(\thv^*)t_A=4w_{\{\al\}}t_{\{\al\}}\;.
\end{equation}
These two data-independent references can therefore exhibit large local curvature for qualitatively different geometric reasons: the former through large mismatch, the latter through maximal local sensitivity of selected low-weight features.

\subsubsection{Proof of Theorem~\ref{thm:variance-guarantee-agnostic-initialisation-appendix}: Identity initialization part}
\label{app:identity-initialization-curvature}

\begin{proof}
We first consider the identity initialization center, $\thv^*=\vec{0}$.
At this point, the prepared state is $\ket{0}^{\otimes n}$, so all Pauli-$Z$ correlators are extremal:
\begin{equation}
z_A(\vec{0})=1
\qquad\text{for all } A\subseteq[n].
\end{equation}
Moreover, the correlator derivatives vanish,
\begin{equation}
g_A^{(\alpha)}(\vec{0})=0,
\end{equation}
since $\thv^*=\vec{0}$ is an extremal point of each correlator $z_A(\thv)$ along any parameter direction.

Substituting into Eq.~\eqref{eq:MMD-double-derivative-single-parameter}, the model-sensitivity term disappears and the curvature reduces to a purely mismatch-driven contribution:
\begin{align}
\label{eq:MMD-double-derivative-identity-initialisation}
\frac{\partial^2 \mathcal{L}(\vec{0})}{\partial\theta_{\alpha}^2}
= \sum_{\substack{A\subseteq [n]\\ \alpha \in A}} 8w_A\,(t_A-1).
\end{align}
Since $t_A\le 1$ for all $A$, each term in the sum is non-positive, and therefore
\begin{equation}\label{eq:iden-local-max}
\left|\frac{\partial^2 \mathcal{L}(\vec{0})}{\partial\theta_{\alpha}^2}\right|
=
-\frac{\partial^2 \mathcal{L}(\vec{0})}{\partial\theta_{\alpha}^2}.
\end{equation}

Equation~\eqref{eq:MMD-double-derivative-identity-initialisation} shows that identity initialization can exhibit large local curvature whenever the target differs appreciably from the fully aligned correlator configuration.
In particular, we do not expect the curvature to be exponentially small unless the target distribution is highly concentrated around the all-zero bit string. Such strongly peaked targets do not exhibit the structural complexity associated with meaningful generative tasks. Thus, identity initialization provides a natural \emph{mismatch-driven baseline} for local trainability.

To make this explicit, consider the MMD loss truncated to correlators of weight at most $2$, denoted by $\tilde{\mathcal{L}}_2(\thv)$.
Then
\begin{align}
\label{eq:second-order-MMD-double-derivative-identity-initialisation}
\left|\frac{\partial^2 \tilde{\mathcal{L}}_2(\vec{0})}{\partial\theta_{\alpha}^2}\right|
&=
8p_{\sigma}(1-p_\sigma)^{n-1}(1-t_{\alpha})
+8p_{\sigma}^2(1-p_\sigma)^{n-2}\sum_{j\neq \alpha}(1-t_{\alpha j})\;,
\end{align}
which is also a lower-bound on the full MMD curvature since each term has the same sign.  
Hence, if at least one single qubit correlator $\al$ such that $1-t_{\al}\in\Theta(1)$, then the local curvature scales as $\theta(1/n)$ and Theorem~\ref{thm-sup:patch-variance-double-derivative} yields a non-vanishing patch-variance guarantee around $\vec{0}$ which completes the proof.
\end{proof}

\subsubsection{Proof of Theorem~\ref{thm:variance-guarantee-agnostic-initialisation-appendix}: Unbiased data-independent initialization part}
\label{app:unbiased-initialization-curvature}

\begin{proof}
To clarify the curvature decomposition further, it is useful to consider a second data-independent initialization center, distinct from the identity.
As a reference point, we take
\begin{align}
\theta_j^* &= \pi/4 \qquad \text{for all single-qubit parameters } j, \\
\theta_{jk}^* &= 0 \qquad \text{for all two-qubit parameters } (j,k).
\end{align}
The two-qubit angles are set to zero, so from Eq.~\eqref{eq:z_a-as-sine-cosine-series} the correlators factorize over the single-qubit rotations:
\begin{equation}
z_A(\thv^*)=\prod_{j\in A}\cos(2\theta_j^*).
\end{equation}
Therefore, for every nontrivial subset $A\neq\emptyset$,
\begin{equation}
z_A(\thv^*)=0,
\end{equation}
because $\cos(2\cdot \pi/4)=\cos(\pi/2)=0$.
Equivalently, the corresponding model distribution is uniform over computational basis states, so this initialization is unbiased with respect to all measured Pauli-$Z$ correlators.

For a single-qubit parameter $\theta_\alpha$, the only nonzero first derivative at $\thv^*$ is the single-site one:
\begin{align}
z_{\{\alpha\}}(\theta_\alpha) &= \cos(2\theta_\alpha),\\
\partial_{\theta_\alpha} z_{\{\alpha\}}(\theta_\alpha) &= -2\sin(2\theta_\alpha),
\end{align}
hence
\begin{equation}
\left.\partial_{\theta_\alpha} z_{\{\alpha\}}(\thv)\right|_{\thv=\thv^*}=-2,
\qquad
\left.\partial_{\theta_\alpha} z_A(\thv)\right|_{\thv=\thv^*}=0
\quad\text{for }A\neq\{\alpha\}.
\end{equation}
Thus, at $\thv^*$ the mismatch term in Eq.~\eqref{eq:MMD-double-derivative-single-parameter} vanishes (because it is proportional to $z_A(\thv^*)$), while the model-sensitivity term remains nonzero through the single-site feature.
The curvature therefore reduces to
\begin{align} \label{eq:mmd-double-derivative-unbiased-app}
\frac{\partial^2 \mathcal{L}(\thv^*)}{\partial\theta_\alpha^2}
&=
\sum_{\substack{A\subseteq[n]\\ \alpha\in A}} 2w_A \big(g_A^{(\alpha)}(\thv^*)\big)^2
=2w_{\{\alpha\}}\cdot 4
=8\,w_{\{\alpha\}}.
\end{align}
Therefore we can use that $e^{-\frac{1}{x}}> 1-\frac{1}{x}$, to see that for $\sigma\in\Theta(\sqrt{n})$ we obtain $p_\sigma \in\Theta(1/n)$ which leads to $w_{\{\al\}}\in\Theta(1/n)$. Finally, applying Theorem~\ref{thm-sup:patch-variance-double-derivative} completes the proof.
\end{proof}

\subsection{Data-dependent initialization strategy}
\label{app:data-dependent-initialization-curvature}

We now consider a data-dependent initialization chosen to match the first-order target correlators at initialization, i.e.,
\begin{equation}
z_j(\thv^*)=t_j
\qquad\text{for all }j\in[n].
\end{equation}
Our goal is to show that, despite this improved target alignment (and hence reduced mismatch at low order), the MMD curvature can remain non-vanishing at initialization. Intuitively, there are two sources of curvature in this case: (i) the mismatch of 2-nd order and above moments and (ii) the model curvature (which theoretically could be enhanced by the data-dependent initialization).

As a simple and explicit example, we consider the standard marginal-matching initialization in which the single-qubit angles are chosen such that the model matches the one-site marginals, while all two-qubit angles are set to zero.
Concretely, one may choose (equivalently) $\theta_j^*$ so that
\begin{equation}
\cos(2\theta_j^*)=t_j,
\end{equation}
which matches the first-order Pauli-$Z$ correlators.
Since the two-qubit parameters are initialized at $0$, from Eq.~\eqref{eq:z_a-as-sine-cosine-series} the correlators factorize as in the unbiased reference case:
\begin{equation}
z_A(\thv^*)=\prod_{j\in A}\cos(2\theta_j^*)
=\prod_{j\in A} t_j.
\end{equation}
Similarly, for a single-qubit parameter $\theta_\alpha$ and any $A$ with $\alpha\in A$, we have
\begin{equation}
g_A^{(\alpha)}(\thv^*)
=\left.\frac{\partial z_A(\thv)}{\partial\theta_\alpha}\right|_{\thv=\thv^*}
=-2\sin(2\theta_\alpha^*)\prod_{j\in A\setminus\{\alpha\}}\cos(2\theta_j^*)
=-2\sqrt{1-t_\alpha^2}\prod_{j\in A\setminus\{\alpha\}} t_j,
\end{equation}
where we used $\sin^2(2\theta_\alpha^*)=1-t_\alpha^2$.

Substituting these expressions into Eq.~\eqref{eq:MMD-double-derivative-single-parameter}, the curvature becomes
\begin{equation}\label{eq:double_deribative_w_mismatch_and_modelsense}
    \frac{\partial^2 \mathcal{L}(\thv^*)}{\partial\theta_{\al}^2}=\sum_{{\substack{ A\subseteq [n]\\
    {\al}\in A}}} 8 w_{A}  \left(\underbrace{\left(t_{A}-\prod_{j\in {A}}t_j\right)\prod_{j\in {A}}t_j}_{\rm Mismatch}+\underbrace{(1-t_{\al}^2)\prod_{j\in A\backslash \{\alpha\}}t_j^2}_{\rm Model\; Sensitivity}\right)\;.
\end{equation}
This expression makes the mechanism decomposition explicit in the data-dependent setting.
The first term is a residual mismatch contribution: it vanishes when the target correlator $t_A$ factorizes exactly into one-site marginals, and is therefore expected to be small in weakly correlated target regimes.
The second term is a model-sensitivity contribution, which remains nonzero provided the matched one-site marginals are not saturated (i.e., $1-t_\alpha^2$ is bounded away from zero).

Thus, in contrast to the identity initialization, where the curvature is purely mismatch-driven, the data-dependent initialization retains a nonzero sensitivity contribution even after matching low-order target statistics.
This is the key mechanism underlying the local trainability guarantee in the weakly-correlated target regime analyzed below. \\

\paragraph*{Illustrative second-order truncation.}
To make this comparison more explicit, consider again the MMD truncated to correlators of weight at most $2$, denoted by $\tilde{\mathcal{L}}_2(\thv)$.
Then the curvature at the data-dependent center reads
\begin{align} \label{eq:second-order-MMD-double-derivative-data-dependent-initialisation}
   \frac{\partial^2 \tilde{\mathcal{L}}_2(\thv^*)}{\partial\theta_{\al}^2}&=8p_{\sigma}(1-p_\sigma)^{n-1}\underbrace{(1-t_{\al}^2)}_{\rm Model\; Sensitivity} + 8p_{\sigma}^2(1-p_\sigma)^{n-2}\sum_{j\neq {\al}}(\underbrace{(t_{\alpha j}-t_{\al}t_j)t_{\al}t_j}_{\rm Mismatch}+\underbrace{(1-t_{\al}^2)t_j^2)}_{\rm Model\; Sensitivity}\;.
\end{align}

Compared with Eq.~\eqref{eq:second-order-MMD-double-derivative-identity-initialisation}, this shows explicitly that the leading contribution at the data-dependent center can be sensitivity-driven rather than mismatch-driven.

\subsubsection{Target analysis under an approximately factorizable correlated target assumption}
\label{app:target-analysis}
We now analyse the curvature at the data-dependent initialization under target assumptions of approximately factorizable correlations and non-saturated marginals (Assumptions~\ref{ass:weakly-correlated-target} and~\ref{ass:no-scaling-with-n}), and then invoke Theorem~\ref{thm-sup:patch-variance-double-derivative} to obtain a non-vanishing patch-variance guarantee. Concretely, Assumption~\ref{ass:weakly-correlated-target}  states that the target is close to a product distribution. More precisely, we assume

\begin{equation}\label{eq:assumption-correlations}
\left|t_A-\prod_{j\in A}t_j\right|\le \left(\frac{C}{n}\right)^{|A|/2},
\end{equation}
where the constant $C$ is smaller than $n$. Fig.~\ref{fig:correlation-decay-genomic-dataset-n20} shows the empirical upper-bound for the genomic dataset with $n=20$, which appears to be consistent with our assumption.

\begin{figure*}
    \centering
    \includegraphics[width=0.58\linewidth]{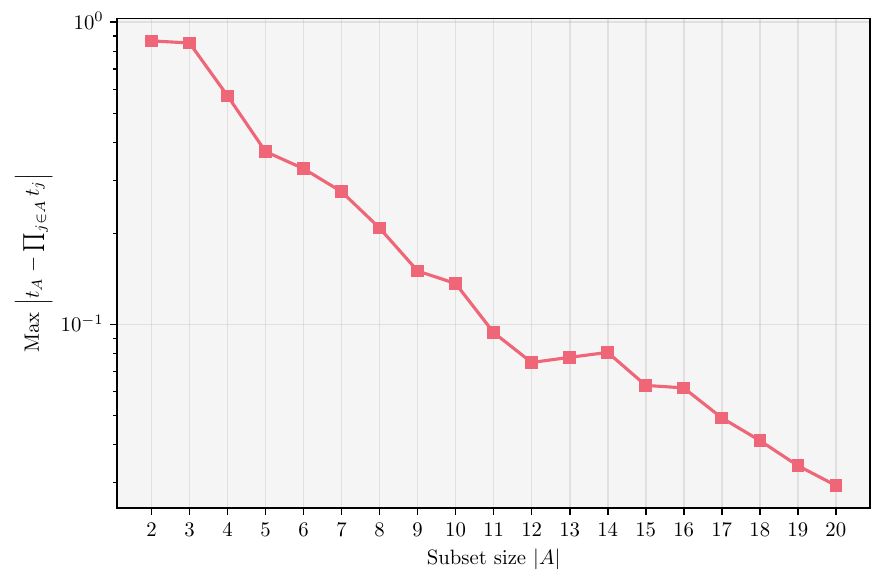}
    \caption{Empirical evaluation of Assumption~\ref{ass:weakly-correlated-target} for the genomic dataset  with $n=20$ qubits.}
    \label{fig:correlation-decay-genomic-dataset-n20}
\end{figure*}
Assumption~\ref{ass:no-scaling-with-n} ensures that at least one local marginal retains non-vanishing fluctuation, i.e. is non-saturated, namely

\begin{equation}\label{eq:assumption-nonsaturated-marginals-app}
\exists\al \in[n] \;: \quad1-t_{\al}^2\in\Theta(1).
\end{equation}

When $t_j=0$ for all $j\in[n]$, the one-site marginals are uniform, and our assumptions are inspired by a mean-field-type regime away from criticality, such as the high-temperature Curie--Weiss model~\cite{kirsch2019curie, arous1999increasing,vsamaj1988improved}. Notice that our Assumption~\ref{ass:weakly-correlated-target} is weaker than the results shown in Ref.~\cite{kirsch2019curie}. More generally, when some $t_j\neq 0$, the one-site marginals become biased, and the same assumptions instead describe a moderately biased analogue of this regime, where the qubits behave approximately like weakly dependent biased coins.

As will become clear, we suspect these assumptions are only necessary for our proof strategy, and that the initialization strategy is likely to be effective for a broader regime of target distributions.
However, (conveniently) under these assumptions, the residual mismatch term in the data-dependent curvature is controlled, while the model-sensitivity term remains non-negligible. The following theorem thus holds. 

\begin{theoremappendix}[Variance guarantee for data-dependent initialization strategy under approximately factorizable target, formal]
\label{thm:poly-patch-variance-data-dependent-app}
Consider the low-body MMD loss with $\sigma\in\Theta(\sqrt{n})$ (so that $w_A\in\Theta(n^{-|A|})$ for $A\subseteq[n]$) and the IQP circuit. Let $\thv^*$ be a data-dependent initialization such that
\begin{enumerate}
    \item all single-qubit marginals are matched, i.e. $z_j(\thv^*)=t_j$ for all $j\in[n]$,
    \item all two-qubit parameters are initialized to $0$.
\end{enumerate}
Assume the target correlators satisfy Eq.~\eqref{eq:assumption-correlations} and the non-saturation condition Eq.~\eqref{eq:assumption-nonsaturated-marginals-app}.

Then there exists at least one single-qubit parameter index $\alpha$ such that
\begin{equation}
\left|\frac{\partial^2 \mathcal{L}(\thv^*)}{\partial\theta_\alpha^2}\right|
\in \Theta\!\left(\frac{1}{n}\right).
\end{equation}
Consequently, by Theorem~\ref{thm-sup:patch-variance-double-derivative}, there exists a patch half-width
\begin{equation}
r\in\OC\!\left(\frac{1}{\poly(n)}\right)
\end{equation}
such that, for $\thv\sim\thv^*+\mathrm{Unif}([-r,r]^m)$,
\begin{equation}
\Var_{\thv}[\mathcal{L}(\thv)]\in\Omega\!\left(\frac{1}{\poly(n)}\right).
\end{equation}
\end{theoremappendix}

\begin{proof}
 We use the curvature expression derived in the previous subsection, concretely in Eq.~\eqref{eq:double_deribative_w_mismatch_and_modelsense}
\begin{equation}\label{eq:data-dependent-curvature-recalled}
\frac{\partial^2 \mathcal{L}(\thv^*)}{\partial\theta_{\alpha}^2}
=\sum_{{\substack{ A\subseteq [n]\\
    {\al}\in A}}} 8 w_{A}  \left(\underbrace{\left(t_{A}-\prod_{j\in {A}}t_j\right)\prod_{j\in {A}}t_j}_{\rm Mismatch}+\underbrace{(1-t_{\al}^2)\prod_{j\in A\backslash \{\alpha\}}t_j^2}_{\rm Model\; Sensitivity}\right)\;.
\end{equation}
The first term is the residual mismatch contribution, and the second term is the model-sensitivity contribution. Before starting to analyze the curvature, let us recall that for the low-body MMD, we set the bandwidth $\sigma\in\Theta(\sqrt{n})$ in order to have $np_{\sigma}\in\Theta(1)$ i.e. the effective bodyness of the MMD is constant~\cite{rudolph2023trainability}. Therefore, the weights $w_A$ in the low-body MMD scales as
\begin{equation}\label{eq:weight-MMD-scaling-last-proof}
    w_A\in\Theta\left(n^{-|A|}\right)\;. 
\end{equation}
In particular, we have that $p_{\sigma}^k\in\Theta(n^{-k})$ and $(1-p_{\sigma})^k\in\Theta(1)$ for $k\leq n$.

In Eq.~\eqref{eq:data-dependent-curvature-recalled}, notice that we only keep terms $Z_A$ with $\alpha\in A$ since the single qubit gate with parameter $\theta_{\alpha}$ must act non-trivially on qubit $\alpha$ to avoid vanishing derivative. Therefore, the first order contribution ($|A|=1$) in Eq.~\eqref{eq:data-dependent-curvature-recalled} is the term $A=\{\al\}$ which is given by
\begin{equation} \label{eq:first-order-curvature-term-proof-app}
    8w_{\{\al\}}(1-t_{\al}^2)\in\Theta(1/n)\;,
\end{equation}
where we used the non-saturation condition in Eq.~\eqref{eq:assumption-nonsaturated-marginals-app} together with the scaling of the weights in Eq.~\eqref{eq:weight-MMD-scaling-last-proof}. We might expect some negative contributions due to higher-order correlators arising from the mismatch term in Eq.~\eqref{eq:data-dependent-curvature-recalled}. In the following, we will bound these negative contributions using the assumption in Eq.~\eqref{eq:assumption-correlations} in order to ensure negligible mismatch contribution from the higher-order terms. 

 For a fixed order $|A|=K$ with $\alpha\in A$, we bound each summand from below as
\begin{equation}\label{eq:K-order-term-in-MMD-double-derivative}
\left(t_A-\prod_{j\in A}t_j\right)\prod_{j\in A}t_j
+(1-t_\alpha^2)\prod_{j\in A\setminus\{\alpha\}}t_j^2
\ge
\prod_{j\in A\setminus\{\alpha\}}|t_j|
\left(
(1-t_\alpha^2)\prod_{j\in A\setminus\{\alpha\}}|t_j|
-\left(\frac{C}{n}\right)^{K/2}|t_\alpha|
\right)\,,
\end{equation}
where we used Eq.~\eqref{eq:assumption-correlations}. 
In the following, we will first consider the case where $t_{\al}=0$ for at least one $\alpha\in [n]$, then generalizing to the case where $|t_{\alpha}|\geq 0$.  
\paragraph*{\underline{If $t_{\al}=0$},}
 the negative contribution (mismatch) in Eq.~\eqref{eq:K-order-term-in-MMD-double-derivative} vanishes and this trivially leads to a positive contribution (model sensitivity) of each term. Therefore, the first order term ensures that the curvature scales as $\Omega(1/n)$ i.e. 
\begin{equation} \label{eq:curvatue-t_alpha=0-case-proof-last-thm}
    \frac{\partial^2 \mathcal{L}(\thv^*)}{\partial\theta_{\alpha}^2}\geq 8w_{\{\al\}}\in\Theta\left(\frac{1}{n}\right)\;,
\end{equation}
where the inequality is obtained by keeping only the first order term given in  Eq.~\eqref{eq:first-order-curvature-term-proof-app} together with the assumption $t_{\al}=0$. Notice that the choice of $\alpha$ is arbitrary so we can fix $\al=\argmin_j(|t_j|)$ to maximize $1-t_{\al}^2$. In this case all single site marginals are uniform which is equivalent to the unbiased initialization, but the guarantee from Eq.~\eqref{eq:curvatue-t_alpha=0-case-proof-last-thm} is still valid if only a single site $\al$ satisfies $t_{\alpha}=0$. 

\paragraph*{\underline{If $|t_{\al}|\geq 0$,}}the expression in Eq.~\eqref{eq:K-order-term-in-MMD-double-derivative} can be further lower-bounded by setting $\prod_{j\in A\setminus\{\alpha\}}|t_j|=\left(\frac{C}{n}\right)^{K/2}\frac{|t_{\al}|}{2(1-t_{\al}^2)}$. Indeed, the function $f(x)=x(ax-b)$ reaches its extremum at $x=b/(2a)$ which is a minimum for positive $a$ and here we identify $a=1-t_{\al}^2$, $b=\left(\frac{C}{n}\right)^{K/2}|t_\alpha|$ and $x=\prod_{j\in A\setminus\{\alpha\}}|t_j|$.  Therefore, we have

\begin{equation} \label{eq:K-order-term-MMD-double-derivative-general-lowerbound}
    \left(t_A-\prod_{j\in A}t_j\right)\prod_{j\in A}t_j
+(1-t_\alpha^2)\prod_{j\in A\setminus\{\alpha\}}t_j^2
\ge -\left(\frac{C}{n}\right)^{K}\frac{t_\alpha^2}{4(1-t_{\al}^2)}\,.
\end{equation}

Let us substitute this in the full expression of the curvature in Eq.~\eqref{eq:data-dependent-curvature-recalled} for each term with $|A|>1$ (i.e. $A\neq\{\al\}$). Therefore, the first order contribution from Eq.~\eqref{eq:first-order-curvature-term-proof-app} is the only positive one and the contribution of all the remaining terms becomes:
\begin{equation} \label{eq:negative-contribution-lower-bound-assumption1-proof}
       \sum_{{\substack{ A\subseteq [n]\\
    {\al}\in A \\ A\neq\{\al\}}}} 8 w_{A}  \left(\left(t_{A}-\prod_{j\in {A}}t_j\right)\prod_{j\in {A}}t_j+(1-t_{\al}^2)\prod_{j\in A\backslash \{\alpha\}}t_j^2\right)  \geq   -\sum_{{\substack{ A\subseteq [n]\\
    {\al}\in A \\ A\neq\{\al\}}}} 8 w_{A}\left(\frac{C}{n}\right)^{|A|}\frac{t_\alpha^2}{4(1-t_{\al}^2)}\;,
\end{equation}
where the inequality is obtained from Eq.~\eqref{eq:K-order-term-MMD-double-derivative-general-lowerbound}. Therefore, inserting this together with the first order term from Eq.~\eqref{eq:first-order-curvature-term-proof-app} into Eq.~\eqref{eq:data-dependent-curvature-recalled} leads to
\begin{equation}\label{eq:curvature-lower-bound-proof-datadep-full-before-last}
    \frac{\partial^2 \mathcal{L}(\thv^*)}{\partial\theta_{\alpha}^2}\geq 8w_{\{\al\}}(1-t_{\al}^2)-\frac{2t_\alpha^2}{1-t_{\al}^2}\underbrace{\sum_{{\substack{ A\subseteq [n]\\
    {\al}\in A \\ A\neq\{\al\}}}}  w_{A}\left(\frac{C}{n}\right)^{|A|}}_{\mathcal{:=S}}\;,
\end{equation}
where we rearranged the sum in Eq.~\eqref{eq:negative-contribution-lower-bound-assumption1-proof} by factoring out the terms independent of the summation indices $A$.
Now, let us further lower-bound this expression by upper-bounding the magnitude of the negative contribution. In particular, we now upper-bound the following sum:
\begin{equation}\label{eq:sum-in-thm-last-proof-negative-contrib}
   \mathcal{S}:=\sum_{{\substack{ A\subseteq [n]\\
    {\al}\in A \\ A\neq\{\al\}}}}  w_{A}\left(\frac{C}{n}\right)^{|A|}\;.
\end{equation}

From the definition of the weights $w_A=(1-p_{\sigma})^{n-|A|}p_{\sigma}^{|A|}$, we can explicitly compute $\mathcal{S}$ as follows:
\begin{align}
 \mathcal{S}   &=\sum_{{\substack{ A\subseteq [n]\\
    {\al}\in A \\ A\neq\{\al\}}}}  (1-p_{\sigma})^{n-|A|}p_{\sigma}^{|A|}\left(\frac{C}{n}\right)^{|A|}\\
    &=\sum_{k=1}^{n-1}\binom{n-1}{k}(1-p_{\sigma})^{n-1-k}p_{\sigma}^{k+1}\left(\frac{C}{n}\right)^{k+1}\\
    &=\frac{Cp_{\sigma}}{n}\sum_{k=1}^{n-1}\binom{n-1}{k}(1-p_{\sigma})^{n-1-k}\left(\frac{Cp_{\sigma}}{n}\right)^{k}\\ \label{eq:curvature-negative-contribution-proof}
    &=\frac{Cp_{\sigma}}{n}\left(\left(1-p_{\sigma}\left(1-\frac{C}{n}\right)\right)^{n-1}-(1-p_{\sigma})^{n-1}\right)\;,
\end{align}
  The second equality is obtained by noticing that the summands depend only on $|A|$ (the cardinality of $A$). Therefore, we can substitute $k=|A|-1$ such that the number of subsets $A$ containing $\alpha$ with a fixed $k$ is $\binom{n-1}{k}$, i.e. the number of choices of $k$ elements between the remaining $n-1$ ones in $[n]\backslash\{\al\}$, and notice that the sum starts from $k=1$ since the first order term is not considered here ($|A|\neq\{\al\}$). In the third equality, we conveniently rewrite the expression to highlight the binomial series. Finally, we compute the binomial series by subtracting $(1-p_{\sigma})^{n-1}$ corresponding to the term $k=0$.

Now, let us upper-bound the following term:
\begin{equation}\label{eq:scaling-sum-high-order-terms-proof-last-thm}
    \left(1-p_{\sigma}\left(1-\frac{C}{n}\right)\right)^{n-1}-(1-p_{\sigma})^{n-1}\;.
\end{equation}
 We use telescoping series as done previously in Eq.~\eqref{eq:eqtolowerbound_with_kpkm}, to get 
\begin{align}
    \left(1-p_{\sigma}\left(1-\frac{C}{n}\right)\right)^{n-1}-(1-p_{\sigma})^{n-1}&= \frac{p_{\sigma}C}{n}\sum_{k=0}^{n-2}\left(1-p_{\sigma}\left(1-\frac{C}{n}\right)\right)^{n-2-k}(1-p_{\sigma})^{k}\\
    &\leq \frac{p_{\sigma}C}{n}(n-1)\left(1-p_{\sigma}\left(1-\frac{C}{n}\right)\right)^{n-2}\;,\label{eq:tocite_later_upperbound_on_p_sigma_et_al}
\end{align}
where the last inequality is obtained by upper-bounding each summands by the largest one using $1-p_{\sigma}\left(1-\frac{C}{n}\right)\geq 1-p_{\sigma}$.  Notice that $0\leq1-p_{\sigma}\left(1-\frac{C}{n}\right)\leq 1$ for $n\geq C$. So for a sufficiently large $n$, we have 
\begin{equation}
  \left(1-p_{\sigma}\left(1-\frac{C}{n}\right)\right)^{n-2}\leq 1\;.  
\end{equation}

Therefore, the expression of $\mathcal{S}$ in Eq.~\eqref{eq:curvature-negative-contribution-proof} is further upper-bounded using Eq.~\eqref{eq:tocite_later_upperbound_on_p_sigma_et_al} as follows:
\begin{equation}
    \mathcal{S}\leq \left(\frac{p_{\sigma}C}{n}\right)^2(n-1)\in\Theta\left(\frac{1}{n^3}\right)\;,
\end{equation}
where we recall $p_{\sigma}\in\Theta(1/n)$ for low-body MMD. Therefore, $\mathcal{S}\in\OC(1/n^3)$ and recalling its definition in Eq.~\eqref{eq:sum-in-thm-last-proof-negative-contrib} together with the curvature lower-bound in Eq.~\eqref{eq:curvature-lower-bound-proof-datadep-full-before-last} leads to
\begin{equation}
        \frac{\partial^2 \mathcal{L}(\thv^*)}{\partial\theta_{\alpha}^2}\geq \underbrace{8w_{\{\al\}}(1-t_{\al}^2)}_{\Theta(1/n)}-\underbrace{\frac{2t_\alpha^2}{1-t_{\al}^2}}_{\OC(1)}\underbrace{\left(\frac{p_{\sigma}C}{n}\right)^2(n-1)}_{\OC(1/n^3)}\;,
\end{equation}
where we recall the scaling of first order term in Eq.~\eqref{eq:first-order-curvature-term-proof-app} together with  the non-saturation assumption in Eq.~\eqref{eq:assumption-nonsaturated-marginals-app} such that.
Finally, the curvature is dominated by the first order term which scales as $\Theta(1/n)$ (see Eq.~\eqref{eq:first-order-curvature-term-proof-app}) and the negative contribution arising for all higher-order terms whose magnitude scales as $\OC(1/n^3)$  such that for sufficiently large $n$ we have:
\begin{equation}
    \frac{\partial^2 \mathcal{L}(\thv^*)}{\partial\theta_{\alpha}^2}\geq 8w_{\{\al\}}(1-t_{\al}^2)-\OC\left(\frac{1}{n^3}\right)=\Theta\left(\frac{1}{n}\right)-\OC\left(\frac{1}{n^3}\right)\in\Theta\left(\frac{1}{n}\right)\;.
\end{equation}
Therefore, applying Theorem~\ref{thm-sup:patch-variance-double-derivative} completes the proof of the theorem.
\end{proof}

\section{Numerical experiment for global MMD}

In this section, we present further numerical experiments with the \textit{global} MMD loss by fixing a \textit{constant bandwidth} $\sigma\in\OC(1)$. We compare the three initialization strategies studied analytically (identity, unbiased and marginal-matching) together with the data-dependent strategy from Ref.~\cite{recio2025train} as done in Sec.~\ref{sec:numerical-studies} from the main text. For the new data-dependent strategy considered here, we recall that single qubit gates are fixed and correspond to the marginal-matching data-dependent strategy presented in previous section. The two-qubit gates are correlated between each other and depends on the target covariances $C_{ij}=t_{ij}-t_it_j$ as:
\begin{equation}\label{eq:covariance-initialization-2-qubit-gate-appendix}
    \theta_{ij}\sim \frac{C_{ij}}{C_{\max}}{\rm Unif}\left(\left[-s\frac{\pi}{2}, s\frac{\pi}{2}\right]\right)\;,
\end{equation}
where $C_{\max}=\max_{ij}|C_{ij}|$ and $s\in [0,1]$ denotes the initialization scale. For the three cases studied analytically, we have $s=\frac{\pi}{2}r$. 

We usually expect global MMD to be untrainable, since weights are now distributed more uniformly between each correlator. In particular, we have $(1-p_{\sigma})\in\OC(1)$ and $p_{\sigma}\in\OC(1)$ such that
\begin{equation} \label{eq:global-MMD-weight-concentration}
    w_A=(1-p_{\sigma})^{n-|A|}p_{\sigma}^{|A|}\in\OC(\exp(-n))\;.
\end{equation}
Thus, the local terms, for which we have a signal, contribute only a small fraction of the global MMD loss.

\begin{figure}
    \centering
    \includegraphics[width=1\linewidth]{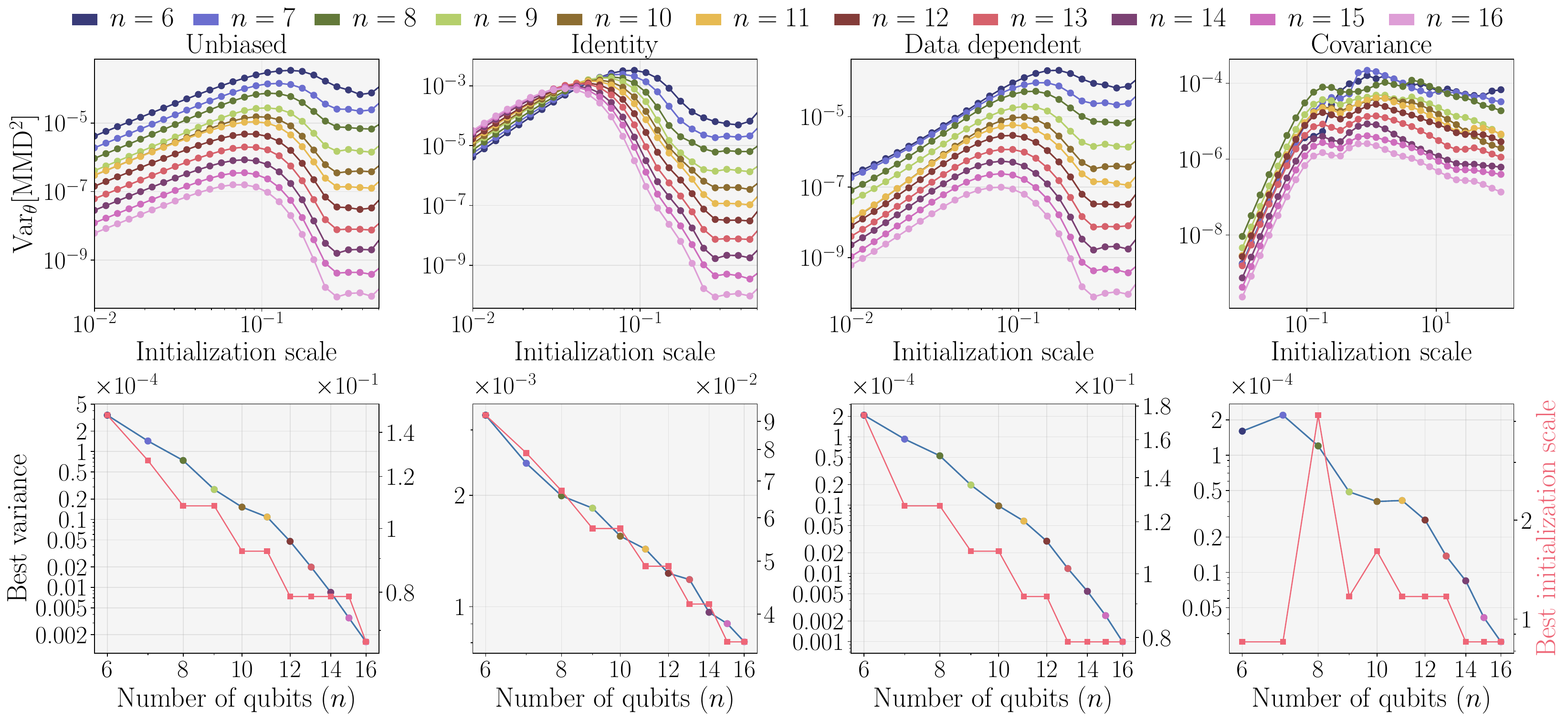}
    \caption{\textbf{Variance of the global MMD estimator versus initialization scale.}
Top row: Variance of the MMD estimator as a function of the initialization scale for four initialization schemes (identity, unbiased, data-dependent, and covariance) using a kernel with bandwidth that scales as $\sigma\in\OC(1)$. Curves correspond to different numbers of qubits $n$ ranging from $n=6$ to $n=16$. For the first three schemes, parameters are initialized as
$\theta_j\sim\theta_j^*+{\rm Unif}[-\frac{\pi}{2}s,\frac{\pi}{2}s]$
where $s$ is the initialization scale. The covariance initialization (right column) follows Eq.~\eqref{eq:covariance-initialization-2-qubit-gate-appendix}, where two-qubit gate parameters are correlated and perturbations are rescaled according to the target distribution covariances. The target distribution is given by a genomic dataset.
Bottom row: Maximum variance over initialization scales (blue) and the initialization scale achieving this maximum (red dashed) as functions of the number of qubits $n$.}
\label{fig:mmd-blobal-variance-patch-scaling}
\end{figure}
However, the results in Fig.~\ref{fig:mmd-blobal-variance-patch-scaling} show that identity initialization yields surprisingly large gradients within restricted patches in terms of scaling with $n$. Indeed, for the unbiased and data-dependent case, we observe a lower gradient around initialization both in terms of magnitude and scaling with the number of qubits (notice that a constant spacing between curves in the top figures correspond to exponential decay with $n$). This can be partially understood from Proposition~\ref{prop:correlator-variance-around-identity}, where we derived an inverse-polynomial variance guarantee for the correlators within those restricted patches. 
Moreover, the specific structure of the distribution at identity center where all correlators are set to $z_A(\vec{0})=1$ leads to a large curvature if the target is far from from this distribution and recalling that this is a local maximum of the MMD loss (see Appendix~\ref{app:data-agnostic-initialization-curvature}).
As a trivial example, let us assume that the target is a uniform distribution such that the curvature of the identity initialization in Eq.~\eqref{eq:MMD-double-derivative-identity-initialisation} becomes
\begin{equation}    \left|\partial^2_{\th_\al}\mathcal{L}(\vec{0})\right| = 8\sum_{\substack{A\subseteq n \\ \al\in A}}w_A=8p_{\sigma}\in\OC(1),
\end{equation}
where here we recall that for $\sigma\in\OC(1)$ we have $p_{\sigma}\in\OC(1)$. In contrast, the curvature of the unbiased initialization center in Eq.~\eqref{eq:mmd-double-derivative-unbiased-app} only contains a single term:
\begin{equation}    \left|\partial^2_{\th_\al}\mathcal{L}(\thv^*)\right| = 8w_{\{\al\}}=8p_{\sigma}(1-p_{\sigma})^{n-1}\in\OC(\exp(-n))\;,
\end{equation}
where we recall the scaling of the weights from Eq.~\eqref{eq:global-MMD-weight-concentration}. Therefore, we can see that the specific structure of the correlators at identity initialization usually lead to non-negligible curvature and thus a non-negligible MMD variance from Theorem~\ref{thm-sup:patch-variance-double-derivative}. 

\end{document}